\documentclass[11pt,reqno]{article}
\usepackage[margin=1in]{geometry}
\usepackage{dsfont, amssymb,amsmath,amscd,latexsym, amsthm, amsxtra,amsfonts}
\usepackage{lineno}
\usepackage[all]{xy}
\usepackage[active]{srcltx}

\usepackage{tikz}
\usepackage[round]{natbib}
\usepackage{bbm}
\usepackage{enumerate}
\usepackage{mathrsfs}
\usepackage{graphicx}
\usepackage{comment}
\usepackage{mathtools}
\usepackage{cases}
 \usepackage{tikz}
\usetikzlibrary{calc,arrows}
\usepackage{verbatim}
\usepackage{graphicx}
\usepackage{subfigure}
\usepackage{color}

\usepackage{tikz}
\usetikzlibrary{calc,arrows}
\usepackage{verbatim}
\usepackage{graphicx}
\usepackage{subfigure}
\usepackage{epstopdf}
\usepackage{verbatim}
\usepackage{graphicx}
\usepackage{subfigure}
\usepackage{color}
\usepackage[affil-it]{authblk}

\newtheorem{theorem}{Theorem}[section]

\newtheorem{lemma}[theorem]{Lemma}

\theoremstyle{definition}
\newtheorem{definition}[theorem]{Definition}

\theoremstyle{definition}

\theoremstyle{definition}
\newtheorem{remark}{Remark}
\theoremstyle{definition}
\newtheorem{assumption}{Assumption}

\newcommand{\mE}{\mathbb{E}}

\renewcommand{\epsilon}{\varepsilon}

\newcommand{\cA}{\mathcal{A}}

\usepackage[pdfstartview=FitH, bookmarksnumbered=true,bookmarksopen=true, colorlinks=true, pdfborder=001, citecolor=blue, linkcolor=blue,urlcolor=blue]{hyperref}
\usepackage{graphics}
\graphicspath{{figures/}}
 \title{Consumption-investment decisions with endogenous reference point and drawdown constraint}
 
 \author{Zongxia Liang$^a$\thanks{Email: \texttt{liangzongxia@mail.tsinghua.edu.cn}}\ \ \ \ \  Xiaodong Luo$^a$\thanks{Corresponding author, email: \texttt{luoxd21@mails.tsinghua.edu.cn}}\ \ \ \ \  Fengyi Yuan$^a$\thanks{Email: \texttt{yfy19@mails.tsinghua.edu.cn}}
 }	
\affil{$^a$Department of Mathematical Sciences, Tsinghua University, China}
\numberwithin{equation}{section}

\begin{document}
\maketitle
\begin{abstract}
We propose a consumption-investment decision model where past consumption peak $h$ plays a crucial role. There are two important consumption levels: the lowest constrained level and a reference level, at which the risk aversion in terms of consumption rate is changed. We solve this stochastic control problem and derive the value function, optimal consumption plan, and optimal investment strategy in semi-explicit forms. We find five important thresholds of wealth, all as functions of $h$, and most of them are nonlinear functions. As can be seen from numerical results and theoretical analysis, this intuitive and simple model has significant economic implications, and there are at least three important predictions: the marginal propensity to consume out of wealth is generally decreasing but can be increasing for intermediate wealth levels, and it jumps inversely proportional to the risk aversion at the reference point; the implied relative risk aversion is roughly a smile in wealth; the welfare of the poor is more vulnerable to wealth shocks than the wealthy. Moreover, locally changing the risk aversion influences the optimal strategies globally, revealing some risk allocation behaviors.
\vskip 5pt\noindent
\end{abstract}
\section{Introduction}
It is intuitive that the historical peak of past consumption has great impact on individual consumption decision. For example, to consume below a certain ratio of the historical peak brings an impulsion to ``reclaim the past glory". If the consumption is forced to further decline to a level that is far below the historical peak, it becomes unbearable and people will try to increase their wealth (by financing, borrowing or selling illiquid assets) at any cost to satisfy consumption at a certain (but low) ratio of past peak. Therefore it is not surprising that there has been literature studying so-called drawdown constraint (\cite{Dybvig1995} and \cite{arun2012}).
\vskip 5pt
We develop a theoretical model based on dynamic portfolio choice theory to incorporate aforementioned psychological insights and study the effects of them on consumption and risky investment decisions. In our model, the process of {\it standard of living} represented by running maximum of consumption,  $h_t=h_0\vee\sup\limits_{0\leq s\leq t}\{c_s\}$, plays a central role. Here $c\triangleq\{c_s, s\geq 0\}$ is the process of consumption rate, and $h_0$ is the {\it inherited} running maximum level, which is usually determined by exogenous factors, say,
family fortunes. To be specific, { the } consumption is {\it constrained to be no less than} $\lambda h_t$, and the preference on consumption is assumed to be of the form:
\begin{equation*}
	U(c,h)=\left\{
	\begin{array}{l}
    \frac{1}{\beta_{1}}\left[1-e^{-\beta_{1}(c-\alpha h)}\right],\ \lambda h\le c<\alpha h,\\
    \frac{1}{\beta_{2}}\left[1-e^{-\beta_{2}(c-\alpha h)}\right],\hskip 10pt \alpha h\le c,
    \end{array}
	\right.
\end{equation*}
 where   $\lambda$ and $\alpha$ with $0<\lambda<\alpha<1$ are exogenous constants, representing two important thresholds of consumption level, and $ \beta_1 $ and $\beta_2 $ with  $\beta_1,\beta_2>0$ are absolute risk aversions.  We allow $\beta_1>\beta_2$, $\beta_1<\beta_2$ or $\beta_1=\beta_2$. 
 \vskip 5pt
 In our model, the utility is produced by the difference between the agent's instantaneous consumption and a reference point $\alpha h$. Besides, we also consider a change of risk aversions on different sides of the reference point. This reflects psychological effects when consumption rises above or falls below the reference point. In economical literature, changes of the individual risk aversions are discussed and empirical evidences have been found. On the one hand, based on both naturally
 occurring data and lab data, it is acknowledged that people become more risk averse when experiencing crisis or fear, see \cite{Cohn2015}, \cite{Guiso2018} and references therein. In this paper, choosing $\beta_1>\beta_2$, we model such crisis as the decline of consumption level. In this case, people are more risk averse in crisis ($c<\alpha h$). On the other hand, as an important part of the Nobel-wining prospect theory, loss aversion is an effect that people become risk seeking when the pay-offs fall below some certain reference point. When incorporating such extreme gambling behavior into consumption decision, somewhat extreme optimal decisions are derived: people never consume between 0 (or lowest constrained level) and the reference point, see \cite{VanBilsen2020} and \cite{Li2021} for examples. One possible reason is that consumption falling below certain reference level is not generally treated as loss, but rather as bad luck or temporary crisis. To incorporate gambling effect into consumption decision, we can conveniently choose $\beta_1<\beta_2$ in our model. In this case, people are willing to take more risk when their consumption is in danger ($c<\alpha h$). Our main interests are investigating consumption and portfolio behaviors under the aforementioned preference change at the reference point. The reference point itself, however, can be further generalized from the particular choice $\alpha h$. We just list main results of this generalization in Appendix \ref{genpre} because solution techniques we use are still applicable.
\vskip 5pt  
 It turns out that our simple risk-aversion-changing preference leads to consumption and investment decisions with significant economic implications.

\vskip 5pt
\begin{figure}
	\centering
	\includegraphics[width=\linewidth]{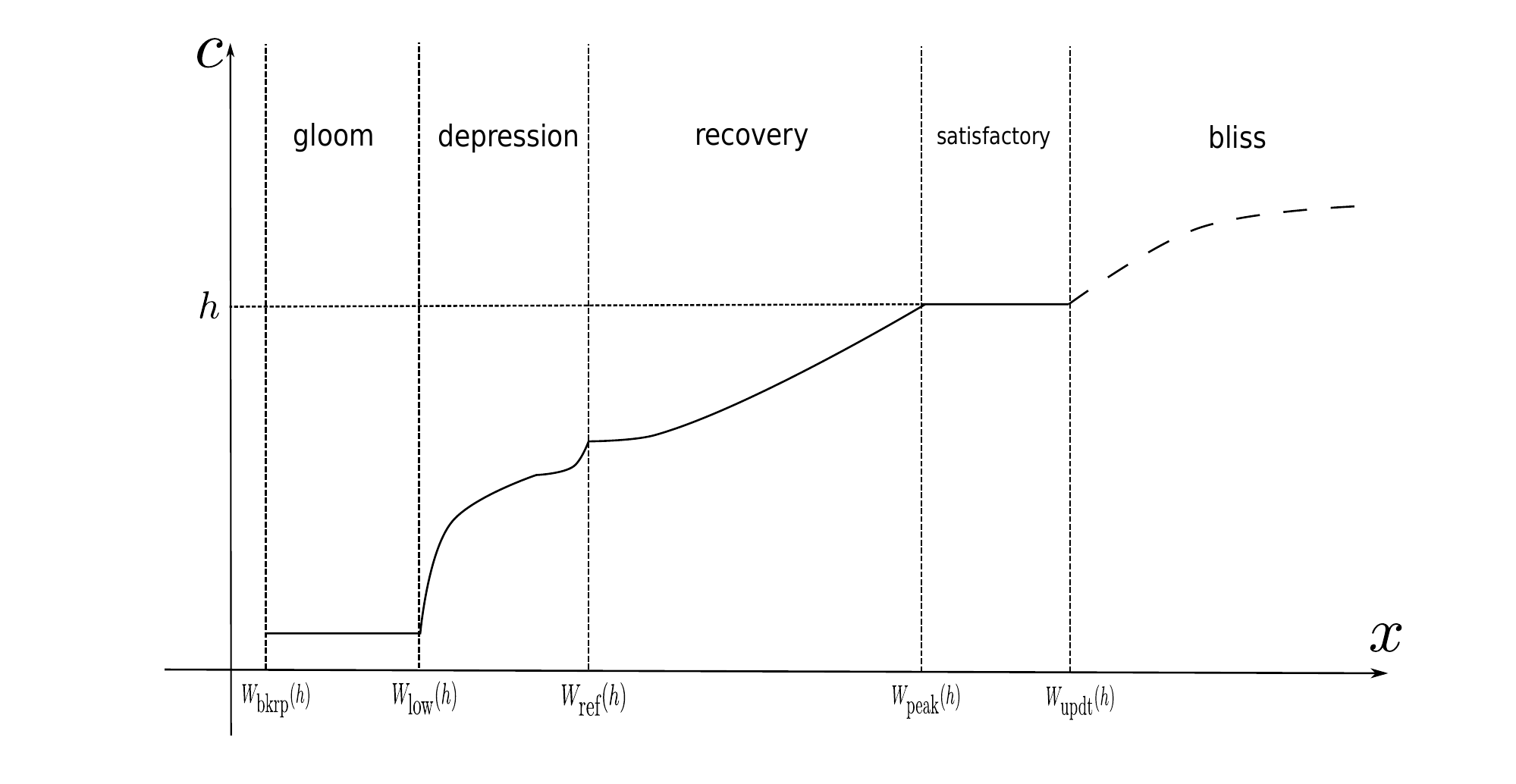}\\
	\includegraphics[width=\linewidth]{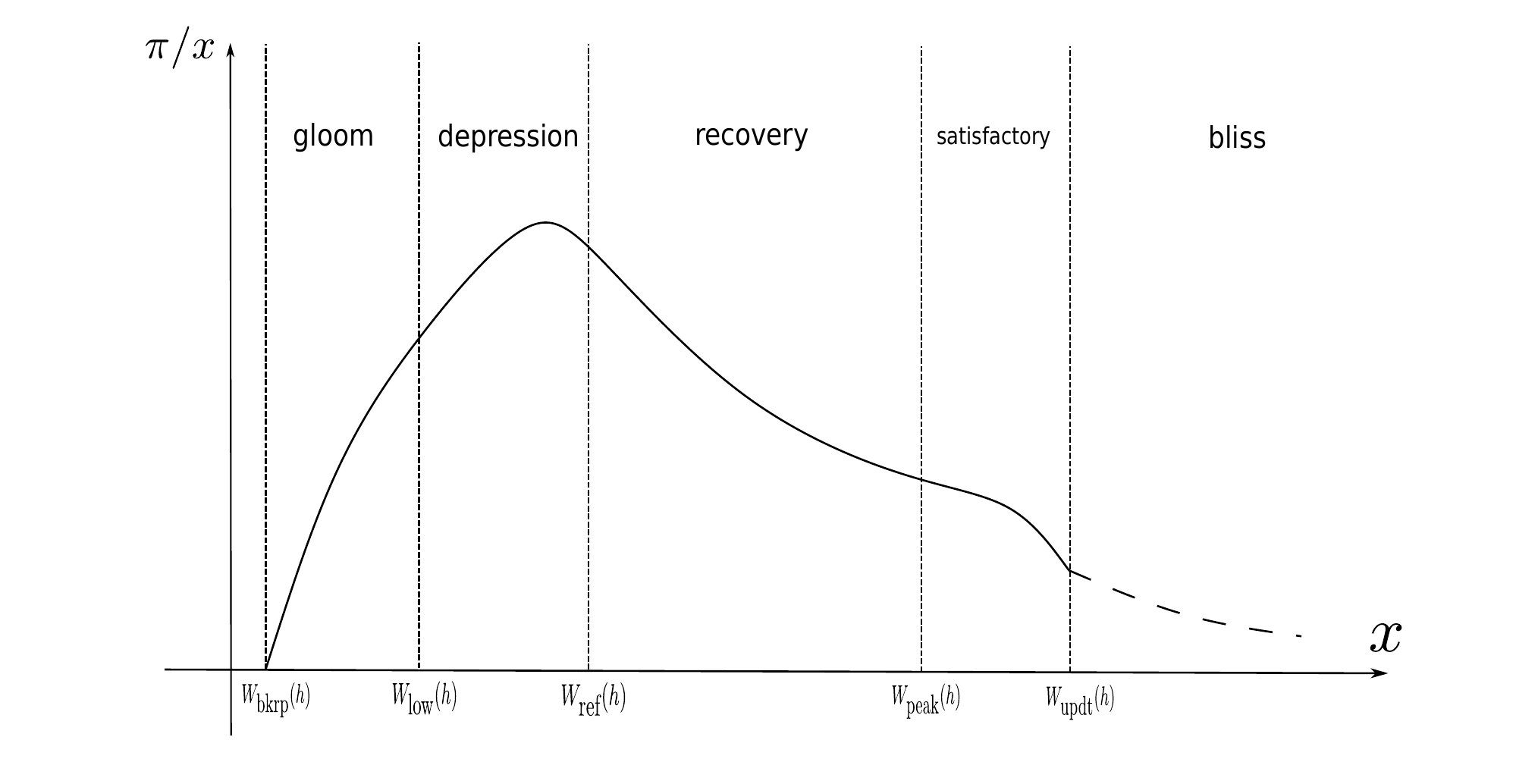}
	\caption{Optimal consumption (above) and optimal risky investment proportion (below) when fixing an $h$.}
	\label{appendixfigure}
\end{figure}

We find five important thresholds of wealth, all depending on $h$ and  denoted respectively by $W_{\rm bkrp}(h),W_{\rm low}(h),W_{\rm ref}(h),W_{\rm peak}(h)$ and $W_{\rm updt}(h)$, which are crucial to describe the derived consumption and investment decisions. For those with $x\geq W_{\rm updt}(h)$, i.e., the very wealthy ones, the best choice of consumption is to update the consumption peak in order to move to or maintain on the line $x=W_{\rm updt}(h)$. If $x<W_{\rm bkrp}(h)$, the wealth of the agent is insufficient for keeping a consumption rate $c\geq \lambda h$ considering his standard of living in the past and leads to bankruptcy, hence is not allowed in our discussion. The above results indicate that only those $x\in [W_{\rm bkrp}(h),W_{\rm updt}(h)]$ are of special interests, and they are said to be in {\it the effective region}. For the wealthiest and poorest people in effective region, the optimal consumption rate is clear: the wealthiest ($W_{\rm peak}(h)<x\leq W_{\rm updt}(h)$, named as satisfactory region) revisit the historical peak $c=h$, and the poorest ($W_{\rm bkrp}(h)\leq x\leq W_{\rm low}(h)$, named as gloom region) choose to consume at the lowest level $c=\lambda h$. As for the region of intermediate wealth $W_{\rm low}(h)<x\leq W_{\rm peak}(h)$, by the reference threshold $W_{\rm ref}(h)$, we divide it into two sub-regions: depression and recovery. We document two important phenomena. {\it First}, the (instantaneous) marginal propensity to consume (MPC) out of wealth $\frac{\partial c^*}{\partial x}$ is decreasing for people in lower part of depression-recovery region (most or all part of depression, sometimes plus part of recovery), which is consistent with empirical studies  (e.g. \cite{arrondel2015wealth}), and we predict that it is increasing for wealthier people, i.e., those in the rest of depression-recovery region. However, if the wealth level is $x=W_{\rm updt}(h)$, called bliss curve by us, the MPC out of wealth is then decreasing again. {\it Second}, the MPC out of wealth jumps by a fixed proportion $\beta_1/\beta_2$ at the threshold $W_{\rm ref}(h)$, at which the risk aversion is changed. This fact reveals one of salient features of our model and indicates that, generally, when getting wealthier, people have lower or higher MPC based on its particular type of risk attitude (gambling type or stop-loss type). See Figure \ref{appendixfigure} for a graphic illustration of optimal consumption decision.
\vskip 5pt
We also obtain the optimal investment strategy $\pi^*$, representing the {\it amount of wealth} invested in risky assets, and optimal investment proportion $\pi^*/x$ is obtained as a by-product. Recall that in classical Merton's problem, the optimal investment proportion is $\frac{\mu-r}{\sigma^2\gamma}$ if relative risk aversion is $\gamma$. Therefore, the inverse of $\pi^*/x$ can be treated equivalently as relative risk aversion, which we call implied relative risk aversion (see also \cite{Jeon2020}). Keeping this in mind, we find decreasing relative risk aversion (DRRA) and increasing relative risk aversion (IRRA) are both possible even for a single agent and the implied relative risk aversion is roughly a $U$-shaped curve (smile) in variable $x$ with trough around $W_{\rm ref}(h)$. Equivalently, risky investment proportion is a hump with peak around $W_{\rm ref}(h)$. We further predict that this effect is more pronounced for poorer people. In economic literature, there has been long standing debate on how relative risk aversion varies in wealth distribution, and evidence for both DRRA and IRRA are found (\cite{RES82} and \cite{Bellante2004}). We provide an explanation for this: the RRA can decrease in wealth because there is impulsion to get back to a higher consumption above $\alpha h$ if $x$ is not so large; the RRA can also increase in $x$ if $x$ has been enough for maintaining a satisfactory standard of living. As for the portfolio choice itself, although it is widely admitted that for {\it macro} data, the wealthier people tends to invest more proportion of their wealth in risky assets, there is no solid agreement on the same question in {\it micro} aspect. That is to say, what happens to risky investment if the wealth increases for {\it a given} household? Our model predicts that it is very likely that people proportionally reduce risky investment if their wealth grows, which is consistent with empirical studies (\cite{Brunnermeier2008} and \cite{Paya2016}) or other possible alternative models (\cite{Wachter2010}). Our model also allows opposite result, provided that the wealth is not enough, which provides explanations for co-existence of both phenomena in some literature such as \cite{Brunnermeier2008}.  See Figure \ref{appendixfigure} for a graphic illustration of optimal risky assets allocation.
\vskip 5pt
 The literature most closely related to the present paper is \cite{deng2021}. In the aspect of reference point, we adopt the setting of \cite{deng2021} and rely on their solving techniques as well as other celebrated tools such as dynamic programming, dual transformation and region-wise solving method. Our choice of the reference point is for simplicity. It turns out that the set of solving techniques we have used, inspired by \cite{deng2021}, is widely applicable to other forms of reference point, see Appendix \ref{genpre}. In addition to the reference point itself, our main interests include investigating consumption and portfolio behaviors under the preference change at the reference point. There are several distinctive features of our model, serving as complement to the one proposed and solved in \cite{deng2021}. First of all, the preference change leads to an upward or downward turn of optimal consumption at the threshold $W_{\rm ref}(h)$, instead of a relatively smooth curve in \cite{deng2021}. Moreover, we emphasize that the influences of preference change on optimal consumption choice are global. For example, adjusting the risk aversion from $\beta_2$ to $\beta_1(<\beta_2)$ in the region $x<W_{\rm ref}(h)$ even changes the consumption behavior in the region $x>W_{\rm ref}(h)$, revealing a kind of strategical risk allocation among different states of economics. In terms of risky investment proportion, the preference change results in significant increase or decrease with a wide range of wealth level. The preference change together with drawdown constraint leads to a totally different optimal investment proportion curve comparing with \cite{deng2021}. Another advantage of the present model is that considering preference change highlights the importance of the wealth threshold $W_{\rm ref}(h)$ that has been more or less neglected before. For example, the risky investment proportion attains its maximum around $W_{\rm ref}(h)$, and the value function has very different sensitivity to wealth shock on different sides of the curve $x=W_{\rm ref}(h)$. The detailed differences between our paper and \cite{deng2021}, as well as other implications of our model, will become clear in Section \ref{numana}.
\vskip 5pt
To sum up, the model studied in this paper has considered and combined three aspects of the economical and psychological background of the consumption and investment problem: (1) the running maximum of past consumption is taken as the habit formation process; (2) a drawdown constraint is imposed on consumption; (3) risk aversion is changing in the model and a reference point is added.  We have applied the solution method inspired by \cite{deng2021} as well as other celebrated tools. We have used dynamic programming principle to derive the HJB equation of the problem, solved the optimal consumption and portfolio strategy thanks to the duality method and numerically analyzed the solution and sensitivity of thresholds to certain parameters. From this simple and intuitive model, we have found several interesting economic implications such as MPC jump and RRA smile.
\vskip 5pt
The rest of the paper is organized as follows: Section \ref{modfor} is devoted to mathematically formulating the optimal consumption and investment problem focused in this paper. We deduce the HJB equation and obtain the feedback form of solution in dual form in Section \ref{dersol}. In Section \ref{veropt}, the verification theorem is established and the optimal strategy is obtained in primal form. Numerical analysis with fixed parameters are in Section \ref{numana}, while in Section \ref{sean} we present sensitivity analysis. Section \ref{concl} gives a brief conclusion. Technical proofs and some generalizations are presented in Appendices.
\vskip 5pt
{\bf Related literature.} Our model is based on dynamic consumption-investment decision model, whose classical form dates back to \cite{merton1969}. A wide range of literature extends this problem by taking habit formation into account, see \cite{pollak1970}, \cite{detemple1992} and \cite{chapman1998} for instance. The habit formation is usually modelled through habit formation process and habit formation preference. Habit formation process is a process whose value at time $t$ is determined by the consumption process up to time $t$. Habit formation preference suggests that the decision maker's utility is depending on both the consumption process and the habit formation process, which sketches how consumption habit affects current consumption behavior.
\vskip 5pt
One may just set the habit formation process as the average of the past consumption process. A more reasonable and flexible choice of the habit formation process is the so-called linear habit formation, i.e., the weighted average of the past consumption where more weight is placed on the recent consumption. Such form of habit formation process has dominated the research in habit formation setting since early literature such as \cite{ryder1973}. 
Recently, taking the running maximum process of past consumption as the habit formation process has opened another stream of research in the study of habit formation. The running maximum process is non-decreasing and only updates if the consumption level exceeds the historical running maximum, which brings about mathematical challenges because of the connection to singular control theory. We follow the running maximum habit formation model recently studied by \cite{guasoni2020} and \cite{deng2021}. However, our model takes into account more factors that may influence the decision maker's consumption and portfolio choice, including the constraint on consumption and the change of risk aversion.
\vskip 5pt
How utility depends on the consumption $c$ and habit $h$ is another topic in habit formation. The habit formation preference usually refers to the case that the utility is a function of $c-h$, which suggests that the habit has an inverse impact on the utility. A wide range of literature, \cite{chapman1998} for instance, adopts the above habit formation preference setting. A more flexible model, in \cite{deng2021}, requires the utility to depend on $c-\alpha h$ where $\alpha$ may vary in $[0,1]$. \cite{guasoni2020} uses a different approach by setting the utility to be a function of $\frac{c}{h^{\alpha}}$ where $\alpha\in(0,1)$. All the above settings insure the inverse impact of the habit. How consumption is constrained, on the other hand, is also important when studying habit-related model. Usually, the constraint imposes a lower bound on consumption. In the extreme case, the consumption is prohibited from falling below the habit (\cite{muraviev2011} and \cite{yu2015} for linear habit, \cite{Dybvig1995} and \cite{Jeon2020} for running maximum), which is termed as addictive habit formation. In other literature it is required instead that the consumption is at least a fraction of the habit, such as \cite{arun2012} and \cite{bahman2021}. In this paper, we also impose such drawdown constraint, in which we require $c\ge\lambda h$, where $\lambda\in(0,\alpha)$. We assume $\lambda<\alpha$ to ensure the validity of the reference point $\alpha h$ where risk aversion changes. For $\lambda=0$, our model reduces to a model with no drawdown constraint; for $\lambda=\alpha$, our model reduces to a model with no risk aversion change. It is worth mentioning that both \cite{arun2012} and \cite{bahman2021} obtain a threshold of wealth/habit ratio below which the agent chooses to consume the lowest. We derive similar phenomenon, but with a rather complicated threshold curve ($x=W_{\rm low}(h)$) of wealth-habit pair, instead of a simple ray. 
\vskip 5pt
The paradigm of our setting is utility with reference points, which has been widely studied in different economical problems, such as \cite{jin2008} and \cite{he2019}. Another common model under this paradigm is the S-shaped utility developed by D. Kahneman and A. Tversky.  \cite{Li2021} has studied such S-shaped utility in the context of running maximum habit formation. However, such non-concave utility results in an extreme optimal strategy where the decision maker never consumes between $0$ and the reference point. Instead, we analyze a utility with risk aversion change but in a concave form. As a result, we obtain a more reasonable optimal strategy where the optimal consumption varies from the lower bound $\lambda h$ to the running maximum $h$ in continuous values. Another related work is \cite{VanBilsen2020}, which permits the agent to be risk averse in loss domain (setting $\gamma_L>1$ therein). This setting is consistent with ours, but they model risk aversion change in terms of relative risk aversions.
\vskip 15pt

\section{Model Formulation}
\label{modfor}
	The financial market consists of one risk-free asset and one risky asset in our model. The risk-free asset $\{S^{0}_{t},  t\ge 0\}$ satisfies the dynamic
	
	\begin{equation*}
	{dS^{0}_{t}}={S^{0}_{t}}rdt,
	\end{equation*}
	where $r>0$ is the constant interest rate. The risky asset $\{S^{1}_{t},  t\ge 0\}$ satisfies
	
	\begin{equation*}
	{dS^{1}_{t}}={S^{1}_{t}}\left[\mu dt+\sigma dB_{t}\right],
	\end{equation*}
	where $\mu\ge r$ is the expected return and $\sigma>0$ is the volatility. The Brownian motion $B$ is defined on the filtered probability space $(\Omega,\mathcal{F},\{\mathcal{F}_{t}\}_{t\ge 0},\mathbb{P})$  satisfying the usual conditions.
\vskip 5pt	
The decision maker determines his dynamic spending rate $c\triangleq\{c_{t},  t\ge 0\}$ and the dynamic amount of his wealth invested in the risky asset denoted by $\pi\triangleq\{\pi_{t},  t\ge 0\}$. Let $x_{0}$ be the investor's initial wealth, then his wealth process $\{X_{t},  t\ge 0\}$  follows the following stochastic differential equation (abbr. SDE):
	\begin{equation}\label{wealproc}
	\left\{
	\begin{array}{l}
	dX_{t}=rX_{t}dt+\pi_{t}(\mu-r)dt+\pi_{t}\sigma dB_{t}-c_{t}dt,\\
	X_{0}=x_{0}.
	\end{array}
	\right.
	\end{equation}
\vskip 5pt
Given the consumption process $\{c_{t}, t\ge 0\}$, the running maximum process of past consumption is defined as  $h=\{h_{t}, t\ge 0\}$:
	\begin{equation}
	\label{maxproc}
	    h_{t}=h_{0}\vee \sup\limits_{s\le t}c_{s},\ t>0
	\end{equation}
	and $h_{0}$ is the inherited running maximum level. It is naturally required that the wealth level is always non-negative under any admissible strategy to avoid bankruptcy. Thus we now introduce the admissible strategies.
	\begin{definition}
		Process $(c,\pi)=\big\{(c_{t},\pi_{t}),\ t\ge 0\big\}$ is an admissible strategy of Problem (\ref{problem}) if it is adapted w.r.t $\{\mathcal{F}_{t}\}_{t\ge 0}$ with

		\begin{align*}
		&\int_{0}^{t}c_{s}ds<\infty,\ a.s.,\\
		&\int_{0}^{t}\pi_{s}^{2}ds<\infty,\ a.s.,\\
		&c_{t}\ge \lambda h_{t},\ a.s., \\
		&X_{t}>0,\ a.s.,
		\end{align*}
		
		for all $t\ge 0$,  where  $\{h_{t},  t\ge 0\}$ is the corresponding running maximum process given by (\ref{maxproc}) and $\{X_{t}, t\ge 0\}$ is the corresponding wealth process governed by SDE  (\ref{wealproc}). We denote by $\cA$ the set of admissible strategies.
	\end{definition}
		The goal of the decision maker is to maximize his expected total discounted utility on  infinite planning horizon $[0, \infty )$, which can be formulated as the optimization problem:
	\begin{equation}
	\label{problem}
	\sup\limits_{(c,\pi)\in\mathcal{A}}\mathbb{E}_{x_{0},h_{0}}\int_{0}^{\infty}e^{-\gamma t}U(c_{t},h_{t})dt,
	\end{equation}
	where $\gamma>0$ is the discount factor. The utility function $U(c,h)$ depends on both the consumption $c$ and the running maximum $h$. Precisely,
	\begin{equation*}
	U(c,h)=\left\{
	\begin{array}{l}
    \frac{1}{\beta_{1}}\left[1-e^{-\beta_{1}(c-\alpha h)}\right],\ \lambda h\le c<\alpha h,\\
    \frac{1}{\beta_{2}}\left[1-e^{-\beta_{2}(c-\alpha h)}\right],\ \alpha h\le c\le h,
    \end{array}
	\right.
	\end{equation*}
where $0\le\lambda\le\alpha<1$, $\beta_{1},\beta_{2}>0$. $\alpha h$ is the reference point where risk aversion increases as the consumption exceeds it. The utility is segmented with both segments taking the form of constant absolute risk aversion (abbr.CARA) utility. The absolute risk aversion above reference $\alpha h$ is $\beta_{2}$, while the absolute risk aversion below reference is $\beta_{1}$. The utility is concave and the marginal utilities at the two sides of the reference point $\alpha h$  are equal and finite.

\begin{remark}
When $\lambda=0$, our model reduces to the model without drawdown constraint. When $\alpha=\lambda$ or $\beta_{1}=\beta_{2}$, our model reduces to the model without reference $\alpha h$ and the corresponding risk aversion change. However, our model does not reduces to the model without reference $\alpha h$ if $\alpha$ approaches $1$ from below. The reason lies in the fact that the updating of running maximum is under the leveled-up risk aversion coefficient $\beta_{2}$ when $\alpha<1$ but under risk aversion coefficient $\beta_{1}$ when $\alpha=1$ (the case $\alpha=1$ actually reduces to the model without reference $\alpha h$). In other words, the model is somewhat ``not continuous" at $\alpha=1$. Hence, we just consider the case $\alpha<1$. Another two limiting cases $\beta_{1}\rightarrow 0$ and $\beta_{2}\rightarrow 0$ are discussed in Subsection \ref{lc}.
\end{remark}
\begin{remark}
     Alternatively, we can consider more general reference point instead of $\alpha h$. It turns out that the solution techniques we used are still applicable. Base on these techniques and tools, we give the results with an alternative reference in Appendix \ref{genpre}, which is more general and reduces to the current model in a special case. 
\end{remark}

\vskip 5pt

For convenience, we only deal with the case that $r=\gamma$ in this paper. For the general case, all the results are paralleled to the case $r=\gamma$ but more complicated. Interested readers can refer to Appendix \ref{genpara} for more details.
	\vskip 5pt
	At the end of this section, we provide a lemma which gives a necessary condition for an admissible strategy, which will be used later in Section \ref{dersol}.
\vskip 5pt
\begin{lemma}\label{Lemma1}
For any $(c,\pi)\in\mathcal{A}$ (if $\mathcal{A}\neq\emptyset$), the corresponding  wealth process $\{X_{t},t\ge 0\}$ and running maximum process $\{h_{t},t\ge 0\}$ must satisfy $X_{t}\ge \frac{\lambda}{r}h_{t},\ a.s.$ for $\forall t\ge 0$.
\end{lemma}
\begin{proof}
Suppose that $(c,\pi)\in\mathcal{A}$ and there exists $t_{0}$ such that $\mathbb{P}(X_{t_{0}}<\frac{\lambda}{r}h_{t_{0}})>0$. Then there exists $\epsilon>0$ such that $\mathbb{P}(X_{t_{0}}<(1-\epsilon)\frac{\lambda}{r}h_{t_{0}})>0$. We only need to show that even for strategies with lowest consumption, i.e., $c_{t}=\lambda h_{t_{0}},\ \forall t\ge t_{0}$, there exists $t_{1}>t_{0}$ such that $\mathbb{P}(X_{t_{1}}<0)>0$. For strategies with $c_{t}=\lambda h_{t_{0}},\ \forall t\ge t_{0}$, the dynamic of wealth becomes

\begin{equation*}
    dX_{t}=rX_{t}dt+\pi_{t}(\mu-r)dt+\pi_{t}\sigma dB_{t}-\lambda h_{t_{0}}dt,\ \forall t\ge t_{0}.
\end{equation*}
Solving this SDE, we obtain

\begin{equation}
\label{xtT}
    X_{t}=e^{r(t-t_{0})}X_{t_{0}}+\frac{\lambda h_{t_{0}}}{r}\left[1-e^{r(t-t_{0})}\right]+(\mu-r)e^{rt}\int_{t_{0}}^{t}e^{-ru}\pi_{u}du+\sigma e^{rt}\int_{t_{0}}^{t}e^{-ru}\pi_{u}dB_{u},\ \forall t\ge t_{0}.
\end{equation}
For $t=t_{0}+\frac{1}{r}\ln(\frac{1}{\epsilon})$, we have

\begin{align*}
    &\mathbb{P}\left(e^{r(t-t_{0})}X_{t_{0}}+\frac{\lambda h_{t_{0}}}{r}\left[1-e^{r(t-t_{0})}\right]<0\right)\\
    =&\mathbb{P}\left(X_{t_{0}}<(1-\epsilon)\frac{\lambda}{r}h_{t_{0}}\right)\\
    >&0.
\end{align*}
To handle the sum of last two terms, we introduce the probability measure $\tilde{\mathbb{P}}$ by

\begin{equation*}
    \frac{d\tilde{\mathbb{P}}}{d\mathbb{P}}\bigg|_{\mathcal{F}_{t}}:=e^{-\frac{\mu-r}{\sigma}B_{t}-\frac{(\mu-r)^{2}}{2\sigma^{2}}t},\ \forall t\ge 0.
\end{equation*}
Then $\{\tilde{B}_{t}:=B_{t}+\frac{\mu-r}{\sigma}t,t\ge 0\}$ is standard Brownian motion under $\tilde{\mathbb{P}}$ and the sum of last two terms equals $\sigma e^{rt}\int_{t_{0}}^{t}e^{-ru}\pi_{u}d\tilde{B}_{u}$, which has zero expectation under $\tilde{\mathbb{P}}$. Hence

\begin{equation*}
    \tilde{\mathbb{P}}\Big(\sigma e^{rt}\int_{t_{0}}^{t}e^{-ru}\pi_{u}d\tilde{B}_{u}\le 0\Big)>0,\ \forall t\ge 0,
\end{equation*}
then 

\begin{align*}
    &\mathbb{P}\left((\mu-r)e^{rt}\int_{t_{0}}^{t}e^{-ru}\pi_{u}du+\sigma e^{rt}\int_{t_{0}}^{t}e^{-ru}\pi_{u}dB_{u}\right)\\
    =&\mathbb{P}\Big(\sigma e^{rt}\int_{t_{0}}^{t}e^{-ru}\pi_{u}d\tilde{B}_{u}\le 0\Big)\\
    >&0,\ \forall t\ge 0.
\end{align*}
Since $(c,\pi)\in\mathcal{A}$, we deduce that $X_{t_{0}}$ and $h_{t_{0}}$ are $\mathcal{F}_{t_{0}}$ measurable. Then the sum of first two terms and the sum of last two terms of (\ref{xtT}) are independent. Therefore, choosing $t_{1}=t_{0}+\frac{1}{r}\ln(\frac{1}{\epsilon})$, we have $\mathbb{P}(X_{t_{1}}<0)>0$ and the proof is complete.
\end{proof}
	
\vskip 15pt
\section{Derivation of the Solution in Dual Form}
\label{dersol}
In this section, we apply the martingale optimality principle to derive the HJB equation of Problem (\ref{problem}) and use the duality method to obtain the solution in dual form.

To begin with, the value function of Problem (\ref{problem}) is denoted by

	\begin{align*}
	V(x_{0},h_{0})\triangleq\sup\limits_{(c,\pi)\in\mathcal{A}}\mathbb{E}_{x_{0},h_{0}}\int_{0}^{\infty}e^{-\gamma t}U(c_{t},h_{t})dt.
	\end{align*}
	\begin{definition}
		An optimal strategy $(c^*,\pi^*)$ of Problem (\ref{problem}) is an admissible strategy which satisfies
		
		\begin{equation*}
		\mathbb{E}_{x_{0},h_{0}}\int_{0}^{\infty}e^{-\gamma t}U(c^{*}_{t},h^{*}_{t})dt=V(x_{0},h_{0}).
		\end{equation*}
	\end{definition}
	The martingale optimality principle shows that the process $\{\Gamma_{t}, t\ge 0\}$
	
	\begin{equation*}
	\Gamma_{t}\triangleq e^{-\gamma t}V(X_{t},h_{t})+\int_{0}^{t}e^{-\gamma s}U(c_{s},h_{s})ds
	\end{equation*}
	is a local supermartingale for all admissible $(c,\pi)$ and is a local martingale for the optimal $(c^{*},\pi^{*})$. If the value function is smooth enough, applying the It\^{o}'s rule to $\{\Gamma_{t}, t\ge 0\}$, we derive the HJB equation of Problem (\ref{problem}) as follows\footnote{For notational simplicity, we write $x,h,c,\pi$ instead of $X_{t},h_{t},c_{t},\pi_{t}$ in (\ref{HJB}) and (\ref{dHJB}). $\frac{dh}{dt}$ in the second line refers to $\frac{dh_{t}}{dt}$, which is the derivative of $h_t$ in the sense of distribution. Heuristically, $\frac{dh}{dt}\neq0$ means that the process $\{h_{t}\}$ strictly increases at the instant $t$.}:
	\begin{equation}\label{HJB}
	\!\!\!\left\{
	\begin{aligned}
	&\sup_{c\in[0,h],\pi\in\mathbb{R}}\left\{-\gamma V(x,h)+V_{x}(x,h)\big(rx+\pi(\mu-r)\!-\!c\big)\!+\!\frac{1}{2}V_{xx}(x,h)\sigma^{2}\pi^{2}\!+\!U(c,h)\right\}=0,\\
	&V_{h}(x,h)=0\ {\rm on} \ (x,h) \  s.t.\  \frac{dh}{dt}\neq0,
	\end{aligned}
	\right.
	\end{equation}
and the optimal feedback form of $\pi$ is

\begin{equation*}
    \pi^{*}_{\rm primal}(x,h)=-\frac{\mu-r}{\sigma^{2}}\frac{V_{x}(x,h)}{V_{xx}(x,h)}.
\end{equation*}
 As $\frac{x}{h}\rightarrow\frac{\lambda}{r}$, the optimal investment should decline to zero to avoid bankruptcy,  we have
 
\begin{equation}\label{boundcon1}
    \lim\limits_{\frac{x}{h}\rightarrow\frac{\lambda}{r}^{+}}\frac{V_{x}(x,h)}{V_{xx}(x,h)}=0.
\end{equation}
 To solve HJB equation (\ref{HJB}) with the  boundary condition (\ref {boundcon1}) based on the duality method, we introduce the conjugate of the value function as follows:
 
	\begin{equation*}
	\tilde{V}(y,h)=\sup_{x\ge 0}\{V(x,h)-xy\},\ \ y>0.
	\end{equation*}
	Then we have the duality transform
	\begin{equation}
	\label{tranx}
	\left\{
	\begin{array}{l}
	x=-\tilde{V}_{y}(y,h),\\
	V(x,h)=\tilde{V}(y,h)-y\tilde{V}_{y}(y,h),\\
	V_{x}(x,h)=y,\\
	V_{xx}(x,h)=-\frac{1}{\tilde{V}_{yy}(y,h)},\\
	V_{h}(x,h)=\tilde{V}_{h}(y,h),
	\end{array}
	\right.
	\end{equation}
	and $(\ref{HJB})$ is rewritten in dual form as follows:
	\begin{equation}\label{dHJB}
	\!\!\!\!\!\!\!\!\left\{
	\begin{aligned}
	&\sup_{c\in[0,h],\pi\in\mathbb{R}}\Big\{-\gamma \left[\tilde{V}(y,h)-y\tilde{V}_{y}(y,h)\right]+y\left[-r\tilde{V}_{y}(y,h)+\pi(\mu-r)-c\right]-\frac{1}{2\tilde{V}_{yy}(y,h)}\sigma^{2}\pi^{2}\\&\quad\quad\quad\quad+U(c,h)\Big\}=0,\\
	&\tilde{V}_{h}(y,h)=0\, {\rm on} \,(y,h)\, \ s.t.\  \, \frac{dh}{dt}\neq 0.
	\end{aligned}
	\right.
	\end{equation}
	For (\ref{dHJB}), the optimal feedback form $\pi^{*}=\pi^{*}(y,h)$ is
	\begin{equation}
	\label{pi}
	\pi^{*}(y,h)=\frac{(\mu-r)y\tilde{V}_{yy}(y,h)}{\sigma^{2}}.
	\end{equation}
	And the optimal feedback form $c^{*}=c^{*}(y,h)$ maximizing $\bar{U}(c)\triangleq U(c,h)-cy$ on $[\lambda h,h]$ is
	\begin{equation}
	\label{c}
	c^{*}(y,h)=\left\{
	\begin{array}{l}
	\lambda h,\hskip 2.4cm e^{(\alpha-\lambda)\beta_{1}h}\le y,\\
	-\frac{1}{\beta_{1}}\ln(y)+\alpha h,\ \ 1\le y<e^{(\alpha-\lambda)\beta_{1}h},\\
	-\frac{1}{\beta_{2}}\ln(y)+\alpha h,\ \  e^{-(1-\alpha)\beta_{2}h}\le y<1,\\
	h,\ \ \hskip 2.4cm 0\le y<e^{-(1-\alpha)\beta_{2}h}.
	\end{array}
	\right.
	\end{equation}

	For the region where $c^{*}(y,h)=h$, three sub cases need to be distinguished in order to apply the second equation of (\ref{dHJB}) to solve the HJB equation. The first case is that the current consumption just reaches the past consumption peak but does not update it; the second case is that the current consumption reaches the past consumption peak and updates it; the last case is that the current consumption exceeds the past consumption peak and forces the running maximum process to jump. The last case can only happen at $t=0$ where the inherited running maximum level is lower but the initial wealth is abundant, which switches to the former two cases for $t>0$. Hence we only need to consider the first two cases for $t>0$. The second equation of (\ref{dHJB}) refers to the second case where the running maximum $h$ is updated, which instructs us to further separate the region according to the above different sub cases. Specifically, the running maximum is updated if and only if $\mathop{\rm arg max}\limits_{c}\{U(c,c)-cy\}\ge h$. Moreover, the running maximum $h$ jumps if strict inequality holds. As $\mathop{\rm arg max}\limits_{c}\{U(c,c)-cy\}\ge h$ is equivalent to $y\le (1-\alpha)e^{-(1-\alpha)\beta_{2}h}$, we deduce that for any initial point $(y_{0},h_{0})$ s.t. $y_{0}<(1-\alpha)e^{-(1-\alpha)\beta_{2}h_{0}}$, it will jump immediately to $\big(y_{0},\frac{1}{(1-\alpha)\beta_{2}}\ln(\frac{1-\alpha}{y_{0}})\big)$ which is on the curve $y= (1-\alpha)e^{-(1-\alpha)\beta_{2}h}$. Hence we only need to consider initial points $(y_{0},h_{0})$ in the dual region where $y\ge (1-\alpha)e^{-(1-\alpha)\beta_{2}h}$.  Meanwhile, the second formula of (\ref{dHJB}) is equivalent to
	\begin{equation}
	\label{secbound}
	    V_{h}(y,h)=0\ {\rm for} \ y= (1-\alpha)e^{-(1-\alpha)\beta_{2}h}.
	\end{equation}
By the duality transform (\ref{tranx}), the constraint $x\ge \frac{\lambda h}{\gamma}$ is equivalent to $-\tilde{V}_{y}(y,h)\ge \frac{\lambda h}{\gamma}$. As such, the dual effective region $\mathcal{C}_{d}$ can be defined by

	\begin{equation*}
	    \mathcal{C}_{d}\triangleq\big\{(y,h)\big|y\ge (1-\alpha)e^{-(1-\alpha)\beta_{2}h},\ -\tilde{V}_{y}(y,h)\ge \frac{\lambda h}{\gamma},\ h>0\big\}.
	\end{equation*}
	Applying the duality transform again, the effective region $\mathcal{C}$ is
	\begin{equation} \label{effective}    \mathcal{C}\triangleq\big\{(x,h)\big|V_{x}(x,h)\ge (1-\alpha)e^{-(1-\alpha)\beta_{2}h},\ x\ge \frac{\lambda h}{\gamma},\ h>0 \big\}.
	\end{equation}
	Using the first equation of (\ref{dHJB}), we obtain
	\begin{equation}
	-\gamma\tilde{V}(y,h)+\frac{(r-\mu)^{2}}{2\sigma^{2}}y^{2}\tilde{V}_{yy}(y,h)=-\tilde{U}(y,h),
	\label{ode}
	\end{equation}
	where $\tilde{U}(y,h)=\sup\limits_{\lambda h\le c\le h}\{U(c,h)-cy\}$ with
	
	\begin{equation*}
	\tilde{U}(y,h)=\left\{
	\begin{array}{l}
	\frac{1}{\beta_{1}}(1-e^{(\alpha-\lambda)\beta_{1}h})-\lambda hy,\ \ \ \ \  e^{(\alpha-\lambda)\beta_{1}h}\le y,\\
	\frac{1}{\beta_{1}}\big[1-y+y\ln(y)\big]-\alpha hy,\ \ \ 1\le y<e^{(\alpha-\lambda)\beta_{1}h},\\
	\frac{1}{\beta_{2}}\big[1-y+y\ln(y)\big]-\alpha hy,\  \ e^{-(1-\alpha)\beta_{2}h}\le y<1,\\
	\frac{1}{\beta_{2}}\big[(1-e^{-(1-\alpha)\beta_{2}h}\big)-hy,\ \ \ (1-\alpha)e^{-(1-\alpha)\beta_{2}h}\le y<e^{-(1-\alpha)\beta_{2}h}.
	\end{array}
	\right.
	\end{equation*}
		Define $k=\frac{(r-\mu)^{2}}{2\sigma^{2}},\ q_{1}=\frac{k-\sqrt{k^{2}+4k\gamma}}{2k}<0,\ q_{2}=\frac{k+\sqrt{k^{2}+4k\gamma}}{2k}>1$, then the general solution of (\ref{ode}) is
\begin{equation}\label{tilv}
		\!\!\!  \!\!\! \!\!\!  \tilde{V}(y,h)\!\!=\!\!\left\{
		    \begin{array}{l}
		   C_{1}(h)y^{q_{1}}+C_{2}(h)y^{q_{2}}-\frac{1}{\gamma}\lambda hy+\frac{1}{\gamma\beta_{1}}\big[1-e^{(\alpha-\lambda)\beta_{1}h}\big],\ \ \ \ \  e^{(\alpha-\lambda)\beta_{1}h}\le y,\\
		   C_{3}(h)y^{q_{1}}\!\!+\!\!C_{4}(h)y^{q_{2}}\!\!+\!\!\frac{1}{\gamma\beta_{1}}\big[1-y+y\ln(y)\big]+
\frac{k}{\gamma^{2}\beta_{1}}y-\frac{1}{\gamma}\alpha hy,\ 1\le y<e^{(\alpha-\lambda)\beta_{1}h},\\
		   C_{5}(h)y^{q_{1}}\!\!+\!\!C_{6}(h)y^{q_{2}}\!\!+\!\!\frac{1}{\gamma\beta_{2}}\big[1-y+y\ln(y)\big]+\frac{k}{\gamma^{2}\beta_{2}}y-\frac{1}{\gamma}\alpha hy,\ e^{-(1-\alpha)\beta_{2}h}\le y<1,\\
		   C_{7}(h)y^{q_{1}}\!\!+\!\!C_{8}(h)y^{q_{2}}\!\!-\!\!\frac{1}{\gamma}hy\!\!+
\!\!\frac{1}{\gamma\beta_{2}}\big[1\!-\!e^{-(1\!-\!\alpha)\beta_{2}h}\big],\ (1\!-\!\alpha)e^{-(1\!-\!\alpha)\beta_{2}h}\!\le\! y\!<\!e^{-(1-\alpha)\beta_{2}h}.
		    \end{array}
		    \right.
\end{equation}
Based on  the duality transform (\ref{tranx}), the boundary  condition (\ref{boundcon1}) can be rewritten as follows:
		\begin{equation}
		\label{boundcon2}
		  \lim\limits_{\tilde{V}_{y}(y,h)\rightarrow -\frac{\lambda h}{\gamma}^{-}}y\tilde{V}_{yy}(y,h)=0.
		 \end{equation}
Then we deduce from (\ref{boundcon2}) and (\ref{tilv}) that $\tilde{V}_{y}(y,h)\rightarrow -\frac{\lambda h}{\gamma}\Leftrightarrow y\rightarrow\infty$ and that $C_{2}(h)=0,\ C_{1}(h)>0$. Using (\ref{tilv}), we obtain $-\tilde{V}_{y}(y,h)\ge\frac{\lambda h}{\gamma}$ and the dual effective region is simplified to
		 \begin{equation}
	    \label{dualeffective}
	    \mathcal{C}_{d}=\big\{(y,h)\big|y\ge (1-\alpha)e^{-(1-\alpha)\beta_{2}h},\ h>0\big\}.
	\end{equation}
In addition, when $y=(1-\alpha)e^{-(1-\alpha)\beta_{2}h}$ and $h\rightarrow\infty$, the initial wealth $x\rightarrow\infty$ and the utility keeps near its maximum $\frac{1}{\beta_{2}}$ for infinitely long time. Thus the value function tends to $\frac{1}{\gamma\beta_{2}}$. This boundary condition can be expressed as

		\begin{equation*}
		  \lim\limits_{h\rightarrow\infty}\big[\tilde{V}(y,h)-y\tilde{V}_{y}(y,h)\big]\Big|_{y=(1-\alpha)e^{-(1-\alpha)\beta_{2}h}}=\frac{1}{\gamma\beta_{2}}.
		\end{equation*}
		The above boundary condition together with $C_{2}(h)=0$, (\ref{secbound}) and smooth-fit conditions
		
		\begin{align*}
	 \tilde{V}(y+,h)=\tilde{V}(y-,h),\\
	 \tilde{V}_{y}(y+,h)=\tilde{V}_{y}(y-,h)
	\end{align*}
	leads to
	
	\begin{align}
	C_{2}(h)=&0, \ \ \   C_{4}(h)=-\frac{k}{\gamma^{2}\beta_{1}}\frac{1-q_{1}}{q_{2}-q_{1}}e^{-(\alpha-\lambda)(q_{2}-1)\beta_{1}h},\label{C1}\\
	    C_{6}(h)=&C_{4}(h)+\frac{k}{\gamma^{2}}\frac{\beta_{2}-\beta_{1}}{\beta_{1}\beta_{2}}\frac{1-q_{1}}{q_{2}-q_{1}},\label{C2}\\
	    C_{8}(h)=&C_{6}(h)+\frac{k}{\gamma^{2}\beta_{2}}\frac{1-q_{1}}{q_{2}-q_{1}}e^{(1-\alpha)(q_{2}-1)\beta_{2}h},\label{C3}\\
	    C_{7}(h)=&\frac{(1-\alpha)^{q_{2}-q_{1}}\frac{k}{\gamma^{2}}\frac{1-q_{1}}{q_{2}-q_{1}}
(\alpha-\lambda)(q_{2}-1)}{(1-\alpha)(q_{2}-q_{1})\beta_{2}
+(\alpha-\lambda)(q_{2}-1)\beta_{1}}e^{-\big[(1-\alpha)(q_{2}-q_{1})\beta_{2}
+(\alpha-\lambda)(q_{2}-1)\beta_{1}\big]h}\notag\\
	    &+(1-\alpha)^{q_{2}-q_{1}}\frac{k}{\gamma^{2}\beta_{2}}\frac{q_{2}-1}{q_{2}-q_{1}}
e^{-(1-\alpha)(1-q_{1})\beta_{2}h},\label{C5}\\
	    C_{5}(h)=&C_{7}(h)-\frac{k}{\gamma^{2}\beta_{2}}\frac{q_{2}-1}{q_{2}-q_{1}}e^{-(1-\alpha)(1-q_{1})\beta_{2}h},\notag\\
	    C_{3}(h)=&C_{5}(h)-\frac{k}{\gamma^{2}}
\frac{\beta_{2}-\beta_{1}}{\beta_{1}\beta_{2}}\frac{q_{2}-1}{q_{2}-q_{1}},\label{C6}\\
	    C_{1}(h)=&C_{3}(h)+\frac{k}{\gamma^{2}\beta_{1}}
\frac{q_{2}-1}{q_{2}-q_{1}}e^{(\alpha-\lambda)(1-q_{1})\beta_{1}h}.\label{C7}
	\end{align}
We can  directly show

	    \begin{align*}
	        &C_{1}(h)>0,\ \ C_{4}(h)<0,\\
	        &C_{7}(h)>0, \ \ C_{8}(h)>0
	        \end{align*}
	 and obtain the following order estimates of the coefficients $C_{i}(h),\ 1\le i\le 8$, which will later be used in the proof of the verification theorem.
\begin{lemma}	\label{hinfty}
	    As $h\rightarrow\infty$,
	    
	    \begin{align*}
	    &C_{1}(h)=\mathcal{O}(e^{(\alpha-\lambda)(1-q_{1})\beta_{1}h}),\\
	    &C_{3}(h)=\mathcal{O}(1),\ \ C_{4}(h)=\mathcal{O}(e^{-(\alpha-\lambda)(q_{2}-1)\beta_{1}h}),\\
	    &C_{5}(h)=\mathcal{O}(e^{-(1-\alpha)(1-q_{1})\beta_{2}h}),\ \ C_{6}(h)=\mathcal{O}(1),\\
	    &C_{7}(h)=\mathcal{O}(e^{-(1-\alpha)(1-q_{1})\beta_{2}h}),\ \ C_{8}(h)=\mathcal{O}(e^{(1-\alpha)(q_{2}-1)\beta_{2}h}).
	    \end{align*}
	\end{lemma}
\vskip 15pt

\section{Verification Theorem and Optimal Strategy}\label{veropt}
In this section, we establish the verification theorem and apply the duality transform to obtain the optimal strategy. First, we state the verification theorem which gives the optimal consumption and investment policy in dual form.
\begin{theorem}[Verification Theorem]\label{verificationth}
For any $(x_{0},h_{0})\in\mathcal{C}$, where $x_{0}$ and $h_{0}$ are respectively the initial wealth and initial past spending maximum, and $\mathcal{C}$ is the effective region given by (\ref{effective}), the value function $V(x_{0},h_{0})$ can be attained by the optimal consumption and investment strategy given by

		\begin{equation*}
		(c^{*},\pi^{*})=\bigg\{\Big(c^{*}\big(Y_{t}(y^{*}),H^{\dagger}_{t}(y^{*})\big),\pi^{*}\big(Y_{t}(y^{*}),H^{\dagger}_{t}(y^{*})\big)\Big),t\ge 0\bigg\}
		\end{equation*}
		where
		$Y_{t}(\cdot)$ is given by
		
		\begin{equation*}
		Y_{t}(y)\triangleq ye^{\gamma t}M_{t}
		\end{equation*}
with $M \triangleq \left \{ M_{t}:=e^{-(\gamma+\frac{(\mu-r)^{2}}{2\sigma^{2}})t-\frac{\mu-r}{\sigma}B_{t}}, t\geq 0 \right \}$ being the discounted state price density process,  $\left\{H^{\dagger}_{t}(\cdot), t\ge 0\right \}$ is determined by

		\begin{equation*}
		    H^{\dagger}_{t}(y)\triangleq h_{0}\vee \sup\limits_{s\le t}c^{*}\big(Y_{s}(y),H^{\dagger}_{s}(y)\big),
		\end{equation*}
		$y^{*}=y^{*}(x_{0},h_{0})$ is the unique solution of
		\begin{equation}
		\label{ystar}
		   \mE_{x_{0},h_{0}}\int_{0}^{\infty}c^{*}\big(Y_{t}(y),H^{\dagger}_{t}(y)\big)M_{t}dt=x_{0}
		\end{equation}
		and the feedback functions $c^{*}(\cdot,\cdot)$ and $\pi^{*}(\cdot,\cdot)$ are respectively given by
		
		\begin{align}
		    &c^{*}(y,h)=\left\{
	\begin{array}{l}
	\lambda h,\hskip 2.3cm e^{(\alpha-\lambda)\beta_{1}h}\le y,\\
	-\frac{1}{\beta_{1}}\ln(y)+\alpha h,\ 1\le y<e^{(\alpha-\lambda)\beta_{1}h},\\
	-\frac{1}{\beta_{2}}\ln(y)+\alpha h,\ e^{-(1-\alpha)\beta_{2}h}\le y<1,\\
	h,\hskip 2.3cm (1-\alpha)e^{-(1-\alpha)\beta_{2}h}\le y<e^{-(1-\alpha)\beta_{2}h}.
	\end{array}
	\right.\label{cstar}\\
	&\pi^{*}(y,h)=\frac{\mu-r}{\sigma^{2}}\left\{
	\begin{array}{l}
	\frac{r}{k}\big[C_{1}(h)y^{q_{1}-1}+C_{2}(h)y^{q_{2}-1}\big],\hskip 1.3cm e^{(\alpha-\lambda)\beta_{1}h}\le y,\\
	\frac{r}{k}\big[C_{3}(h)y^{q_{1}-1}+C_{4}(h)y^{q_{2}-1}\big]+\frac{1}{\gamma\beta_{1}},\ 1\le y<e^{(\alpha-\lambda)\beta_{1}h},\\
	\frac{r}{k}\big[C_{5}(h)y^{q_{1}-1}+C_{6}(h)y^{q_{2}-1}\big]+\frac{1}{\gamma\beta_{2}},\ e^{-(1-\alpha)\beta_{2}h}\le y<1,\\
	\frac{r}{k}\big[C_{7}(h)y^{q_{1}-1}+C_{8}(h)y^{q_{2}-1}\big],\hskip 0.1cm (1-\alpha)e^{-(1-\alpha)\beta_{2}h}\le y<e^{-(1-\alpha)\beta_{2}h}.\label{pistar}
	\end{array}
	\right.
	\end{align}
	\end{theorem}
\begin{proof}
See  Appendix \ref{appa}.
\end{proof}
Now, to  apply the dual transform to present the primal value function  as well as the optimal consumption and investment policy with the primal variable, we need the following lemma in the dual transform, and its proof is given in Appendix \ref{appc}.
	
	\begin{lemma}
	\label{vyypo}
	    $\tilde{V}_{yy}(y,h)>0$ for $(y,h)\in\mathcal{C}_{d}$ and hence the inverse of $-\tilde{V}_{y}(\cdot,h)$ exists.
	\end{lemma}
Let $f(\cdot,h)$ be the inverse of $-\tilde{V}_{y}(\cdot,h)$ based on Lemma \ref{vyypo}, then, using dual transform (\ref{tranx}), we have
	\begin{equation}
	\label{xy}
	\left\{
	\begin{array}{l}
	x=-\tilde{V}_{y}(y,h),\\
	y=f(x,h).
	\end{array}
	\right.
	\end{equation}
	As such, based on (\ref{xy}) and Lemma \ref{vyypo}, the function $f(\cdot,h)$ is implicitly and uniquely determined by
	\begin{equation}
	\label{detf}
	x=-\tilde{V}_{y}\big(f(x,h),h\big).
	\end{equation}
	  Plugging (\ref{tilv}) into (\ref{detf}), we obtain the following forms of $f(x,h)$ according to different regions:\\
	(1) For $e^{(\alpha-\lambda)\beta_{1}h}\le f(x,h)$, we have $f(x,h)=f_{1}(x,h)$ with $f_{1}(x,h)$ satisfying
	\begin{equation}
	x=-C_{1}(h)q_{1}\big[f_{1}(x,h)\big]^{q_{1}-1}-C_{2}(h)q_{2}\big[f_{1}(x,h)\big]^{q_{2}-1}+\frac{\lambda h}{\gamma}.\label{f1}
	\end{equation}
Based on Lemma \ref{vyypo}, $V_{yy}(\cdot,h)>0$, then the inequality $e^{(\alpha-\lambda)\beta_{1}h}\le f(x,h)$ is equivalent to $x\le W_{\rm low}(h)$ with
	\begin{equation} \label{w1}
	    W_{\rm low}(h)=-C_{1}(h)q_{1}e^{-(\alpha-\lambda)(1-q_{1})\beta_{1}h}-C_{2}(h)q_{2}e^{(\alpha-\lambda)(q_{2}-1)\beta_{1}h}+\frac{\lambda h}{\gamma}.
	\end{equation}
	(2) For $1\le f(x,h)<e^{(\alpha-\lambda)\beta_{1}h}$,  $f(x,h)=f_{2}(x,h)$ with $f_{2}(x,h)$ satisfying
	\begin{equation}
	x=-C_{3}(h)q_{1}\big[f_{2}(x,h)\big]^{q_{1}-1}-C_{4}(h)q_{2}\big[f_{2}(x,h)\big]^{q_{2}-1}-\frac{1}{\gamma\beta_{1}}\ln\big[f_{2}(x,h)\big]-\frac{k}{\gamma^{2}\beta_{1}}+\frac{\alpha h}{\gamma}.\label{f2}
	\end{equation}
	The inequality $1\le f(x,h)<e^{(\alpha-\lambda)\beta_{1}h}$ is equivalent to $W_{\rm low}(h)<x\le W_{\rm ref}(h)$, where
	\begin{equation}
	    \label{w2}
	    W_{\rm ref}(h)=-C_{3}(h)q_{1}-C_{4}(h)q_{2}-\frac{k}{\gamma^{2}\beta_{1}}+\frac{\alpha h}{\gamma}.
	\end{equation}
	(3) For $e^{-(1-\alpha)\beta_{2}h}\le f(x,h)<1$, $f(x,h)=f_{3}(x,h)$ with $f_{3}(x,h)$ satisfying
	\begin{equation}
	x=-C_{5}(h)q_{1}\big[f_{3}(x,h)\big]^{q_{1}-1}-C_{6}(h)q_{2}\big[f_{3}(x,h)\big]^{q_{2}-1}
-\frac{1}{\gamma\beta_{2}}\ln\big[f_{3}(x,h)\big]-\frac{k}{\gamma^{2}\beta_{2}}+\frac{\alpha h}{\gamma}.\label{f3}
	\end{equation}
	The inequality $e^{-(1-\alpha)\beta_{2}h}\le f(x,h)<1$ is equivalent to $W_{\rm ref}(h)<x\le W_{\rm peak}(h)$ with
	\begin{equation}
	    \label{w3}
	    W_{\rm peak}(h)=-C_{5}(h)q_{1}e^{(1-\alpha)(1-q_{1})\beta_{2}h}-C_{6}(h)q_{2}e^{-(1-\alpha)(q_{2}-1)\beta_{2}h}-\frac{k}{\gamma^{2}\beta_{2}}+\frac{h}{\gamma}.
	\end{equation}
		(4) For $(1-\alpha)e^{-(1-\alpha)\beta_{2}h}\le f(x,h)<e^{-(1-\alpha)\beta_{2}h}$, $f(x,h)=f_{4}(x,h)$ with $f_{4}(x,h)$ satisfying
	\begin{equation}
	x=-C_{7}(h)q_{1}\big[f_{4}(x,h)\big]^{q_{1}-1}-C_{8}(h)q_{2}\big[f_{4}(x,h)\big]^{q_{2}-1}+\frac{h}{\gamma}.\label{f4}
	\end{equation}
	The inequality $(1-\alpha)e^{-(1-\alpha)\beta_{2}h}\le f(x,h)<e^{-(1-\alpha)\beta_{2}h}$ is equivalent to $W_{\rm peak}(h)<x\le W_{\rm updt}(h)$ with
	\begin{equation}
	    \label{w4}
	    W_{\rm updt}(h)=-C_{7}(h)q_{1}(1-\alpha)^{q_{1}-1}e^{(1-\alpha)(1-q_{1})\beta_{2}h}-C_{8}(h)q_{2}(1-\alpha)^{q_{2}-1}e^{-(1-\alpha)(q_{2}-1)\beta_{2}h}+\frac{h}{\gamma}.
	\end{equation}
We summarize the forms of the primal value function as well as the optimal consumption and investment policy in terms of primal variable in the following Theorems \ref{vfunction} and \ref{optpolicy}:
	\begin{theorem}\label{vfunction}
		For $(x,h)\in\mathcal{C}$, where $\mathcal{C}$ is the effective region given by (\ref{effective}), the value function of Problem (\ref{problem}) is
		\begin{equation}
		    \label{v}
		    V(x,h)=\left\{
		    \begin{array}{l}
		   C_{1}(h)\big[f_{1}(x,h)\big]^{q_{1}}+C_{2}(h)\big[f_{1}(x,h)\big]^{q_{2}}-\frac{1}{\gamma}\lambda hf_{1}(x,h)\\+\frac{1}{\gamma\beta_{1}}(1-e^{(\alpha-\lambda)\beta_{1}h}),\ \frac{\lambda h}{\gamma}\le x\le W_{\rm low}(h),\\
		   C_{3}(h)\big[f_{2}(x,h)\big]^{q_{1}}+C_{4}(h)\big[f_{2}(x,h)\big]^{q_{2}}
+\frac{1}{\gamma\beta_{1}}\Big[1-f_{2}(x,h)+f_{2}(x,h)\ln\big[f_{2}(x,h)\big]\Big]\\+\frac{k}{\gamma^{2}\beta_{1}}f_{2}(x,h)-\frac{1}{\gamma}\alpha hf_{2}(x,h),\ W_{\rm low}(h)<x\le W_{\rm ref}(h),\\
		   C_{5}(h)\big[f_{3}(x,h)\big]^{q_{1}}+C_{6}(h)\big[f_{3}(x,h)\big]^{q_{2}}
+\frac{1}{\gamma\beta_{2}}\Big[1-f_{3}(x,h)+f_{3}(x,h)\ln\big[f_{3}(x,h)\big]\Big]\\
+\frac{k}{\gamma^{2}\beta_{2}}f_{3}(x,h)-\frac{1}{\gamma}\alpha hf_{3}(x,h),\ W_{\rm ref}(h)<x\le W_{\rm peak}(h),\\
		   C_{7}(h)\big[f_{4}(x,h)\big]^{q_{1}}+C_{8}(h)\big[f_{4}(x,h)\big]^{q_{2}}
-\frac{1}{\gamma}hf_{4}(x,h)\\
+\frac{1}{\gamma\beta_{2}}\big[1-e^{-(1-\alpha)\beta_{2}h}\big],\ W_{\rm peak}(h)<x\le W_{\rm updt}(h),
		    \end{array}
		    \right.
		\end{equation}
		where $W_{\rm low}(h),W_{\rm ref}(h),W_{\rm peak}(h),W_{\rm updt}(h)$ and $f_{i}(x,h),\ 1\le i\le 4$ are given by (\ref{f1})$\sim$(\ref{w4}).
	\end{theorem}
\begin{proof}
		Applying dual transform (\ref{tranx}) and (\ref{xy}) yields
		
		\begin{equation*}
		V(x,h)=\tilde{V}\big(f(x,h),h\big)+xf(x,h).
		\end{equation*}
		Plugging in (\ref{tilv}),  the desired result follows.
\end{proof}
	
\begin{theorem}	\label{optpolicy}
		For $(x_{0},h_{0})\in\mathcal{C}$, where $\mathcal{C}$ is the effective region given by (\ref{effective}), let $c^{*}_{\rm primal}(\cdot,\cdot)$ and $ \pi^{*}_{\rm primal}(\cdot,\cdot)$ be the feedback functions in terms of primal variable given respectively by
		
		\begin{align*}
		    &c_{\rm primal}^{*}(x,h)=\left\{
	\begin{array}{l}
	\lambda h,\hskip 3.3cm \frac{\lambda h}{\gamma}\le x\le W_{\rm low}(h),\\
	-\frac{1}{\beta_{1}}\ln\big[f_{2}(x,h)\big]+\alpha h,\ W_{\rm low}(h)<x\le W_{\rm ref}(h),\\
	-\frac{1}{\beta_{2}}\ln\big[f_{3}(x,h)\big]+\alpha h,\ W_{\rm ref}(h)<x\le W_{\rm peak}(h),\\
	h,\hskip 3.3cm W_{\rm peak}(h)<x\le W_{\rm updt}(h),
	\end{array}
	\right.\\
	&\pi_{\rm primal}^{*}(x,h)=\frac{\mu-r}{\sigma^{2}}\left\{
	\begin{array}{l}
	\frac{r}{k}\left\{C_{1}(h)\big[f_{1}(x,h)\big]^{q_{1}-1}+C_{2}(h)\big[f_{1}(x,h)\big]^{q_{2}-1}\right\},\hskip 0.4cm \frac{\lambda h}{\gamma}\le x\le W_{\rm low}(h),\\
	\frac{r}{k}\left\{C_{3}(h)\big[f_{2}(x,h)\big]^{q_{1}-1}+C_{4}(h)\big[f_{2}(x,h)\big]^{q_{2}-1}\right\}
+\frac{1}{\gamma\beta_{1}},\ W_{\rm low}(h)<x\le W_{\rm ref}(h),\\
	\frac{r}{k}\left\{C_{5}(h)\big[f_{3}(x,h)\big]^{q_{1}-1}+C_{6}(h)\big[f_{3}(x,h)\big]^{q_{2}-1}\right\}+\frac{1}{\gamma\beta_{2}},\ W_{\rm ref}(h)<x\le W_{\rm peak}(h),\\
	\frac{r}{k}\left\{C_{7}(h)\big[f_{4}(x,h)\big]^{q_{1}-1}+C_{8}(h)\big[f_{4}(x,h)\big]^{q_{2}-1}\right\},\hskip 0.3cm W_{\rm peak}(h)<x\le W_{\rm updt}(h),
	\end{array}
	\right.
	\end{align*}
	where $W_{\rm low}(h),W_{\rm ref}(h),W_{\rm peak}(h),W_{\rm updt}(h)$ and $f_{i}(x,h),\ 1\le i\le 4$ are given by (\ref{f1})$\sim$(\ref{w4}).
	
		Then SDE
		\begin{equation*}
		\left\{
		\begin{array}{l}
		dX_{t}=rX_{t}dt+\pi^{*}_{\rm primal}(X_{t},H^{*}_{t})(\mu-r)dt+\pi^{*}_{\rm primal}(X_{t},H^{*}_{t})\sigma dW_{t}-c^{*}_{\rm primal}(X_{t},H^{*}_{t})dt,\\
		X_{0}=x_{0}
		\end{array}
		\right.
		\end{equation*}
		with $H^{*}_{t}\triangleq h_{0}\vee \sup\limits_{s\le t}c^{*}_{\rm primal}(X_{s},H^{*}_{s})$ and $H^{*}_{0}=h_{0}$,  has a unique strong solution $\{X^{*}_{t},\ t\ge 0\}$. The optimal consumption and investment strategy is
		
		\begin{equation*}
		\Big\{\big(c^{*}_{\rm primal}(X^{*}_{t},H^{*}_{t}),\pi^{*}_{\rm primal}(X^{*}_{t},H^{*}_{t})\big),\ t\ge 0\Big\}.
		\end{equation*}
	\end{theorem}
\begin{proof}
	    The proof is based on the following Lemmas \ref{fc1} and \ref{lipsc}. Then, as the proof is similar to that of \cite{deng2021}, we omit it here.
\end{proof}	
\begin{lemma}\label{fc1}
		The function $f$ is $C^{1}$ within each sub-region of $\mathcal{C}$: $\frac{\lambda h}{\gamma}\le x\le W_{\rm low}(h)$, $W_{\rm low}(h)<x\le W_{\rm ref}(h)$, $W_{\rm ref}(h)<x\le W_{\rm peak}(h)$, $W_{\rm peak}(h)<x\le W_{\rm updt}(h)$, and it is continuous at the boundary of $x=W_{\rm low}(h)$, $x=W_{\rm ref}(h)$, $x=W_{\rm peak}(h)$. Moreover,  we have
		
		\begin{align}
		&f_{x}(x,h)=-\frac{1}{\tilde{V}_{yy}(f,h)}\notag\\
		&=\left\{
		\begin{array}{l}
		\frac{k}{r}\left\{-C_{1}(h)\big[f_{1}(x,h)\big]^{q_{1}-2}-C_{2}(h)\big[f_{1}(x,h)\big]^{q_{2}-2}\right\}^{-1},
\hskip 1.7cm \frac{\lambda h}{\gamma}\le x\le W_{\rm low}(h),\\
	\!\bigg(\frac{r}{k}\left\{\!-\!C_{3}(h)\big[f_{2}(x,h)\big]^{q_{1}-2}\!-\!C_{4}(h)\big[f_{2}(x,h)\big]^{q_{2}-2}\right\}\!-\!\frac{1}{\gamma\beta_{1}f_{2}(x,h)}\bigg)^{-1}\!,\!\ W_{\rm low}(h)\!<\!x\!\le\! W_{\rm ref}(h),\\
	\!\bigg(\frac{r}{k}\left\{\!-\!C_{5}(h)\big[f_{3}(x,h)\big]^{q_{1}-2}\!-\!C_{6}(h)\big[f_{3}(x,h)\big]^{q_{2}-2}\right\}\!-\!\frac{1}{\gamma\beta_{2}f_{3}(x,h)}\bigg)^{-1}\!,\!\ W_{\rm ref}(h)\!<\!x\!\le\! W_{\rm peak}(h),\\
	\frac{k}{r}\left\{-C_{7}(h)\big[f_{4}(x,h)\big]^{q_{1}-2}-C_{8}(h)\big[f_{4}(x,h)\big]^{q_{2}-2}\right\}^{-1},\hskip 1.7cm W_{\rm peak}(h)<x\le W_{\rm updt}(h),
		\end{array}
		\right.\label{fx}\\
		&f_{h}(x,h)=\tilde{V}_{yh}(f,h)f_{x}(x,h).\label{fh}
		\end{align}
\end{lemma}
\begin{proof}
The proof is similar to Lemma 5.6 in \cite{deng2021} and omitted here.
\end{proof}
\begin{lemma}\label{lipsc}
		The function $c^{*}_{\rm primal}$ is locally Lipschitz on $\mathcal{C}$ and the function $\pi^{*}_{\rm primal}$ is Lipschitz on $\mathcal{C}$.
\end{lemma}
\begin{proof}	
See Appendix \ref{appc}.
\end{proof}	
\vskip 15pt
\section{Numerical Analysis with Fixed Parameters}\label{numana}

	This section aims to illustrate and analyze some properties of the optimal policy and relevant boundaries by fixing the market parameters and numerically computing the results presented in Theorems \ref{vfunction} and \ref{optpolicy}.

For simplicity, we define the boundary of the lowest wealth level to satisfy the consumption constraint $c>\lambda h$ as

\begin{equation*}
	    W_{\rm bkrp}(h)=\frac{\lambda h}{\gamma}.
\end{equation*}
	The effective region is then between the two boundaries $x=W_{\rm updt}(h)$ and $x=W_{\rm bkrp}(h)$. Using three boundaries $x=W_{\rm low}(h),x=W_{\rm ref}(h)$ and $x=W_{\rm peak}(h)$, the effective region is further separated into four parts where the investor takes different strategies in consumption and portfolio selection due to different states of wealth and habit.  The ineffective region is separated into two parts $\mathcal{C}^{c}_{1}$ and $\mathcal{C}^{c}_{2}$:
	
	\begin{align*}
	    &\mathcal{C}^{c}_{1}\triangleq\big\{(x,h)\big|x<W_{\rm bkrp}(h),\ h>0 \big\},\\
	    &\mathcal{C}^{c}_{2}\triangleq\big\{(x,h)\big|x>W_{\rm updt}(h),\ h>0 \big\}.
	\end{align*}
	$\mathcal{C}^{c}_{1}$ defines the region where the wealth is too low to maintain the lowest consumption level $c=\lambda h$. The other part $\mathcal{C}^{c}_{2}$ implies that the wealth is so high w.r.t the current running maximum level (it can only happen at time $t=0$ in the optimal case) that it is optimal to consume at a level strictly higher than the running maximum and forces $(x,h)$ to jump to $x=W_{\rm updt}(h)$.
	
	\begin{figure}[ht]
\centering
\includegraphics[width=4in, keepaspectratio]{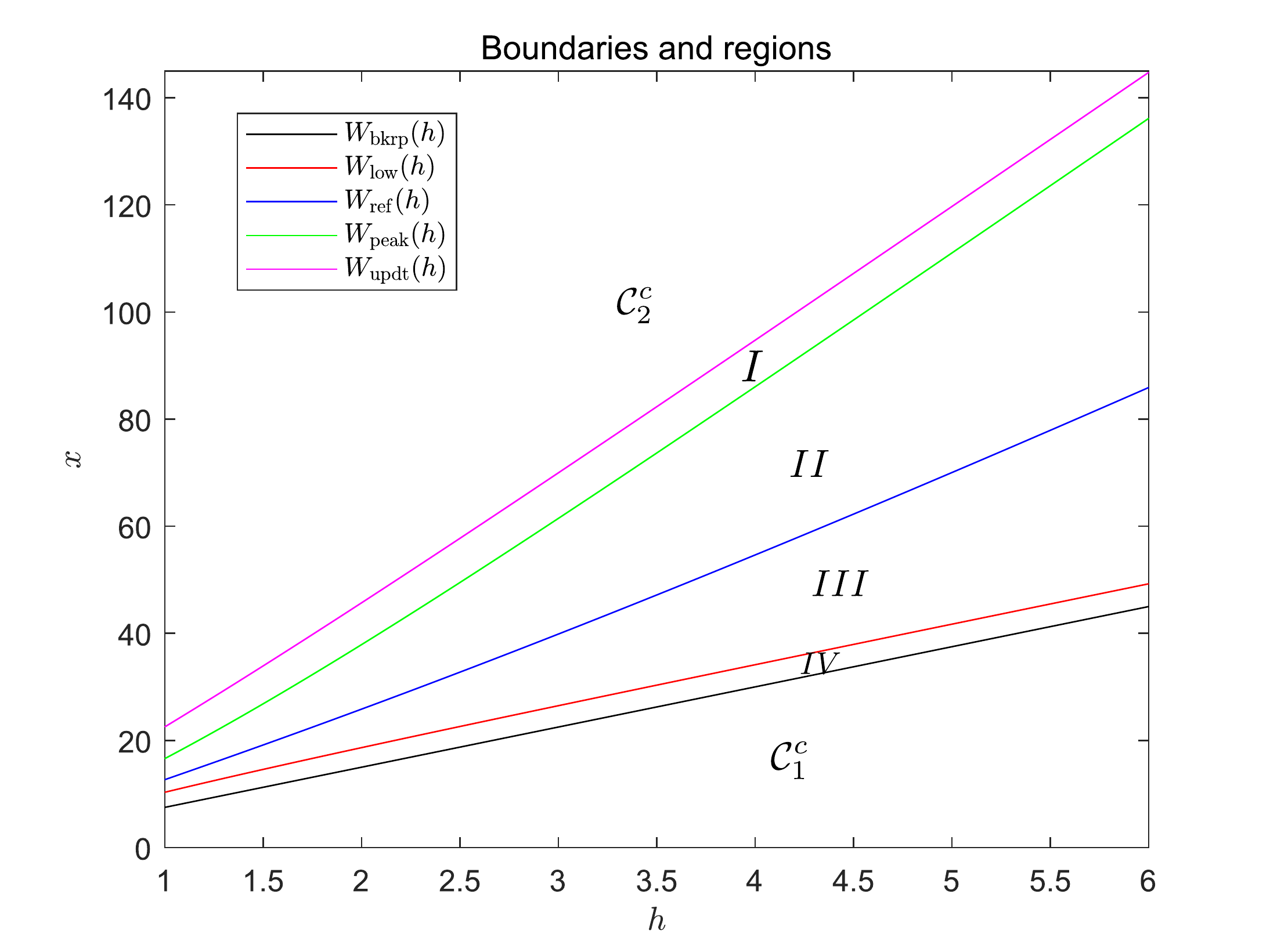}\\
\caption{Boundary curves $x=W_{\rm bkrp}(h),x=W_{\rm low}(h),x=W_{\rm ref}(h),x=W_{\rm peak}(h)$ and $x=W_{\rm updt}(h)$, different sub-regions \uppercase\expandafter{\romannumeral 1}$\sim$\uppercase\expandafter{\romannumeral 4} (\uppercase\expandafter{\romannumeral 1}:satisfactory, \uppercase\expandafter{\romannumeral 2}:recovery, \uppercase\expandafter{\romannumeral 3}:depression, \uppercase\expandafter{\romannumeral 4}:gloom) of the effective region and two sub-regions $\mathcal{C}_{i}^{c},i=1,2$ of the ineffective region with $\lambda=0.3,\ \alpha=0.7,\ \beta_{1}=1,\ \beta_{2}=2,\ r=0.04,\ \mu=0.12,\ \sigma=0.3$.}
\label{boundarygra}
\end{figure}

	We fix parameters $\lambda=0.3,\ \alpha=0.7,\ \beta_{1}=1,\ \beta_{2}=2,\ r=0.04,\ \mu=0.12,\ \sigma=0.3$ and compute all the boundaries that separate the ineffective region and the effective region and different sub-regions of the effective region. As shown in Figure \ref{boundarygra}, all boundaries are graphs of increasing functions w.r.t the variable $h$. For fixed running maximum level $h$, if the wealth $x$ is so low that $(x,h)$ fall into $\mathcal{C}_{1}^{c}$, the investor is too poor to sustain the lower bound consumption constraint; if the wealth is a little higher that $(x,h)$ belongs to sub-region \uppercase\expandafter{\romannumeral 4} of the effective region, then the investor consumes at the lowest level set by the drawdown constraint; if the wealth is higher but not high enough (sub-region \uppercase\expandafter{\romannumeral 3}), the investor chooses to consume at a higher level than the lowest level but no more than the reference point level $\alpha h$ where risk aversion increases; as the wealth increase and $(x,h)$ enters sub-region \uppercase\expandafter{\romannumeral 2}, the investor has enough wealth to consume above the reference point level $\alpha h$ but not enough to reach the running maximum level; with some more wealth than the former case (now in sub-region \uppercase\expandafter{\romannumeral 1}), the investor is able the consume at the running maximum level but not wealthy enough to update it; if the investor has more wealth than $W_{\rm updt}(h)$, then he may update the running maximum immediately by consuming above the historical running maximum level and causes $(x,h)$ to jump onto the boundary $x=W_{\rm updt}(h)$, where he consumes at the running maximum level and continuously updates it. Based on the aforementioned economic interpretation, we name the curve $W_{\rm updt}(h)$ as ``bliss" curve, and four sub-regions \uppercase\expandafter{\romannumeral 1}-\uppercase\expandafter{\romannumeral 4} as ``satisfactory", ``recovery", ``depression" and ``gloom" region, respectively. As can be seen from the analysis below, such a division of state space helps to provide a structural description of both consumption and investment behaviour under our model.
	
	\begin{figure}[ht]
\centering
{%
    \includegraphics[width=.31\linewidth]{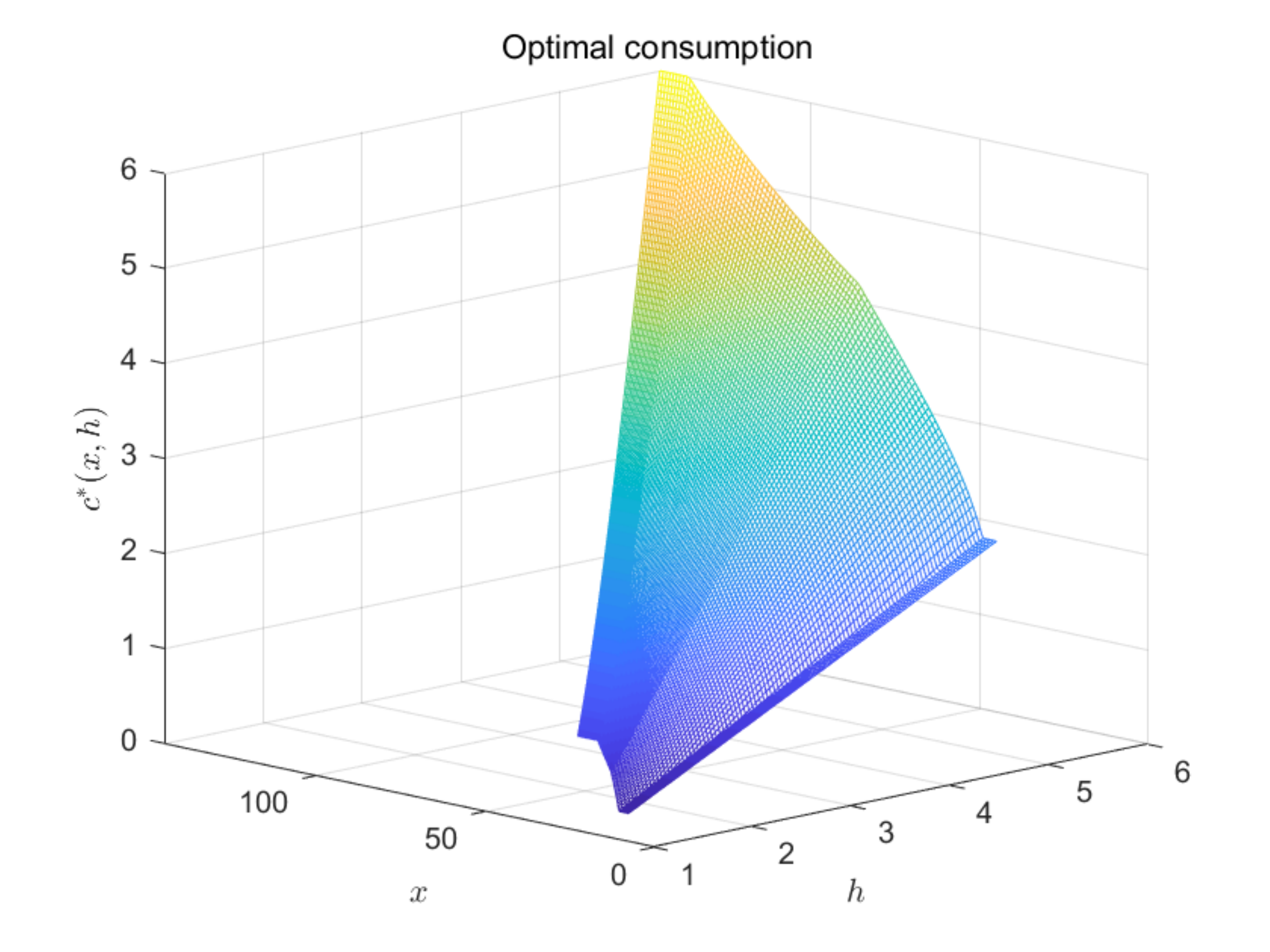}\quad
    \includegraphics[width=.31\linewidth]{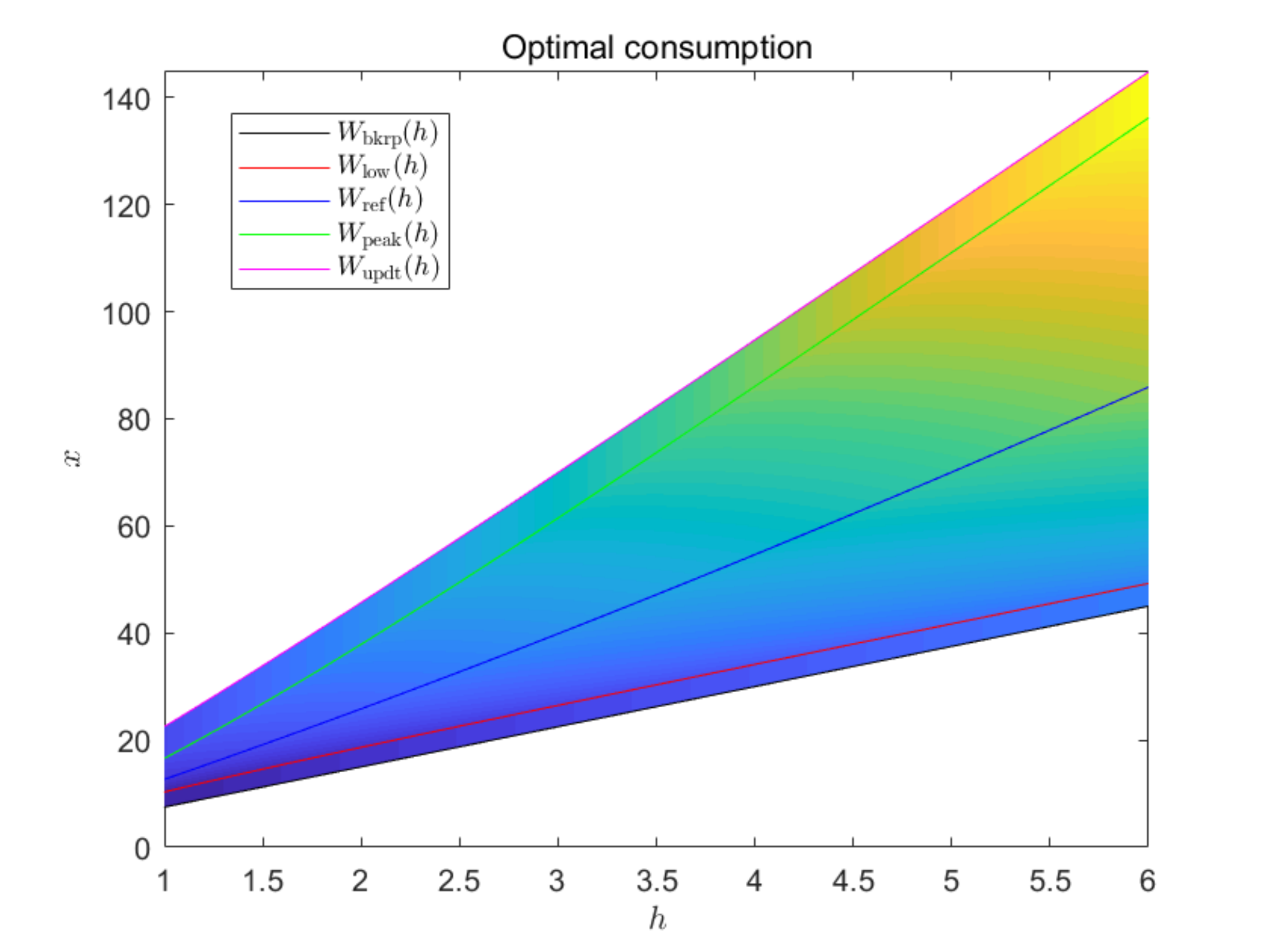}\quad
    \includegraphics[width=.31\linewidth]{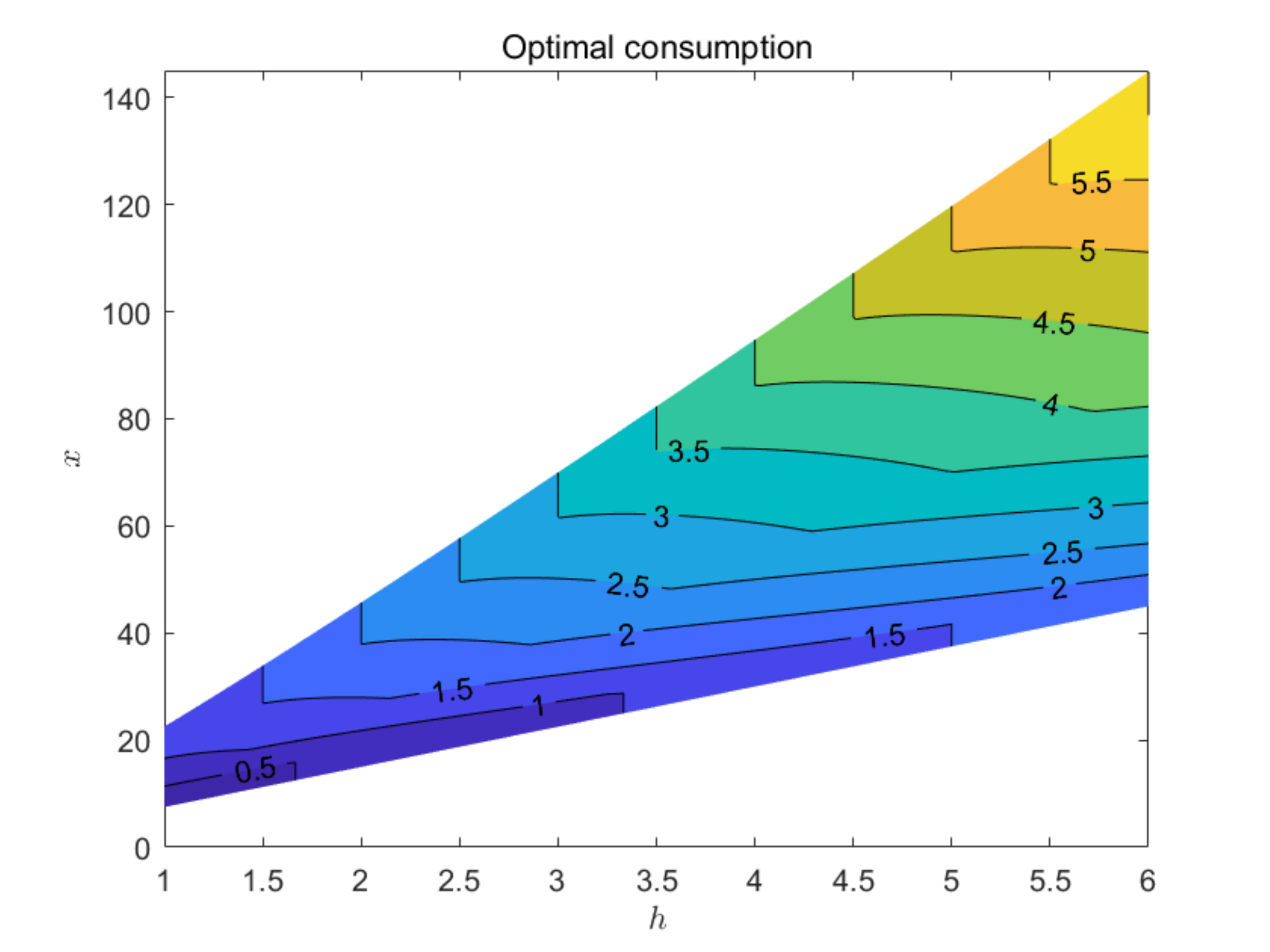}
    }
\caption{Optimal consumption with fixed parameters $\lambda=0.3,\ \alpha=0.7,\ \beta_{1}=1,\ \beta_{2}=2,\ r=0.04,\ \mu=0.12,\ \sigma=0.3$. The left panel plots the optimal consumption as a function of wealth and consumption peak. The middle panel is the two-dimensional projection of the left panel with the boundaries shown. The right panel is the contour plot.}
\label{congra}
\end{figure}
	
	The optimal consumption, shown in Figure \ref{congra}, is non-decreasing in both the wealth $x$ and the habit $h$. In satisfactory region and gloom region, the optimal consumption is indifferent with respect to $x$. That is, for the poorest people in effective region (gloom), they consume as little as possible, while for the wealthiest (satisfactory), they revisit their historical peak of consumption rate. However, in depression region and recovery region, increasing wealth will lead to an increase in optimal consumption and the increase is more substantial in the region with lower risk aversion (i.e., more substantial in depression region when $\beta_{1}<\beta_{2}$, and in recovery region when $\beta_1>\beta_2$). The above analysis suggests that increasing wealth causes one to consume more only in the following two cases: the first case is when the wealth is at least $W_{\rm updt}(h)$, he is so rich that he decides to consume more even at the cost of raising running maximum $h$; the other case is when his wealth is more than $W_{\rm low}(h)$ but less than $W_{\rm peak}(h)$. There are two sub-cases divided by whether the wealth is more than $W_{\rm ref}(h)$ in the second case, and the (marginal propensity to consume) MPC out of wealth is generally higher  in the region with lower risk aversion.
	
		\begin{figure}[ht]
\centering
    {%
    \includegraphics[width=.31\linewidth]{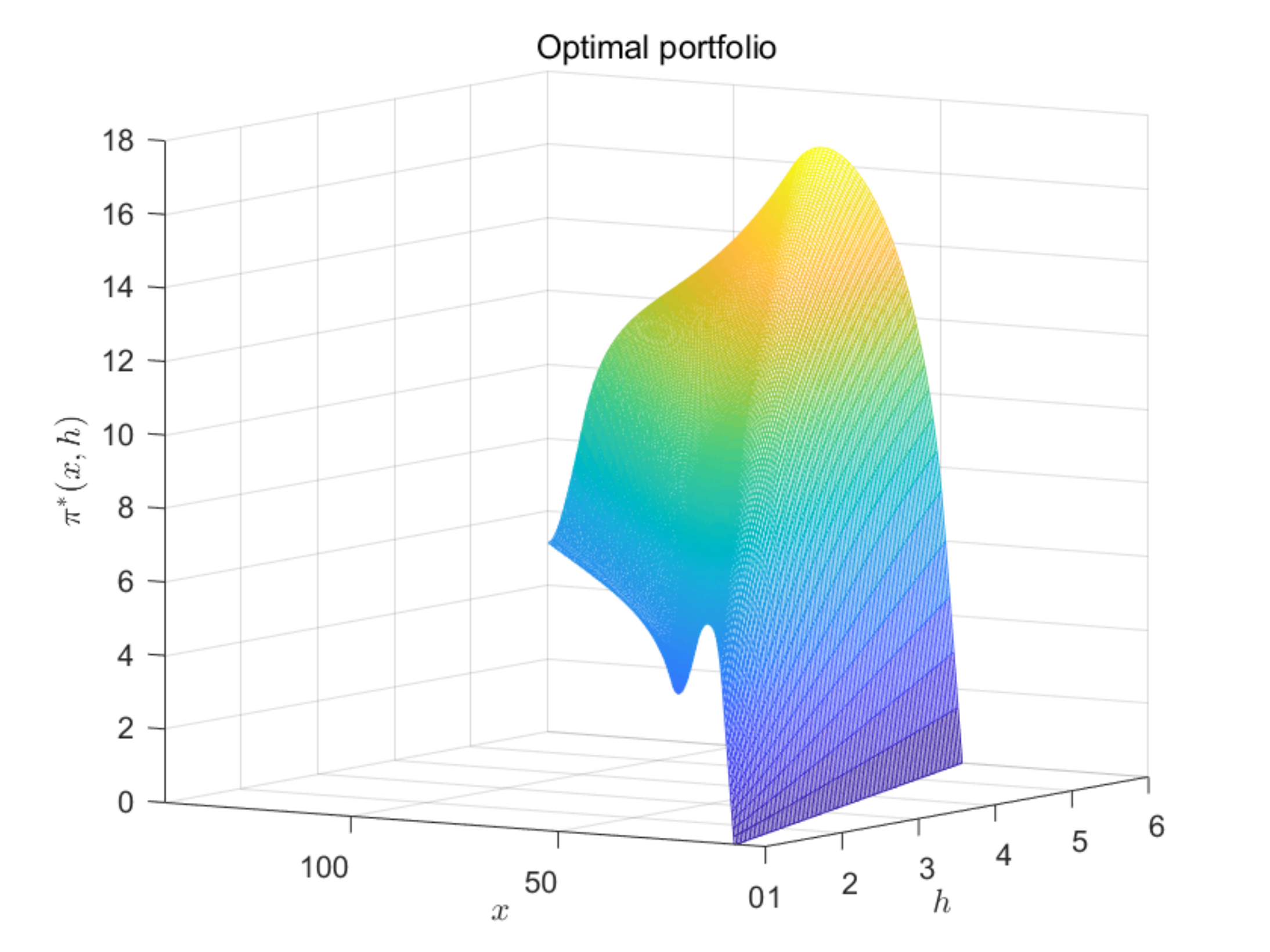}\quad
    \includegraphics[width=.31\linewidth]{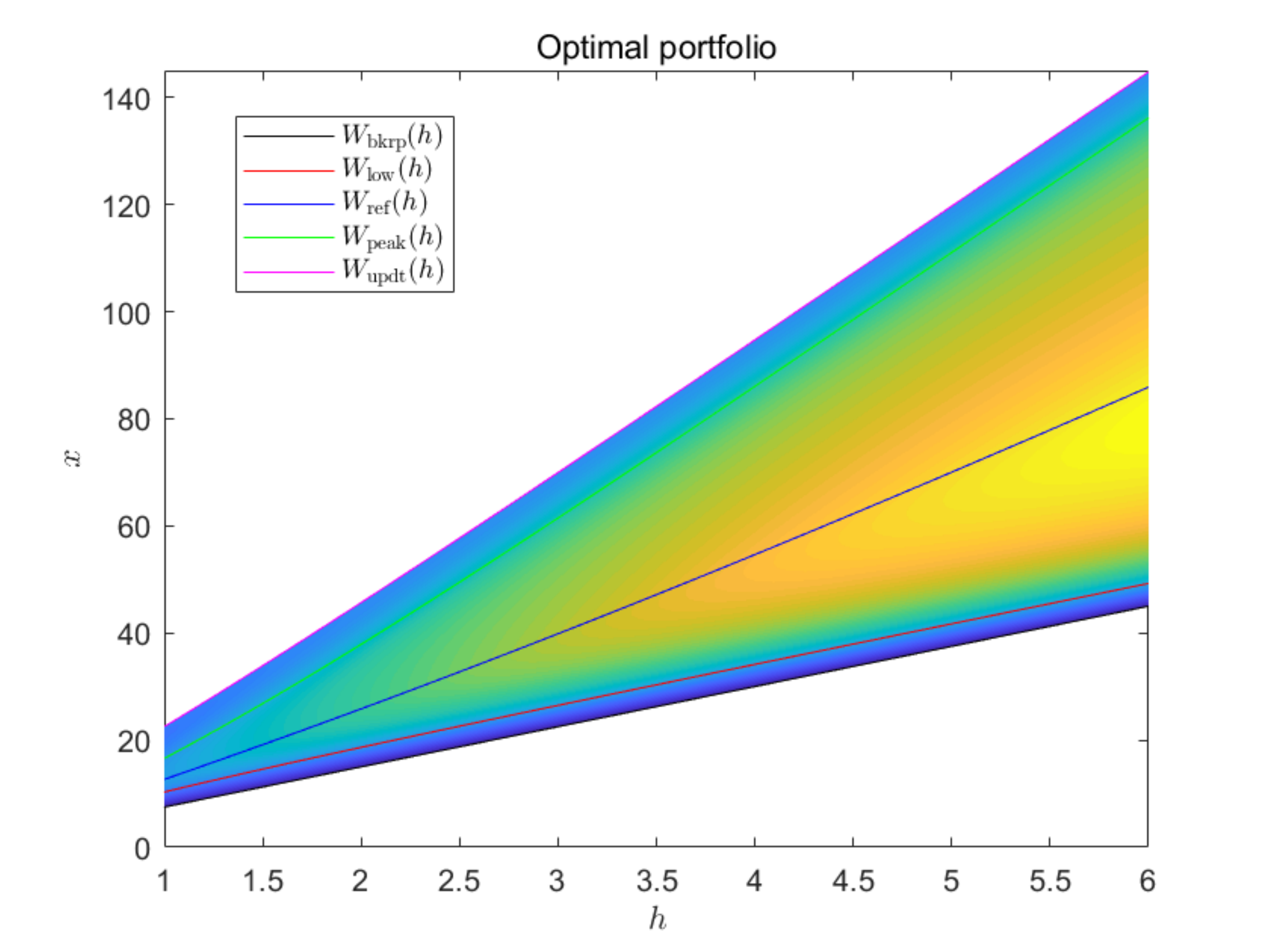}\quad
    \includegraphics[width=.31\linewidth]{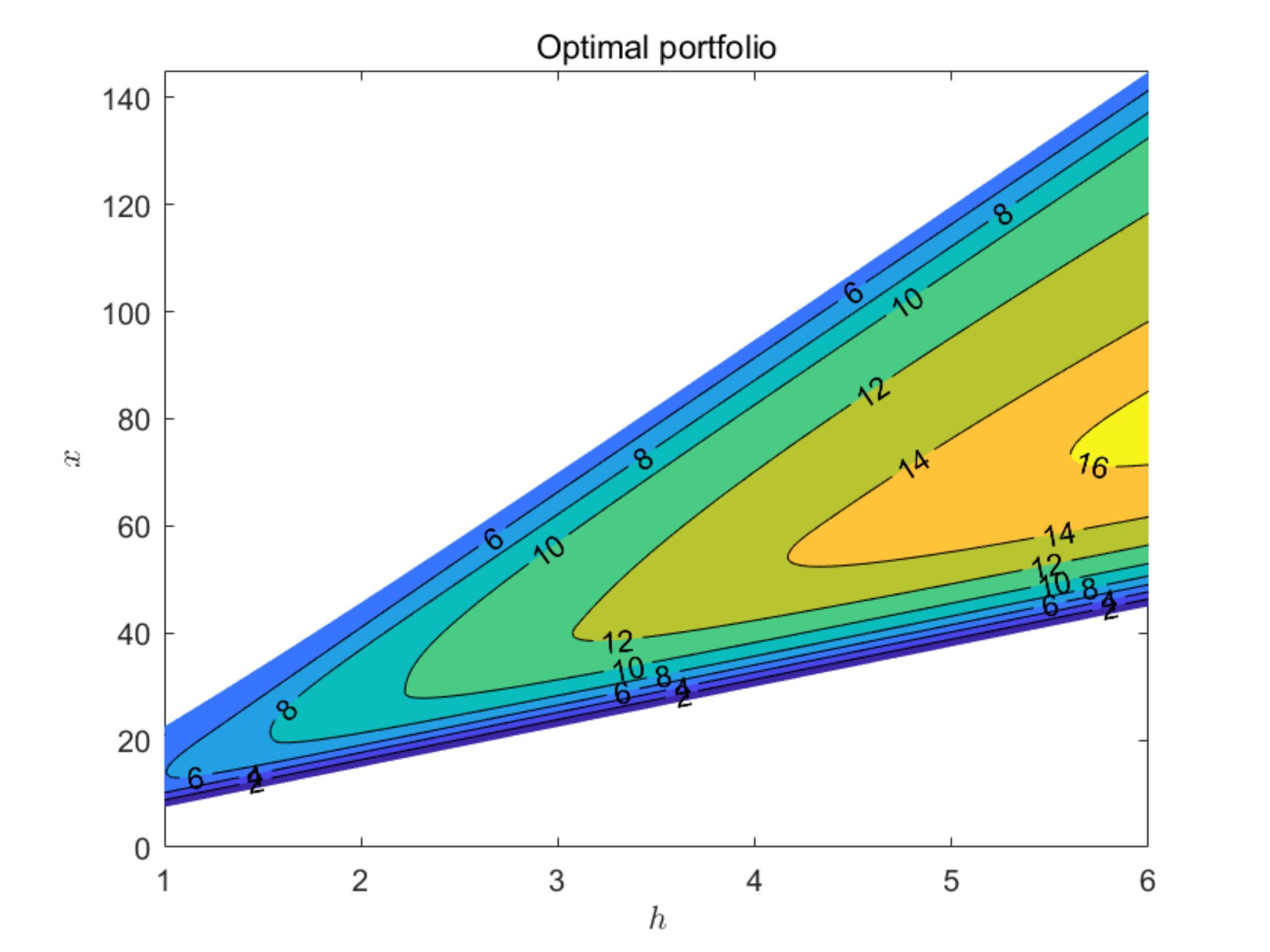}
    }
\caption{Optimal portfolio with fixed parameters $\lambda=0.3,\ \alpha=0.7,\ \beta_{1}=1,\ \beta_{2}=2,\ r=0.04,\ \mu=0.12,\ \sigma=0.3$. The left panel plots the optimal portfolio as a function of wealth and consumption peak. The middle panel is the two-dimensional projection of the left panel with the boundaries shown. The right panel is the contour plot.}
\label{portgra}
\end{figure}

	\begin{figure}[ht]
\centering
    {%
    \includegraphics[width=.31\linewidth]{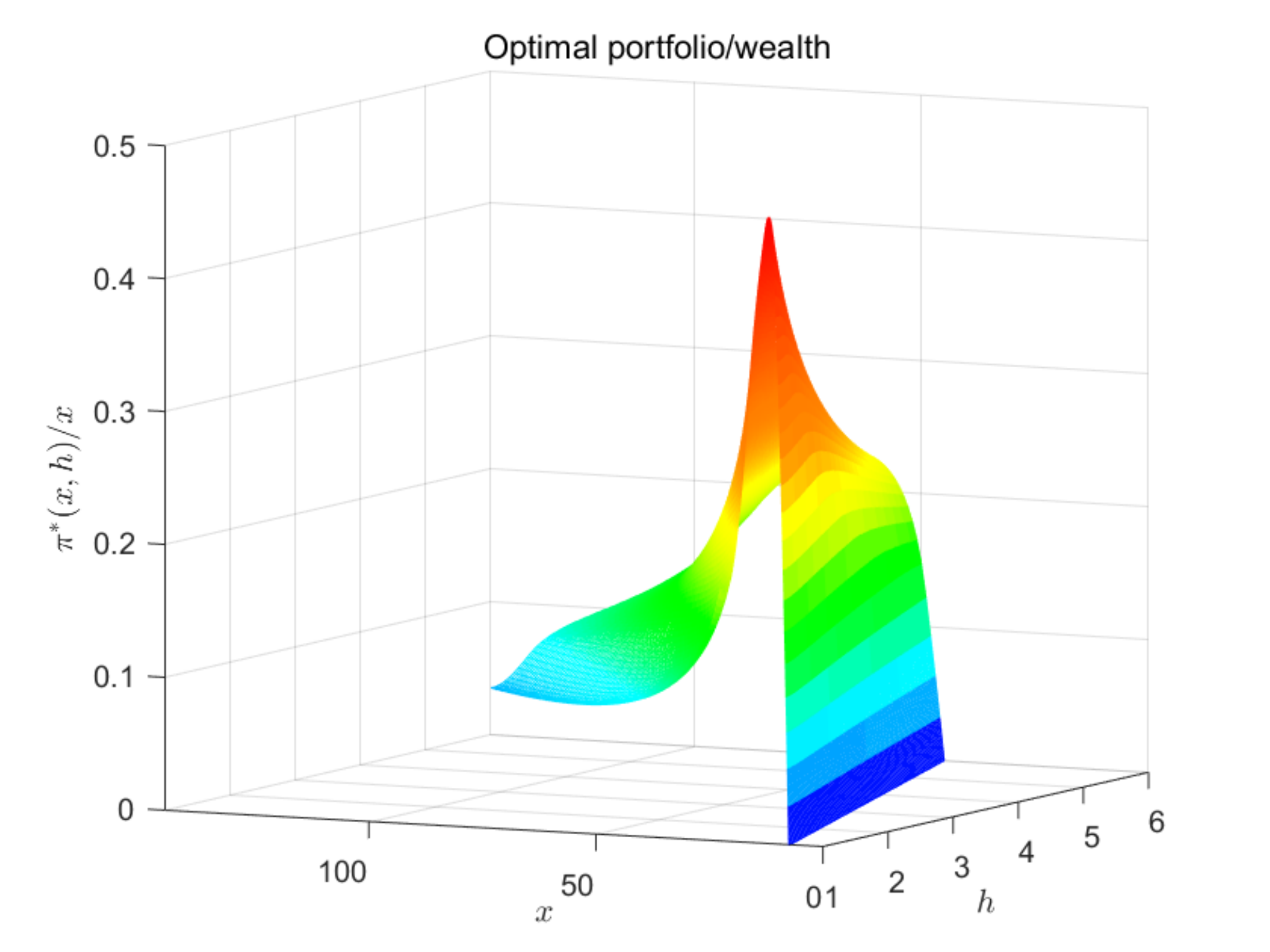}\quad
    \includegraphics[width=.31\linewidth]{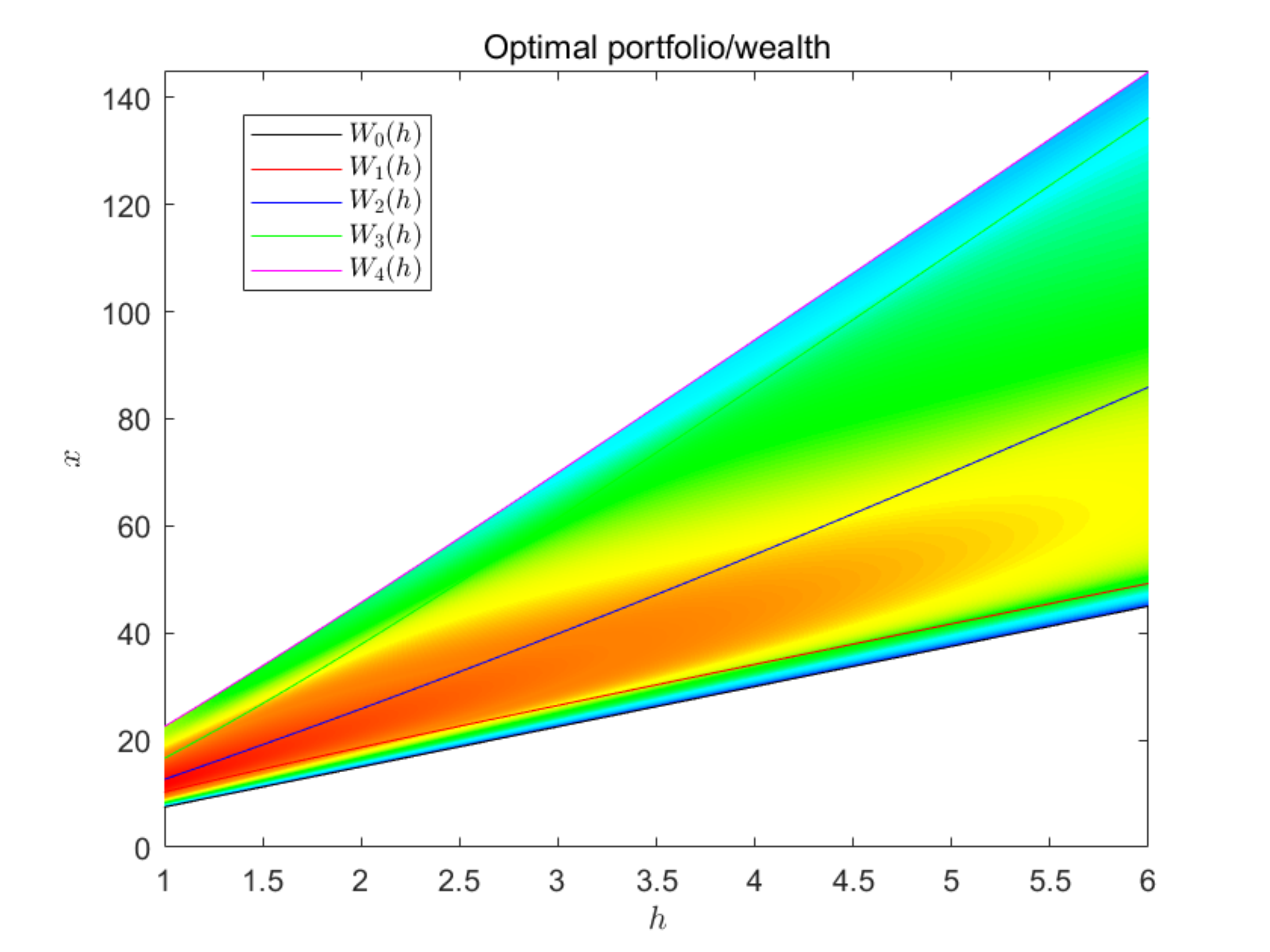}\quad
    \includegraphics[width=.31\linewidth]{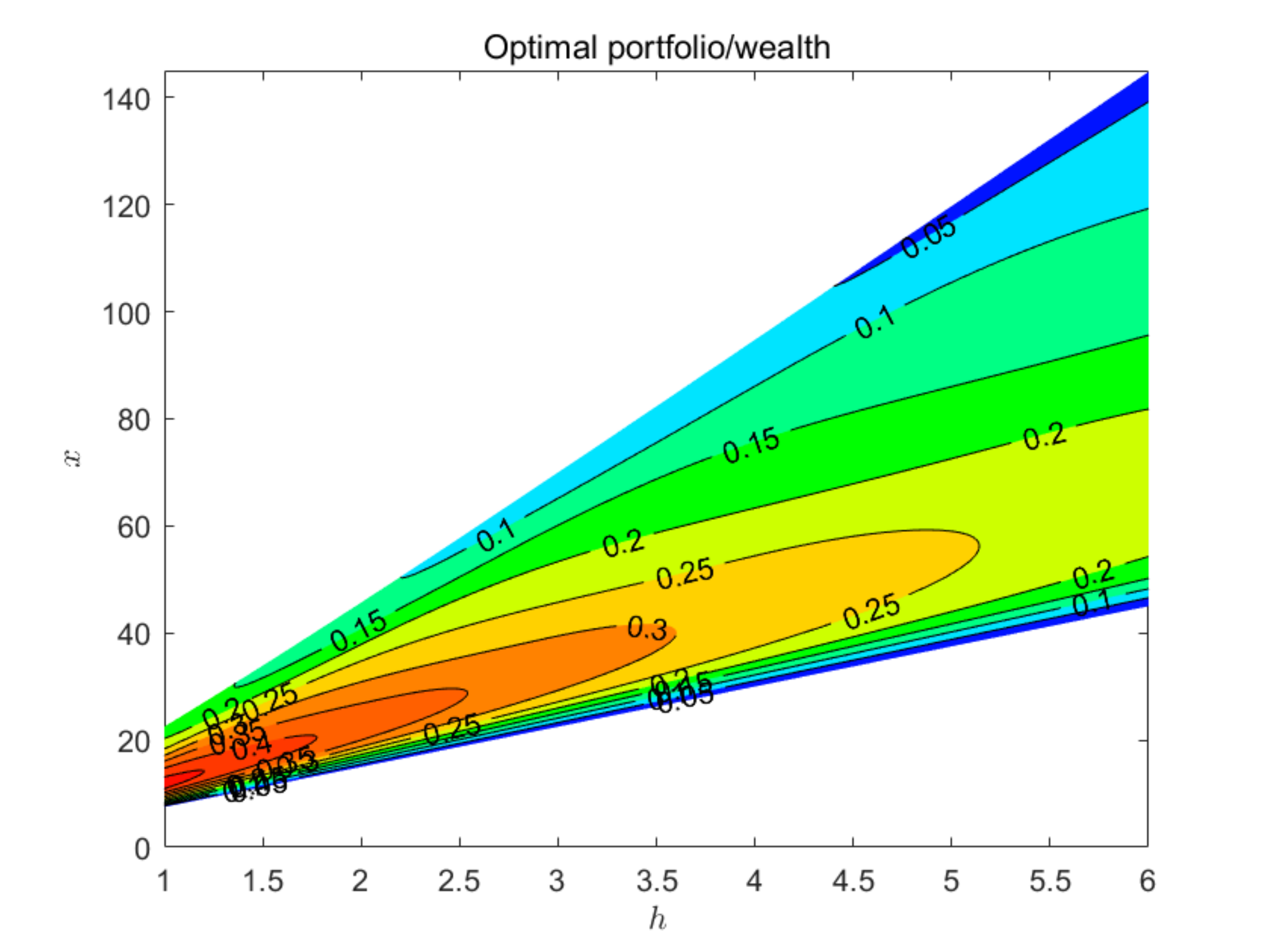}
    }
\caption{Optimal risky investment proportion with fixed parameters $\lambda=0.3,\ \alpha=0.7,\ \beta_{1}=1,\ \beta_{2}=2,\ r=0.04,\ \mu=0.12,\ \sigma=0.3$. The left panel plots the optimal risky investment proportion as a function of wealth and consumption peak. The middle panel is the two-dimensional projection of the left panel with the boundaries shown. The right panel is the contour plot.}
\label{potowe}
\end{figure}

	The optimal portfolio is shown in Figure \ref{portgra}. The behavior of the optimal portfolio varies significantly in different regions. For fixed habit $h$, the optimal portfolio sees a dramatic increase with respect to the variable $x$ in gloom and depression region where the risk aversion is low. However, once crossing $x=W_{\rm ref}(h)$ and the risk aversion shifting to the high level, increasing wealth causes the optimal portfolio to fall instead. The above result indicates that the change of risk aversion has an overwhelming impact on portfolio selection in our model. For those in gloom and depression region, earning money which increases his wealth stimulates him to invest more in risky assert; while for rich people in recovery and satisfactory region, as well as on bliss curve, the more he earns, the less he is willing to invest in risky assert. Similar conclusions can be obtained from the analysis of optimal proportion of wealth invested in risky assets, or economically, optimal portfolio allocation (see Figure \ref{potowe}).
	
	\begin{figure}[ht]
\centering
    {%
    \includegraphics[width=.31\linewidth]{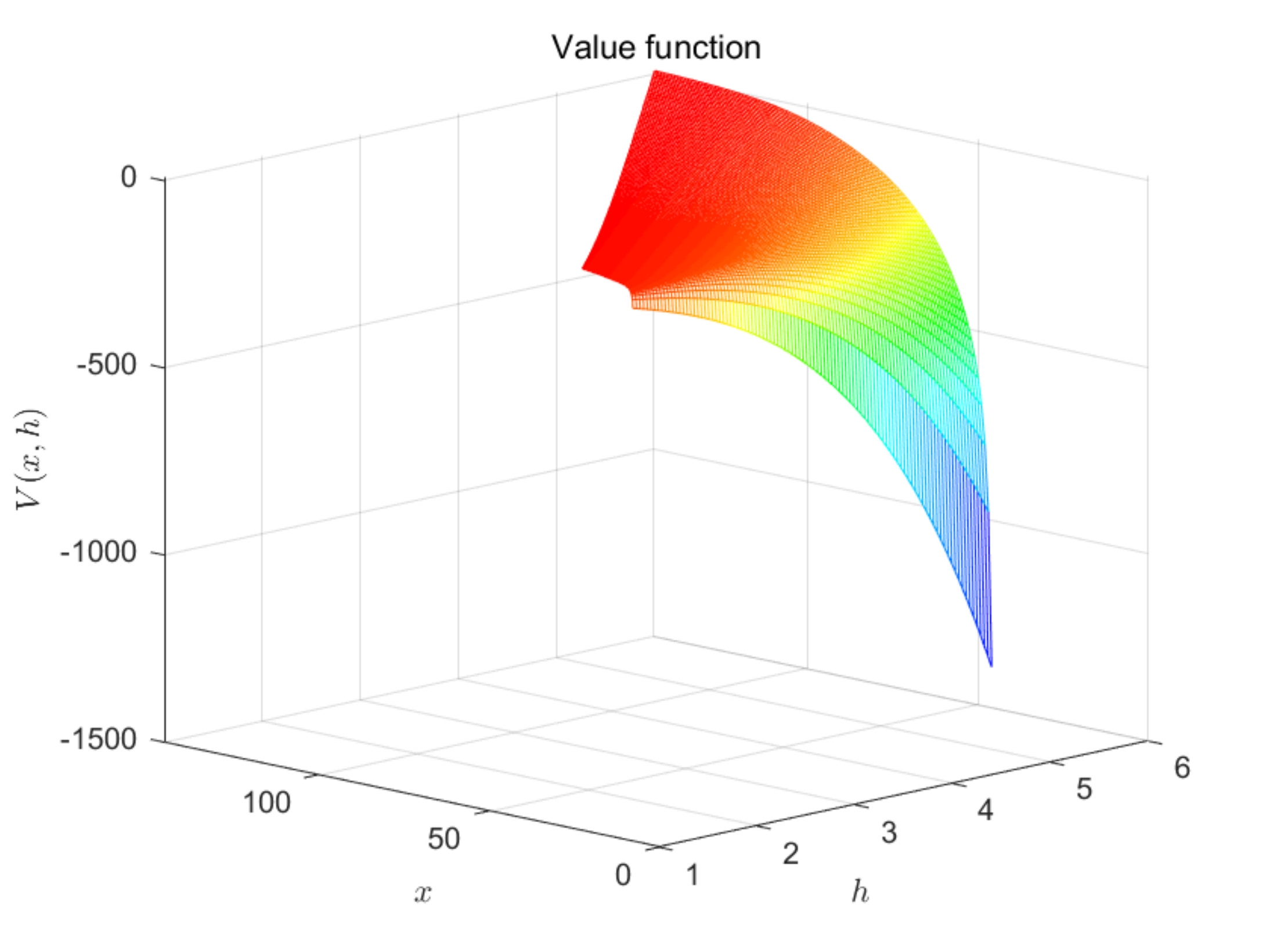}\quad
    \includegraphics[width=.31\linewidth]{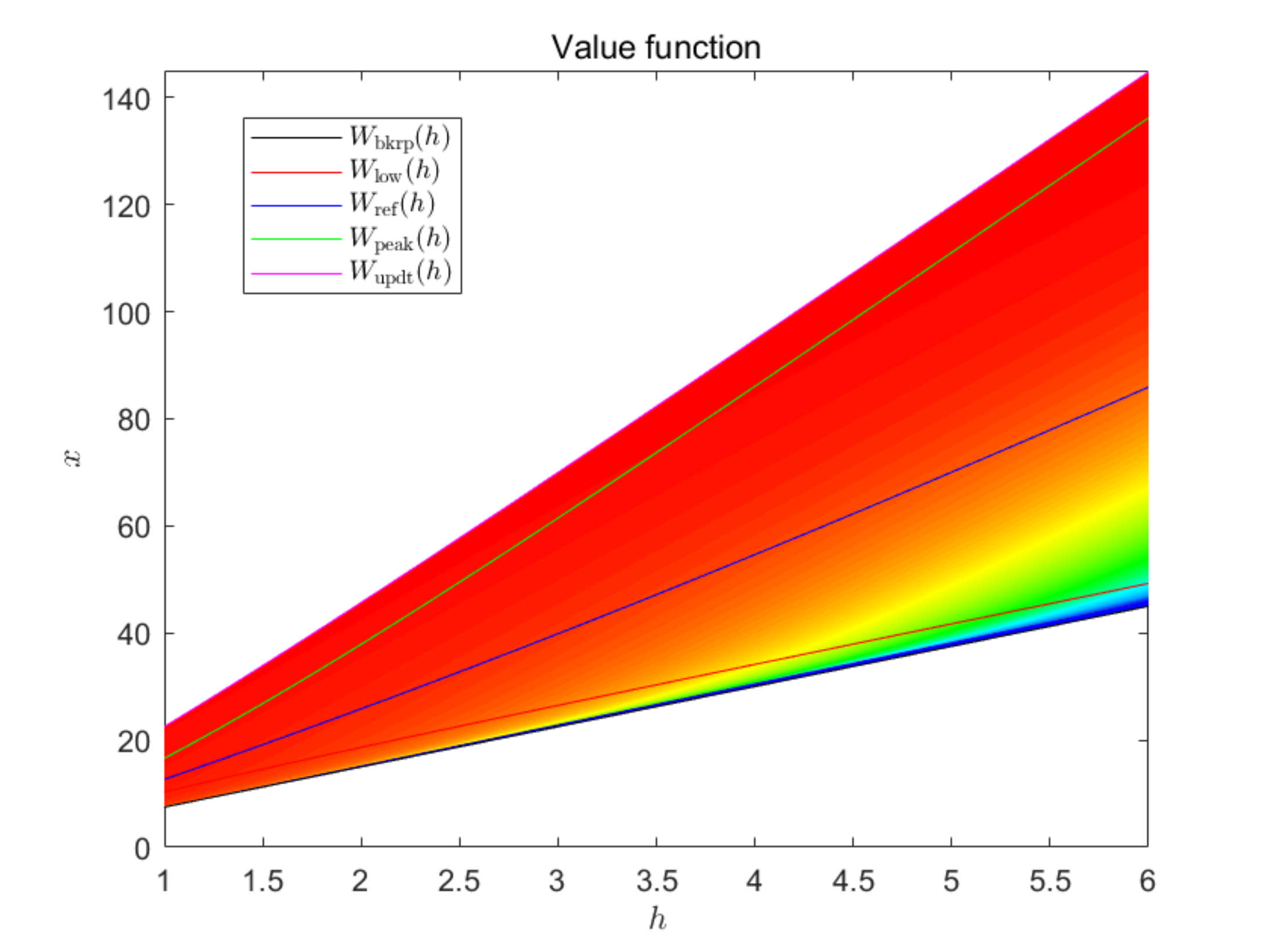}\quad
    \includegraphics[width=.31\linewidth]{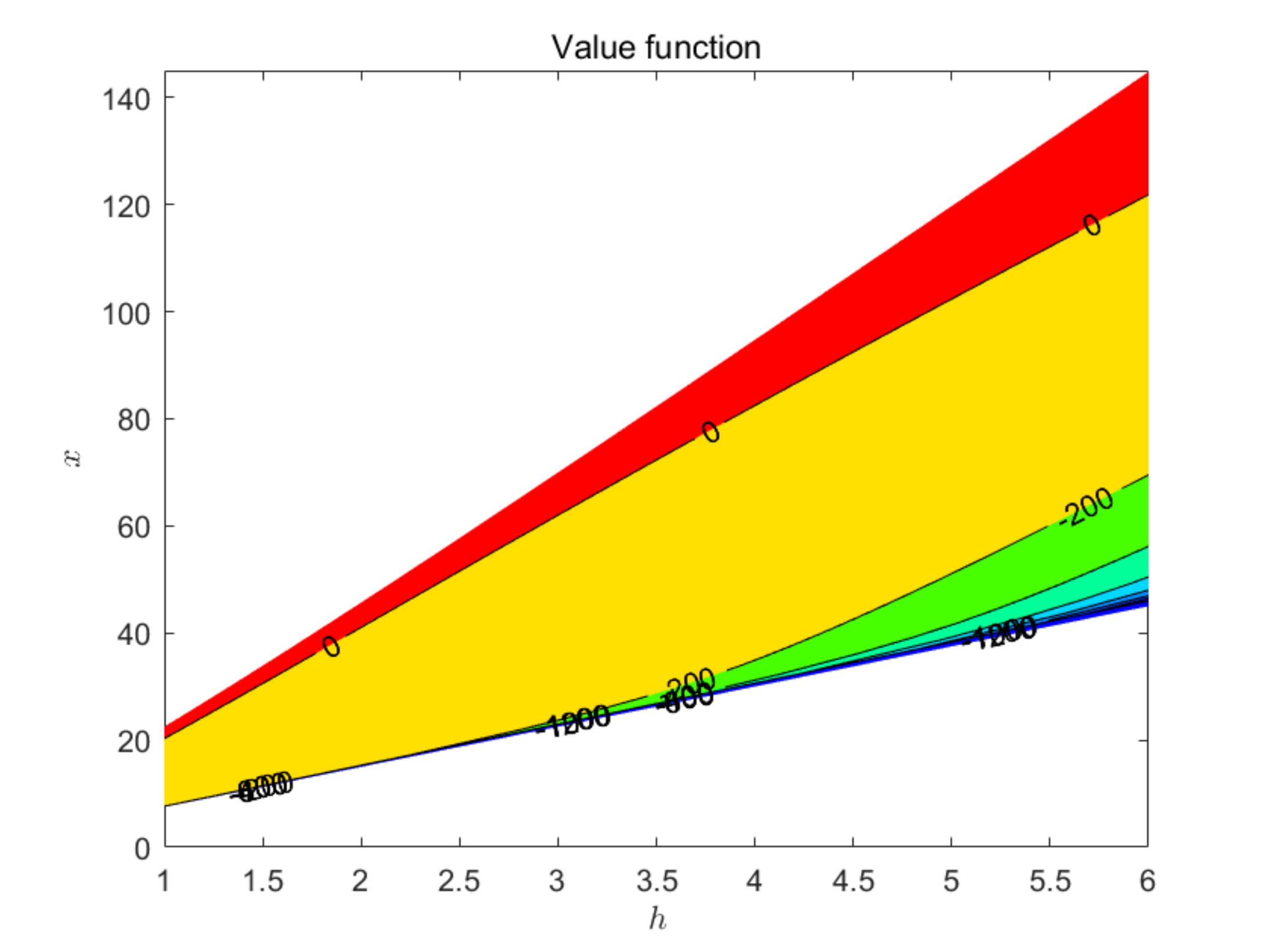}
    }
\caption{Value function with fixed parameters $\lambda=0.3,\ \alpha=0.7,\ \beta_{1}=1,\ \beta_{2}=2,\ r=0.04,\ \mu=0.12,\ \sigma=0.3$. The left panel plots the value function as a function of wealth and consumption peak. The middle panel is the two-dimensional projection of the left panel with the boundaries shown. The right panel is the contour plot.}
\label{valuegra}
\end{figure}

    The value function shown in Figure \ref{valuegra}, is increasing in wealth $x$ and decreasing in habit $h$, which suggests that higher initial wealth and lower inherited past spending maximum result in higher optimal value for Problem (\ref{problem}). Meanwhile, for regions below $x=W_{\rm ref}(h)$ where the risk aversion is low, especially for gloom region, the value function  will fall dramatically due to a slight decrease in $x$ or a slight increase in $h$. Nevertheless, for regions above $x=W_{\rm ref}(h)$ where the risk aversion is high, the value function does not vary significantly as $x$ and $h$ vary. The above result indicates that one can live almost as pleasant as a millionaire with his initial wealth equal to $W_{\rm ref}(h_{0})$ where $h_{0}$ is his historical consumption peak. Another fact shown is that poor people are much more vulnerable to wealth shocks than the wealthy.
	
	\begin{figure}[ht]
\centering
    {%
    \includegraphics[width=.48\linewidth]{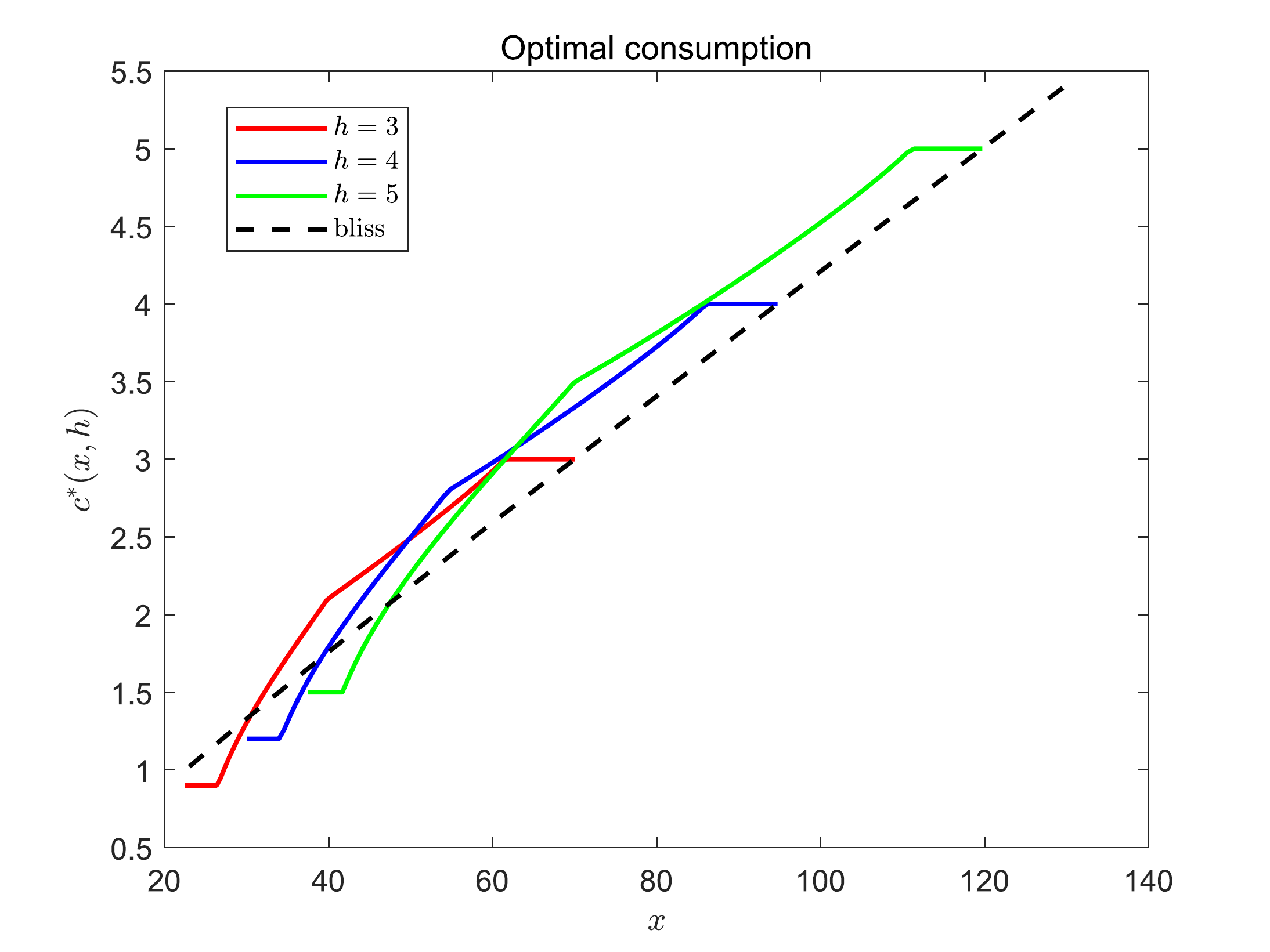}\quad
    \includegraphics[width=.48\linewidth]{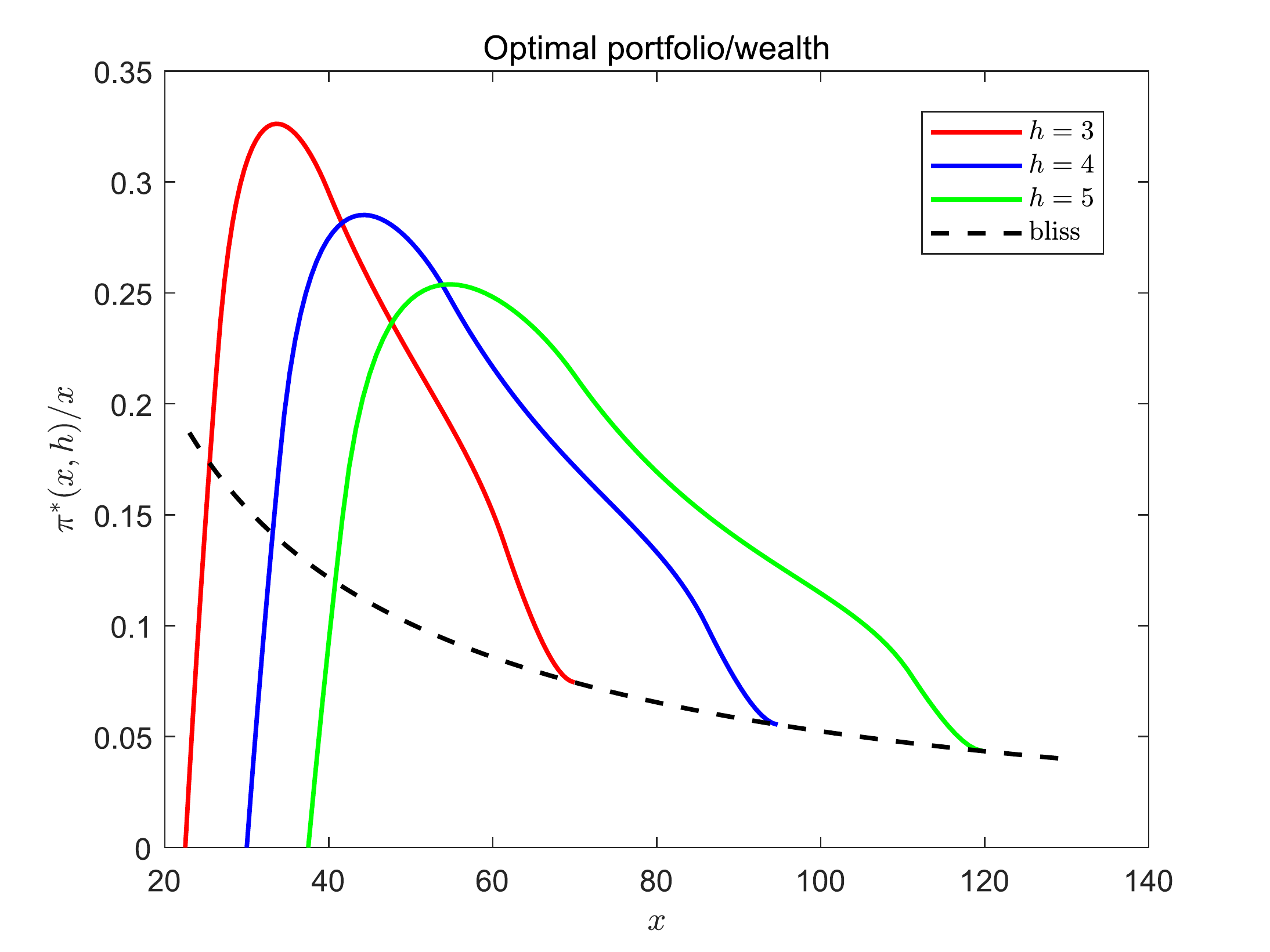}
    }
\caption{Optimal consumption and risky investment proportion in wealth with $h=3,4,5$ and fixed parameters $\lambda=0.3,\ \alpha=0.7,\ \beta_{1}=1,\ \beta_{2}=2,\ r=0.04,\ \mu=0.12,\ \sigma=0.3$. The bliss curve represents the optimal strategies when $x\geq W_{\rm updt}(h)$, and $h$ is updating on this curve.}
\label{difhcpix}
\end{figure}
	We are particularly interested in the decisions of consumption rate and risky investment proportion in terms of wealth, fixing a standard of living, i.e., the function $x\mapsto c^*(x,h)$ and $x\mapsto \pi^*(x,h)/x$, which we present in Figure \ref{difhcpix}. As can be seen, the sensitivity of the optimal consumption with respect to the variable $x$, or economically speaking, the marginal propensity to consume (MPC) out of wealth,  generally decreases with a growing wealth. This is admitted in vast economic literature. However, a dedicated analysis (see Remarks \ref{rmkc1}-\ref{rmkc3}) shows that while MPC out of wealth is indeed decreasing in lower part  of the depression-recovery region and at the bliss curve, it is instead increasing in $x$ in upper part of depression-recovery region. Besides, MPC out of wealth shrinks or swells by $\beta_1/\beta_2$ at $W_{\rm ref}(h)$, which is itself an interesting fact\footnote{Different from other related papers, the marginal utility of consumption is continuous at the reference point $c=\alpha h$, but we still document such an MPC shrink or MPC swell.}. This fact also indicates that the MPC in recovery region can be globally lower or higher than that in depression region, and it is lower in our numerical result (left panel of Figure \ref{difhcpix}). See Figure \ref{appendixfigure} in the introduction for a more illustrative version of optimal consumption when fixing $h$.
\begin{remark}\label{rmkc1}
	    Based on Theorem \ref{optpolicy} and Lemma \ref{fc1}, we have for $W_{\rm low}(h)<x\le W_{\rm peak}(h)$,
	    
	    \begin{equation*}
	        \frac{\partial c^{*}(x,h)}{\partial x}=\left\{
		\begin{array}{l}
		\frac{1}{\beta_{1}}\frac{1}{f_{2}(x,h)\tilde{V}_{yy}(f_{2}(x,h),h)},\ W_{\rm low}(h)<x\le W_{\rm ref}(h),\\
		\frac{1}{\beta_{2}}\frac{1}{f_{3}(x,h)\tilde{V}_{yy}(f_{3}(x,h),h)},\ W_{\rm ref}(h)<x\le W_{\rm peak}(h).\\
	\end{array}
	\right.
	    \end{equation*}
	    $\frac{r}{k}(C_{3}(h)+C_{4}(h))+\frac{1}{\gamma\beta_{1}}=\frac{r}{k}(C_{5}(h)+C_{6}(h))+\frac{1}{\gamma\beta_{2}}$ leads to the continuity of $y\tilde{V}_{yy}(y,h)$ at $y=1$. Hence $\frac{\partial c^{*}(x,h)}{\partial x}|_{x\rightarrow W_{\rm ref}(h)^{+}}=\frac{\beta_{1}}{\beta_{2}}\frac{\partial c^{*}(x,h)}{\partial x}|_{x\rightarrow W_{\rm ref}(h)^{-}}$, which indicates that the MPC out of wealth shrinks or swells by $\frac{\beta_{1}}{\beta_{2}}$ when exceeding $W_{\rm ref}(h)$.
	\end{remark}
	
	\begin{remark}\label{rmkc2}
	The bliss curve for optimal consumption is $c=W_{\rm updt}^{-1}(x)$. Hence the bliss curve is concave if and only if $W''_{\rm updt}(h)>0$. Direct computation shows
	
	\begin{equation*}
	    W''_{\rm updt}(h)=e^{-(1-\alpha)(q_{2}-1)\beta_{2}h}\Big[M_{1}e^{-(\alpha-\lambda)(q_{2}-1)\beta_{1}\big]h}-M_{2}\Big]
	\end{equation*}
	where 
	
	\begin{align*}
	    M_{1}=&\frac{\frac{k}{\gamma^{2}\beta_{1}}(1-q_{1})(q_{2}-1)^{2}(1-\alpha)^{q_{2}-1}\big[(1-\alpha)q_{2}\beta_{2}+(\alpha-\lambda)(q_{2}-1)\beta_{1}\big]}{(1-\alpha)(q_{2}-q_{1})\beta_{2}
+(\alpha-\lambda)(q_{2}-1)\beta_{1}}\Big[(1-\alpha)\beta_{2}+(\alpha-\lambda)\beta_{1}\Big]^{2},\\
M_{2}=&\frac{k}{\gamma^{2}}\frac{\beta_{2}-\beta_{1}}{\beta_{1}\beta_{2}}\frac{1-q_{1}}{q_{2}-q_{1}}q_{2}(1-\alpha)^{q_{2}+1}(q_{2}-1)^{2}\beta_{2}^{2}.
	\end{align*}
	
	Hence for $\beta_{1}<\beta_{2}$, we have $M_{1},M_{2}>0$ and thus there exists $\bar{h}=\frac{\ln(M_{1})-\ln(M_{2})}{(\alpha-\lambda)(q_{2}-1)\beta_{1}}$ such that the bliss curve is concave in $x$ for $h\le\bar{h}$. It implies that the MPC out of wealth decreases when $x\ge W_{\rm updt}(h)$ and $h\le\bar{h}$.  With current parameters, the threshold $\bar{h}$ is approximately $6.6$.
	
	For $\beta_{1}\ge\beta_{2}$, we have $M_{1}>0,M_{2}\le 0$ and thus the bliss curve is concave in $x$, which implies that the MPC out of wealth decreases when $x\ge W_{\rm updt}(h)$.
	
	\end{remark}
	
\begin{remark}\label{rmkc3}
	     With  Lemma \ref{vyypo}, Theorem \ref{optpolicy} and Lemma \ref{fc1}, it can be shown that $\frac{\partial c^{*}(x,h)}{\partial x}$ is increasing (decreasing) in $x$ if and only if $\pi^{*}(y,h)$ is increasing (decreasing) in $y$, where $x$ and $y$ are connected by $y=f(x,h)$. As  for $e^{-(1-\alpha)\beta_{2}h}\le y<e^{(\alpha-\lambda)\beta_{1}h}$,
	     
	    \begin{equation*}
	        \frac{\partial \pi^{*}(y,h)}{\partial y}=\frac{\mu-r}{\sigma^{2}}\frac{r}{k}\left\{
		\begin{array}{l}
		C_{3}(h)(q_{1}-1)y^{q_{1}-2}+C_{4}(h)(q_{2}-1)y^{q_{2}-2},\ 1\le y<e^{(\alpha-\lambda)\beta_{1}h},\\
		C_{5}(h)(q_{1}-1)y^{q_{1}-2}+C_{6}(h)(q_{2}-1)y^{q_{2}-2},\ e^{-(1-\alpha)\beta_{2}h}\le y<1\\
	\end{array}
	\right.
	    \end{equation*}
	    and 
	    
	    \begin{align*}
	        &C_{3}(h)(q_{1}-1)e^{(q_{1}-2)(\alpha-\lambda)\beta_{1}h}+C_{4}(h)(q_{2}-1)e^{(q_{2}-2)(\alpha-\lambda)\beta_{1}h}=C_{1}(h)(q_{1}-1)e^{(q_{1}-2)(\alpha-\lambda)\beta_{1}h}<0,\\
	        &C_{3}(h)(q_{1}-1)+C_{4}(h)(q_{2}-1)=C_{5}(h)(q_{1}-1)+C_{6}(h)(q_{2}-1),
	    \end{align*}
	     we conclude that if
	    \begin{equation}
	    \label{rmk3eq}
	        C_{5}(h)(q_{1}-1)e^{-(q_{1}-2)(1-\alpha)\beta_{2}h}+C_{6}(h)(q_{2}-1)e^{-(q_{2}-2)(1-\alpha)\beta_{2}h}>0,
	    \end{equation}
	    then there exists $\bar{x}(h)\in [W_{\rm low}(h),W_{\rm peak}(h)]$ such that the MPC out of wealth is decreasing for $x\in\big[W_{\rm low}(h),\bar{x}(h)\big]\backslash\{W_{\rm ref}(h)\}$ and increasing for $x\in \big[\bar{x}(h),W_{\rm peak}(h)\big]\backslash\{W_{\rm ref}(h)\}$ ($W_{\rm ref}(h)$ needs to be excluded due to the MPC shrink or swell, see Remark \ref{rmkc1}). Specifically, $\bar{x}(h)=-\tilde{V}_{y}(\bar{y}(h),h)$ with $\bar{y}(h)$ given by
	    
	    \begin{equation*}
	    \bar{y}(h)=\left\{
		\begin{array}{l}
		(-\frac{C_{3}(h)(q_{1}-1)}{C_{4}(h)(q_{2}-1)})^{\frac{1}{q_{2}-q_{1}}},\ {\rm for}\ h \ {\rm s.t.}\ C_{3}(h)(q_{1}-1)+C_{4}(h)(q_{2}-1)>0,\\
		(-\frac{C_{5}(h)(q_{1}-1)}{C_{6}(h)(q_{2}-1)})^{\frac{1}{q_{2}-q_{1}}},\ {\rm for}\ h \ {\rm s.t.}\ C_{3}(h)(q_{1}-1)+C_{4}(h)(q_{2}-1)\le 0\\
	\end{array}
	\right.
	    \end{equation*}
	    It is straight forward to verify numerically that (\ref{rmk3eq}) is satisfied for $h$ in a reasonable range (say, $h\geq 0.2$ with the current parameters).
	    
	    The threshold $\bar{x}(h)$ is above $W_{\rm ref}(h)$ if and only if $C_{3}(h)(q_{1}-1)+C_{4}(h)(q_{2}-1)<0$. We observe two cases: One case is that $\bar{x}(h)$ is always above $W_{\rm ref}(h)$, which is the case for $\beta_{1}\ge\beta_{2}$; the other case is that $\bar{x}(h)$ is above $W_{\rm ref}(h)$ for $h$ above certain threshold $\hat{h}$ and below $W_{\rm ref}(h)$ for $h<\hat{h}$, which is the case for $\beta_{1}<\beta_{2}$. Moreover, as the difference $\beta_{2}-\beta_{1}$ decreases to zero, the threshold $\hat{h}$ increases to $+\infty$ and turns into the first case.
	\end{remark}

	From the right panel of Figure \ref{difhcpix}, it is clear that for sampled $h$, risky investment proportion is a hump in variable $x$, and the peak is around $W_{\rm ref}(h)$ (see also Figure \ref{potowe}) \footnote{Numerical analysis shows that risky investment proportion is indeed a hump in $x$ for $h$ above a relatively small level. For $h$ below that level (i.e, if $h$ is extremely small), there is an apparent increase on the right end of the hump.}. If we interpret the inverse of risky investment proportion as the so-called {\it implied relative risk aversion} (\cite{Jeon2020}), we find that it will be a smile in wealth: people with intermediate wealth level have lowest risk aversion, and hence have the highest risk tolerance. People with either very low or very high level of wealth are much more risk averse. This effect comes intuitively from our model settings. Poor people need to make their deposit above $W_{\rm bkrp}(h)$ in order to satisfy the lowest consumption constraint, hence are very sensitive to risk. It is reasonable for them to keep most part of their wealth in safe assets. Rich people, on the other hand, have already been satisfied by the current level of consumption (or even continuously consume more and more) and they tend to avoid the risk of consumption declining to less than reference $\alpha h$. A more illustrative version of optimal risky investment proportion when fixing $h$ can be found in Figure \ref{appendixfigure}.
	
\begin{figure}[ht]
    \centering
    {%
    \includegraphics[width=.48\linewidth]{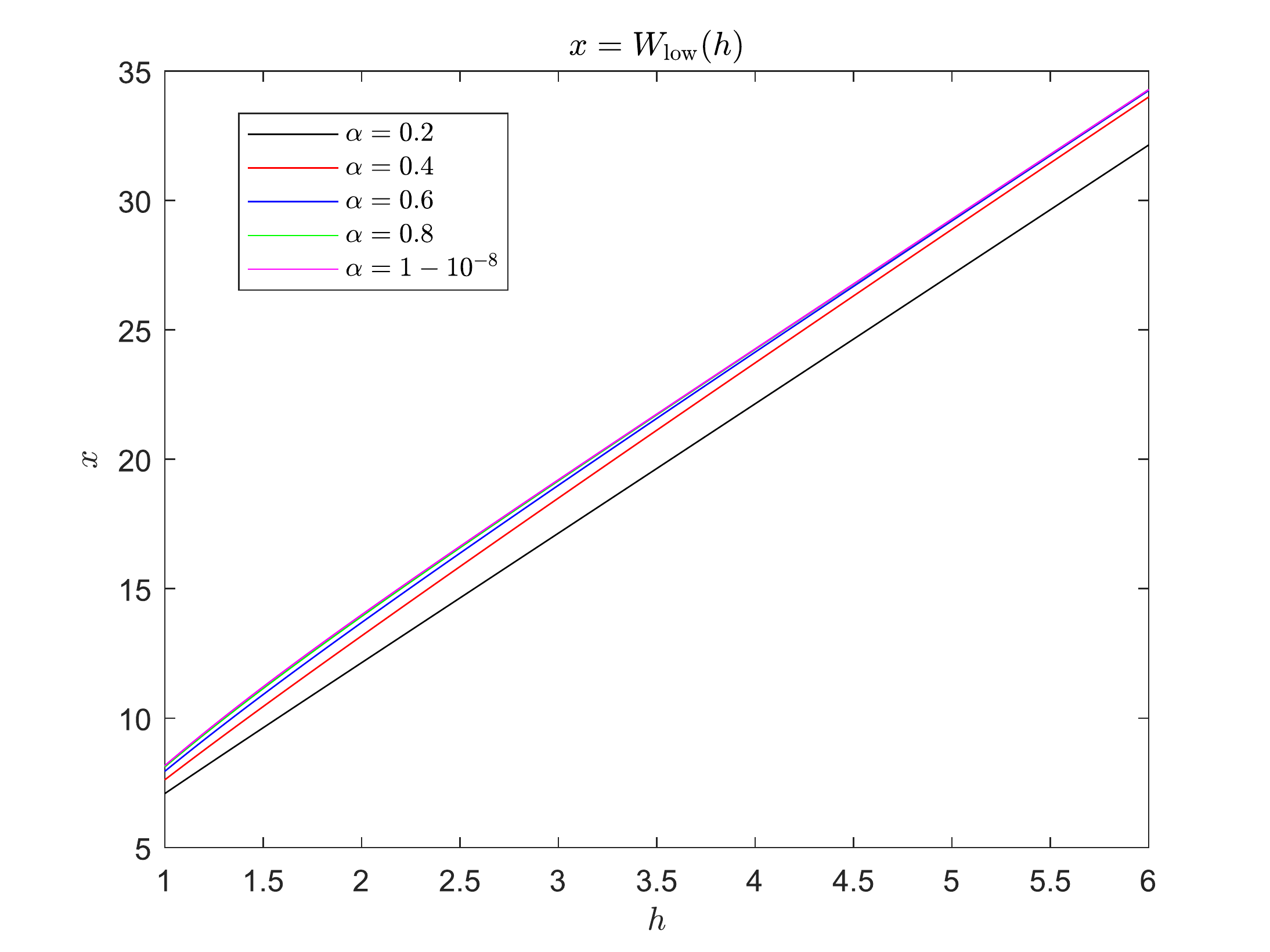}\quad
    \includegraphics[width=.48\linewidth]{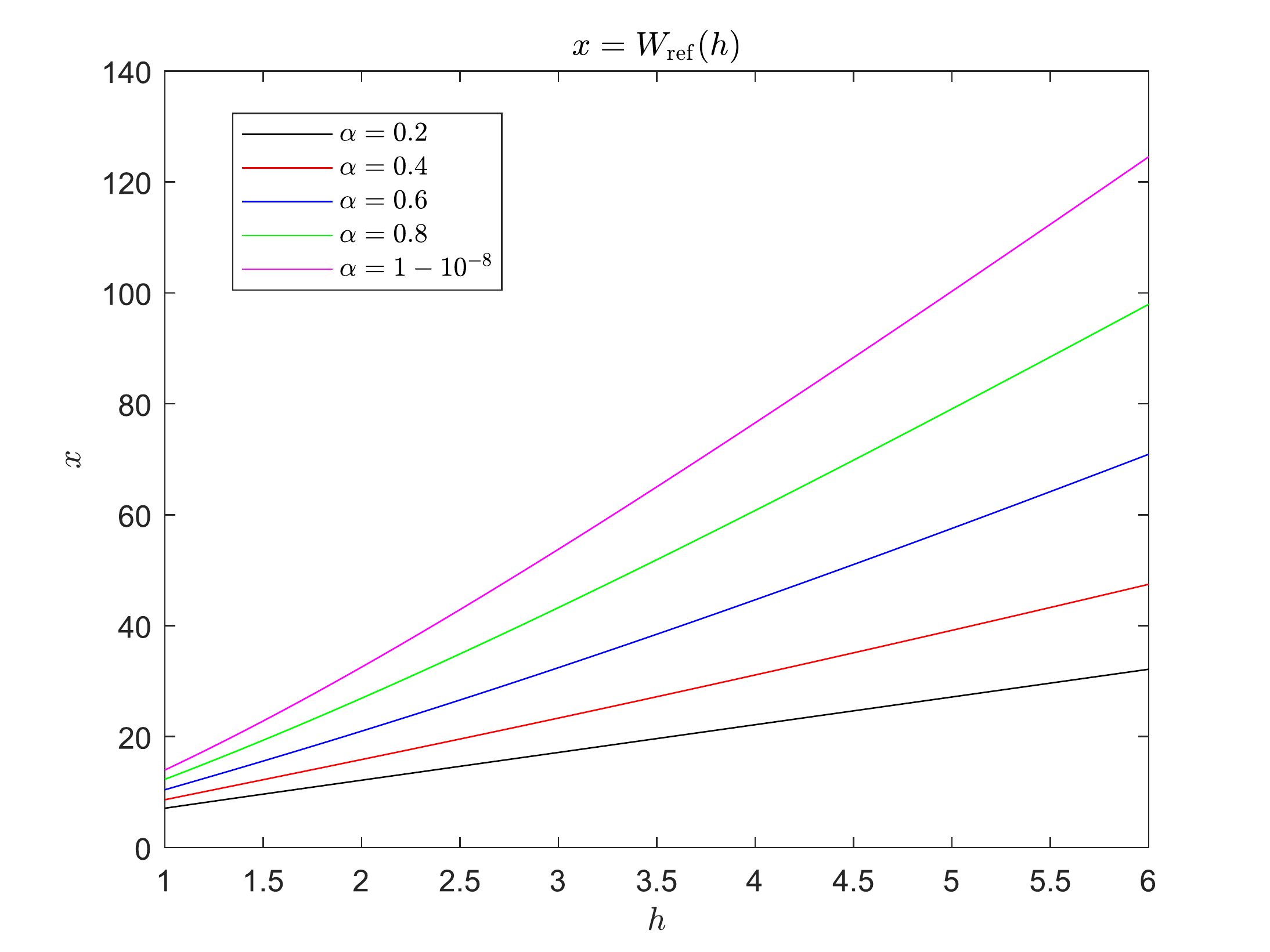}
    }
    {%
    \includegraphics[width=.48\linewidth]{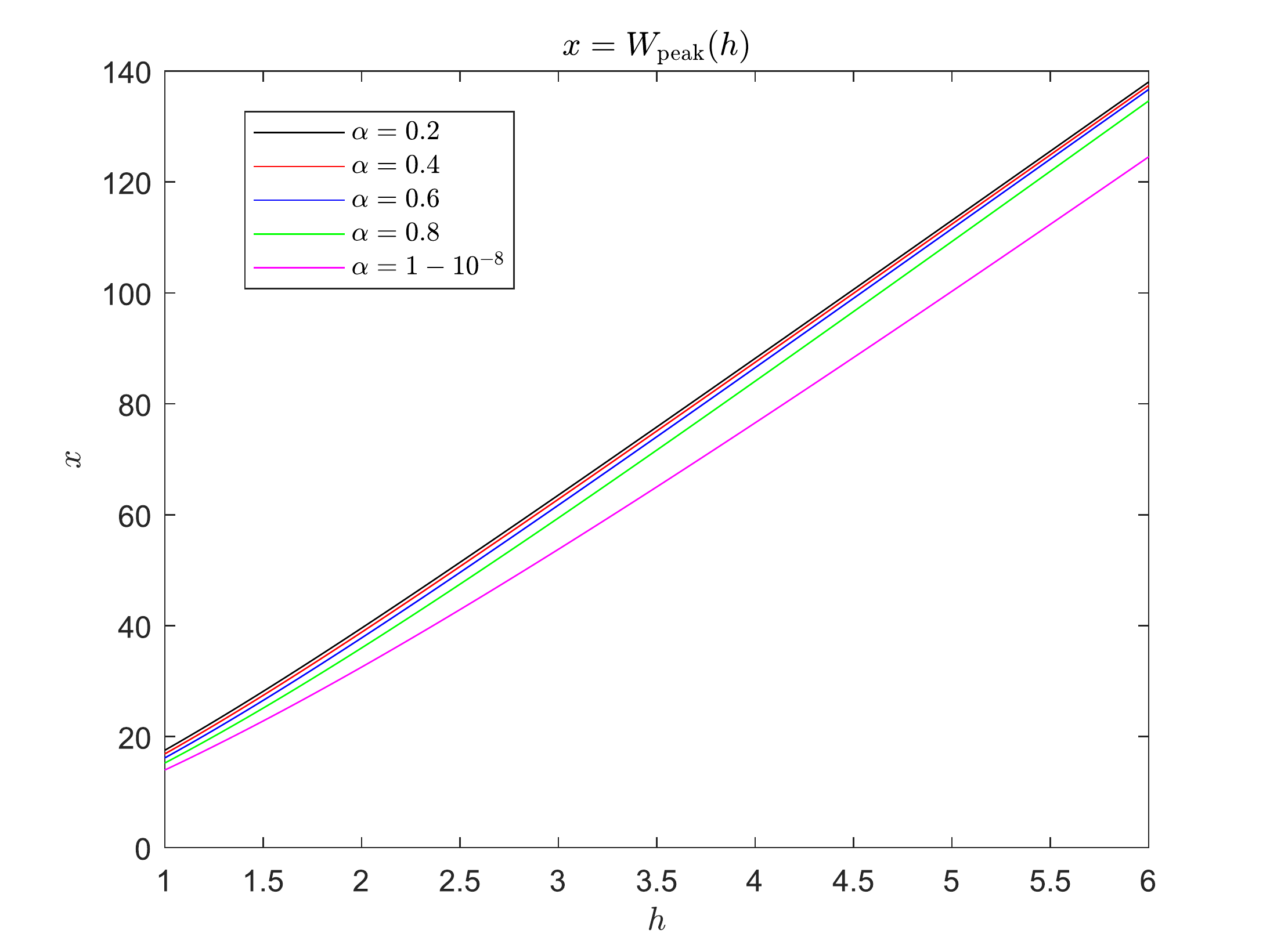}\quad
    \includegraphics[width=.48\linewidth]{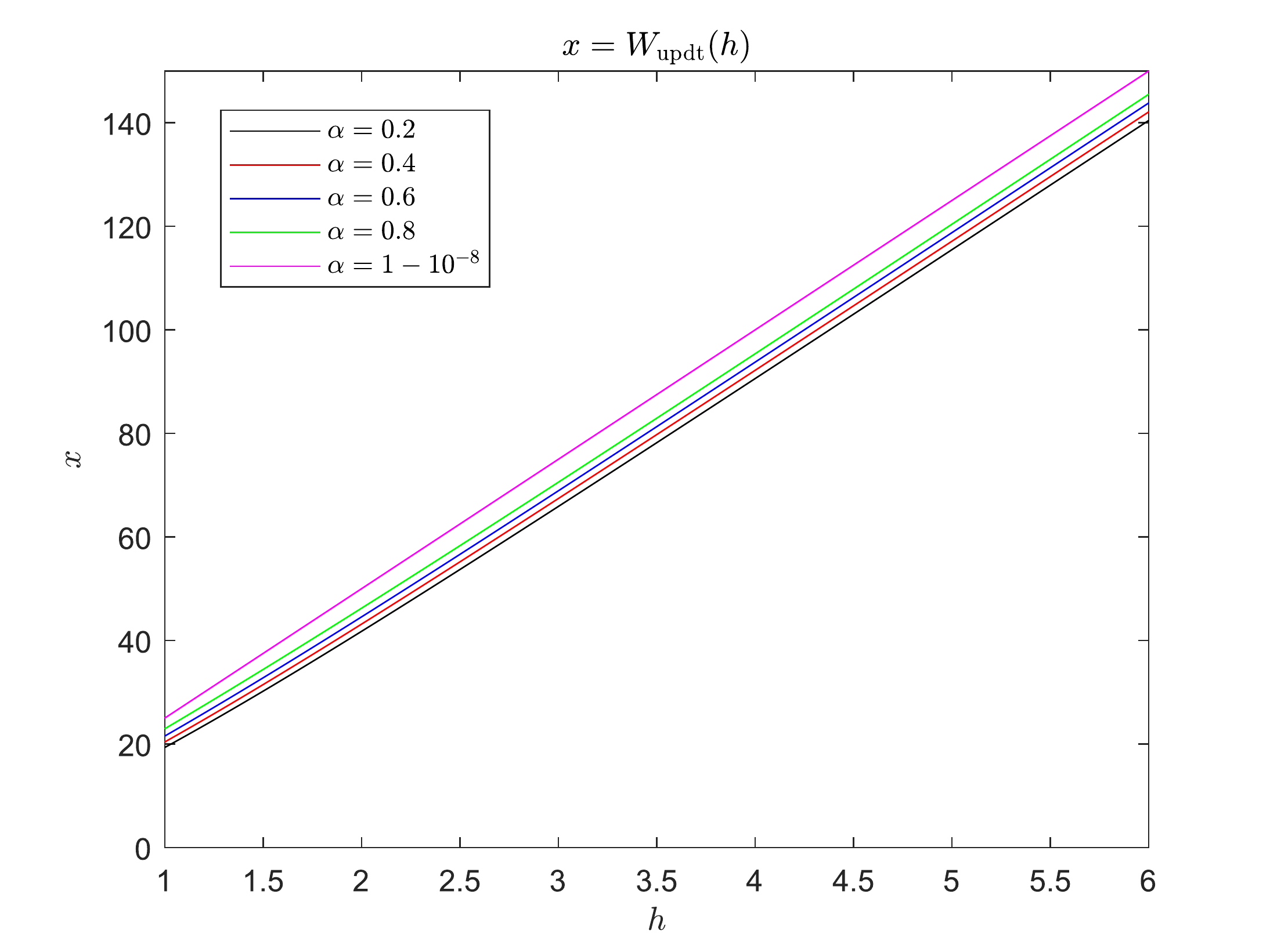}
    }
    \caption{With fixed parameters $\lambda=0.2,\ \beta_{1}=1,\ \beta_{2}=2,\ r=0.04,\ \mu=0.12,\ \sigma=0.3$, the four boundaries varies with different $\alpha$.}
    \label{alphatoboundary}
\end{figure}

\section{Sensitivity Analysis}
\label{sean}
\subsection{Impact of $\lambda$ and $\alpha$ on Thresholds}
	The parameter $\alpha$ determines the reference point $\alpha h$ where risk aversion changes. When $\alpha$ equals $\lambda$, our model reduces to the model without the risk aversion related reference point and the boundary $x=W_{\rm ref}(h)$ coincides with $x=W_{\rm low}(h)$. When $\alpha$ approaches $1$ from below, the reference point approaches the running maximum and the boundary $x=W_{\rm ref}(h)$ coincides with $x=W_{\rm peak}(h)$. To figure out the consequences of the ratio $\alpha$ on different boundaries, we fix $\lambda=0.2,\ \beta_{1}=1,\ \beta_{2}=2,\ r=0.04,\ \mu=0.12,\ \sigma=0.3$, let $\alpha$ vary from $\lambda$ to $1-10^{-8}$ and present the boundaries separately (we do not consider the boundary $W_{\rm bkrp}(h)$ since it does not depend on $\alpha$). As shown in Figure \ref{alphatoboundary}, the boundary $x=W_{\rm peak}(h)$ moves downward as $\alpha$ increases, while the other three boundaries $x=W_{\rm low}(h),W_{\rm ref}(h)$ and $W_{\rm updt}(h)$ have the tendency to move upward as $\alpha$ increases. The above phenomenon illustrates that the criterion for the investor to be able to consume at the historical running maximum level lowers as $\alpha$ increases while other three criteria to enter a higher consumption level would generally raise up as $\alpha$ increases. The change of $\alpha$ is most influential on $x=W_{\rm ref}(h)$, which is the boundary of whether to consume more than $\alpha h$ and an important boundary in our analysis of both optimal portfolio and value function.
	\begin{figure}[ht]
    \centering
    {%
    \includegraphics[width=.48\linewidth]{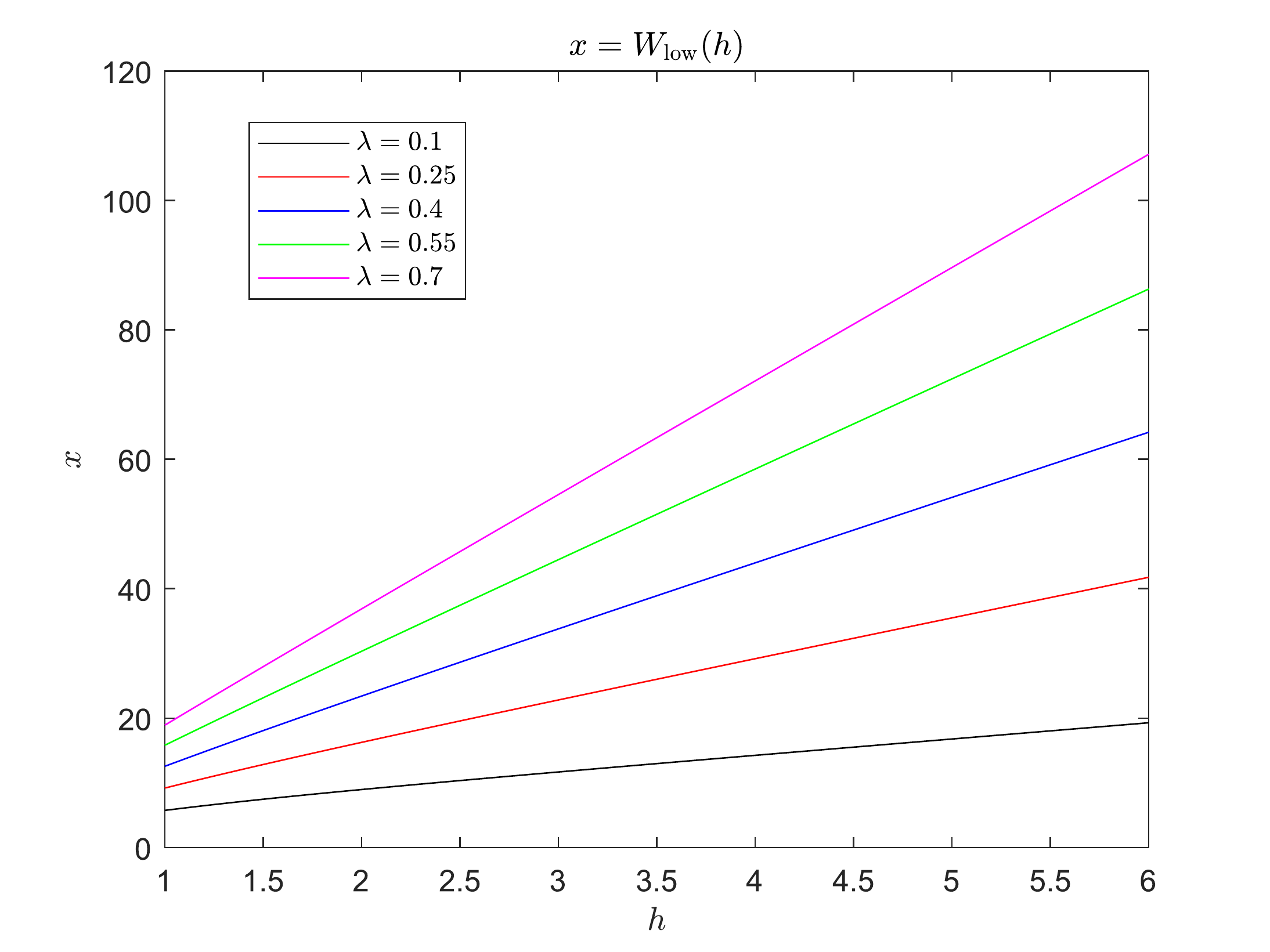}\quad
    \includegraphics[width=.48\linewidth]{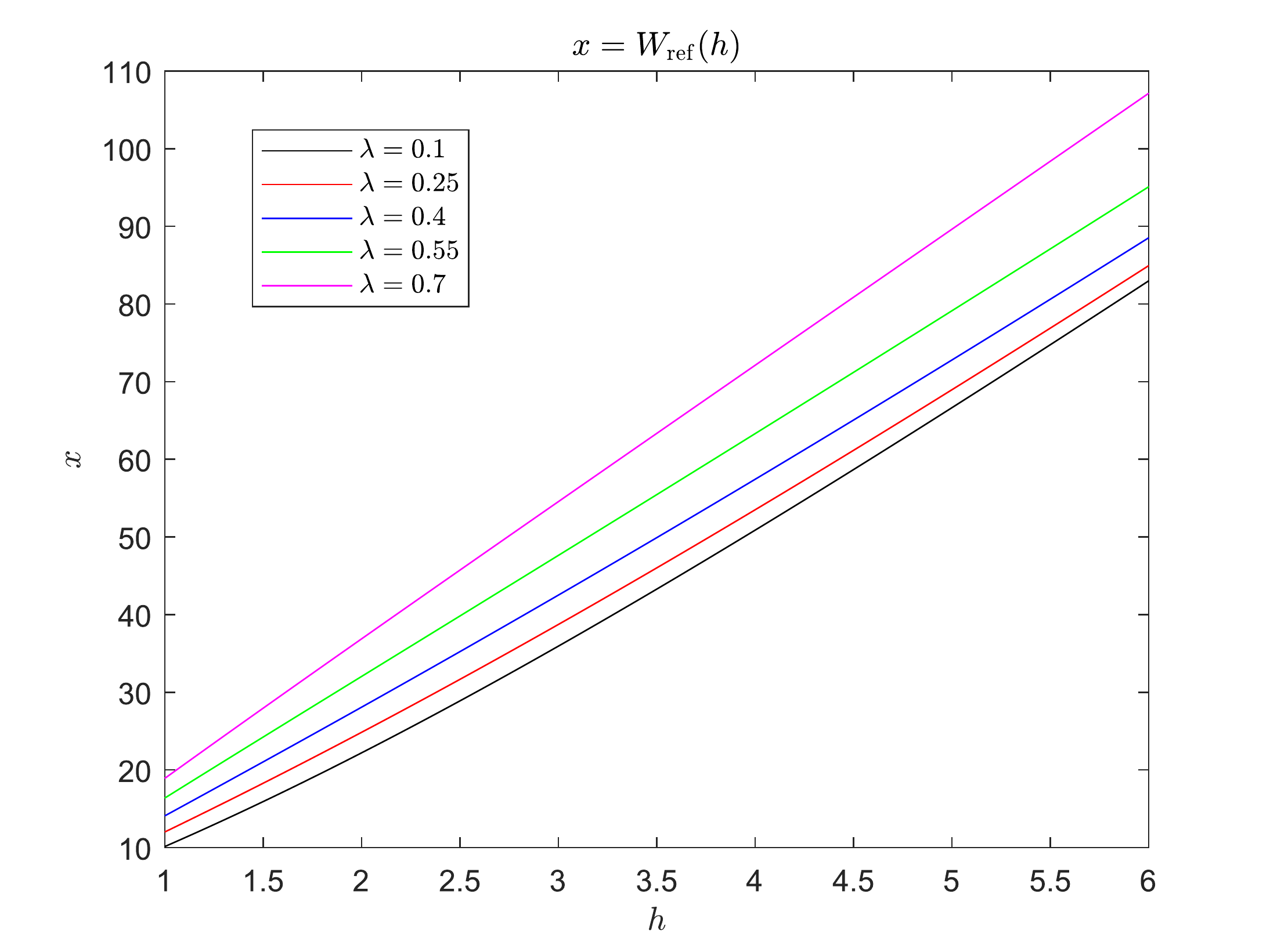}
    }
    {%
    \includegraphics[width=.48\linewidth]{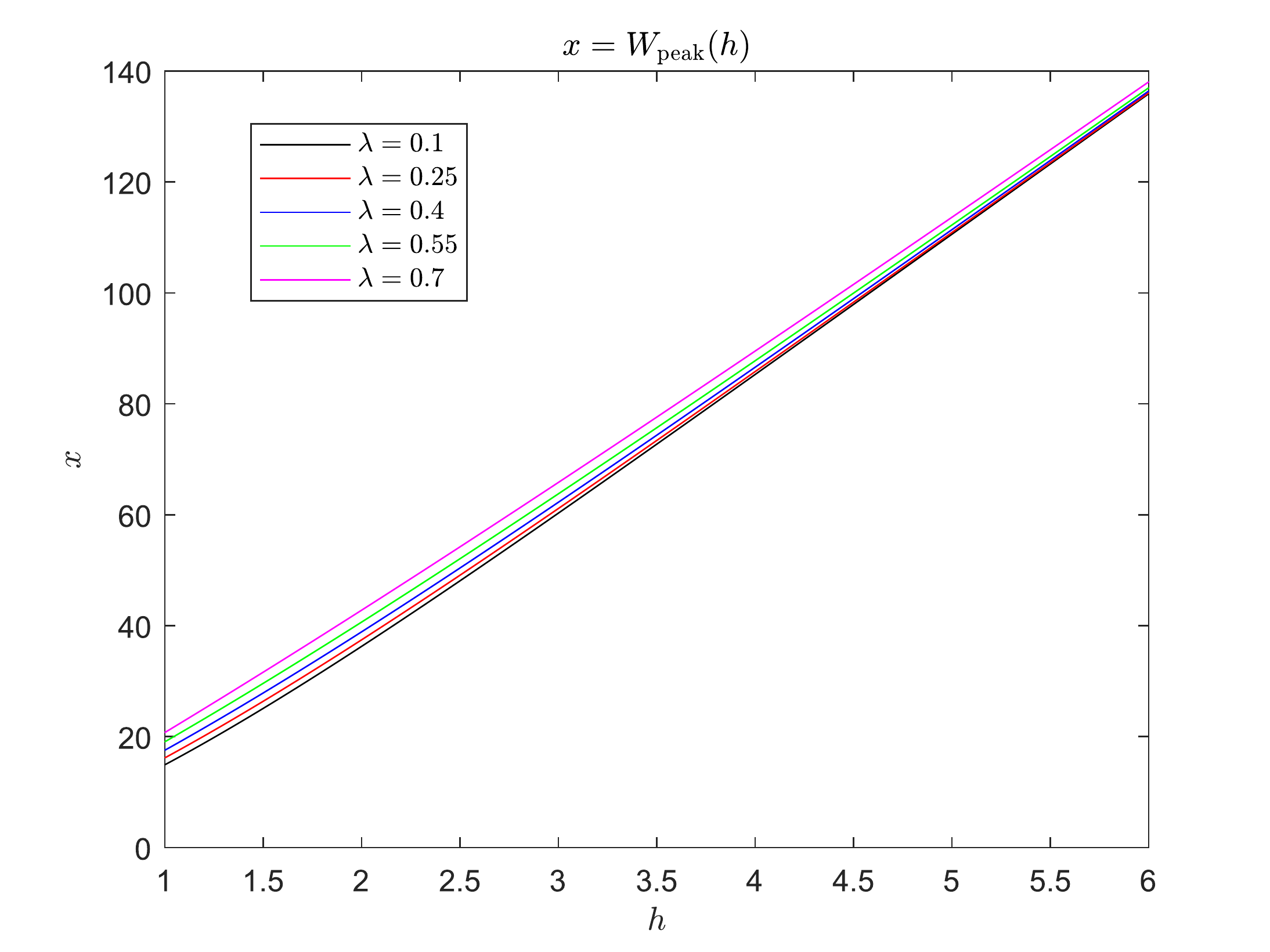}\quad
    \includegraphics[width=.48\linewidth]{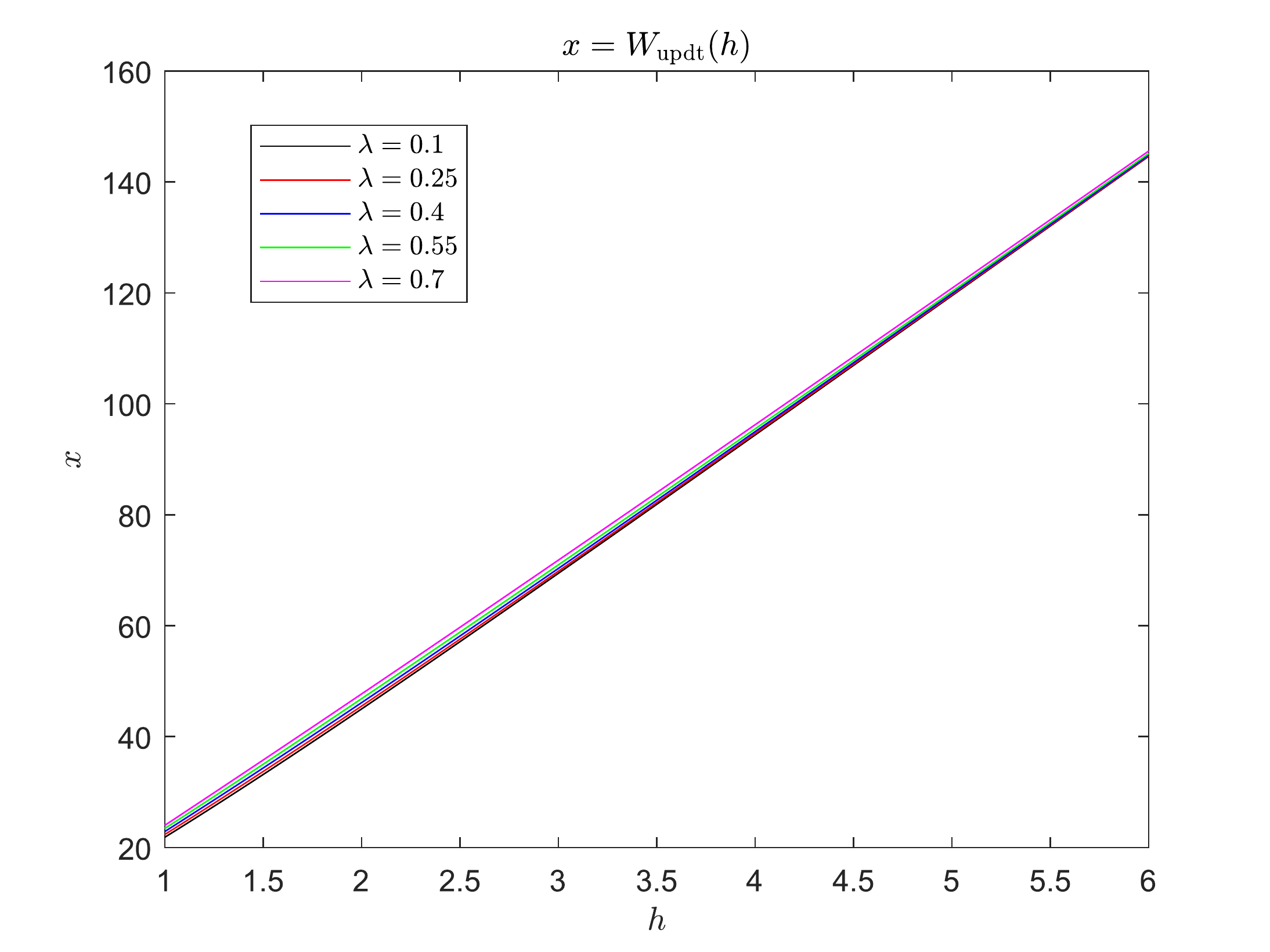}
    }
    \caption{With fixed parameters $\alpha=0.7,\ \beta_{1}=1,\ \beta_{2}=2,\ r=0.04,\ \mu=0.12,\ \sigma=0.3$, the four boundaries varies with different $\lambda$.}
    \label{lambdatoboundary}
\end{figure}

Another parameter of interest in our model is $\lambda$, which reflects the degree of the drawdown constraint on consumption. The consumption level is allowed to be relatively lower with smaller $\lambda$. When $\lambda$ equals $0$, our model reduces to the model without drawdown constraint and the boundary $x=W_{\rm bkrp}(h)$ coincides with the $h$ axis. When $\lambda$ equals $\alpha$, our model reduces to the model without the risk aversion related reference point and the boundary $x=W_{\rm ref}(h)$ coincides with $x=W_{\rm low}(h)$. We fix $\alpha=0.7,\ \beta_{1}=1,\ \beta_{2}=2,\ r=0.04,\ \mu=0.12,\ \sigma=0.3$, let $\lambda$ vary below $\alpha$ and graph the boundaries separately (boundary $x=W_{\rm bkrp}(h)$ is not shown here since it is simply linear) in order to illustrate the impact of $\lambda$ on the boundaries. As shown in Figure \ref{lambdatoboundary}, all four thresholds are higher when $\lambda$ increases. Among the four boundaries, $\lambda$ is most influential on the boundary $x=W_{\rm low}(h)$, which suggests that, with larger $\lambda$, more wealth is needed for getting rid of gloom. However, $\lambda$ has a negligible effect on $x=W_{\rm peak}(h)$ and $x=W_{\rm updt}(h)$, indicating that whether to consume at the historical level and whether to update the historical level have almost no reliance on the degree of drawdown constraint.

\begin{figure}[ht]
    \centering
    {%
    \includegraphics[width=.31\linewidth]{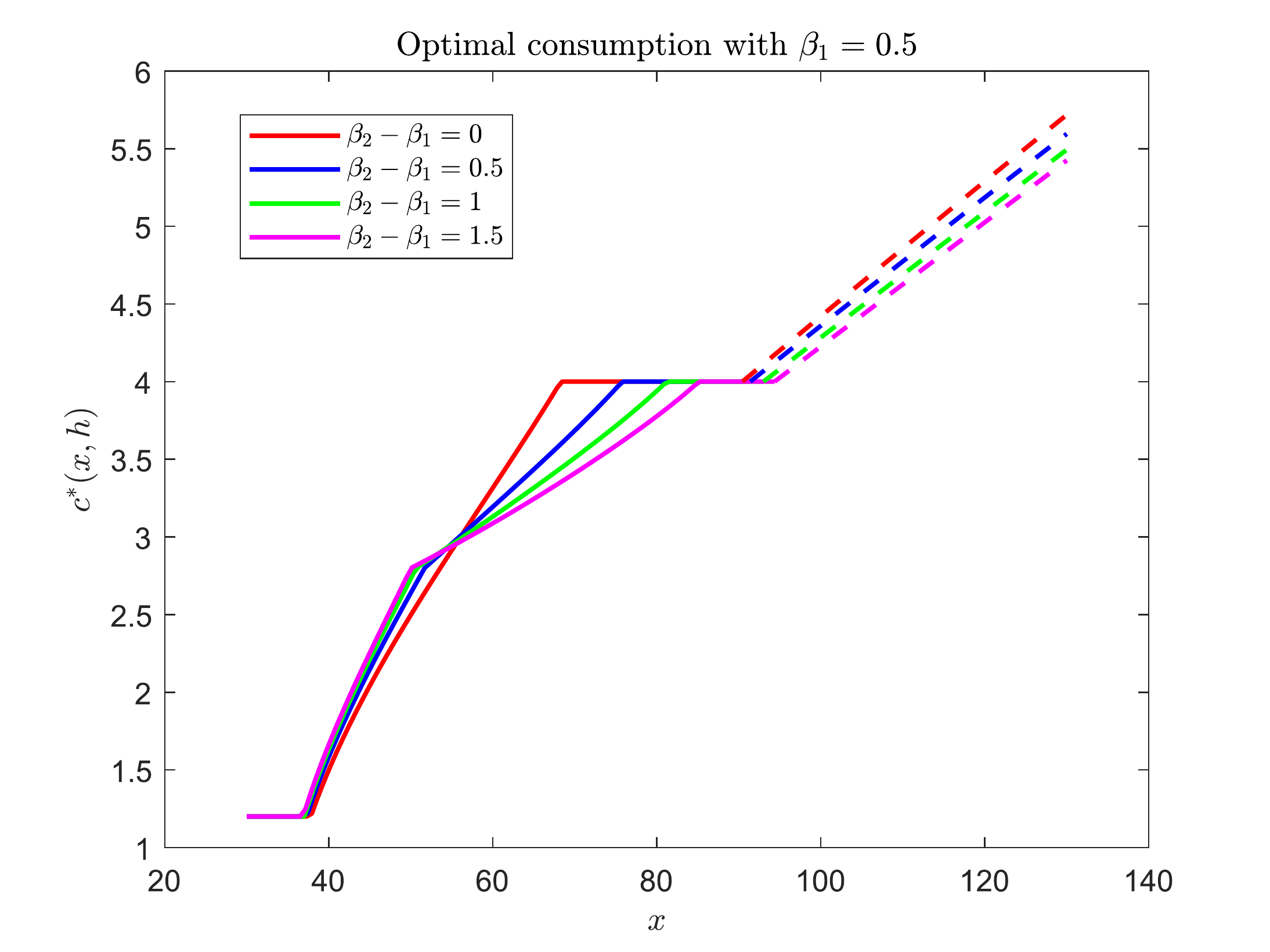}\quad
    \includegraphics[width=.31\linewidth]{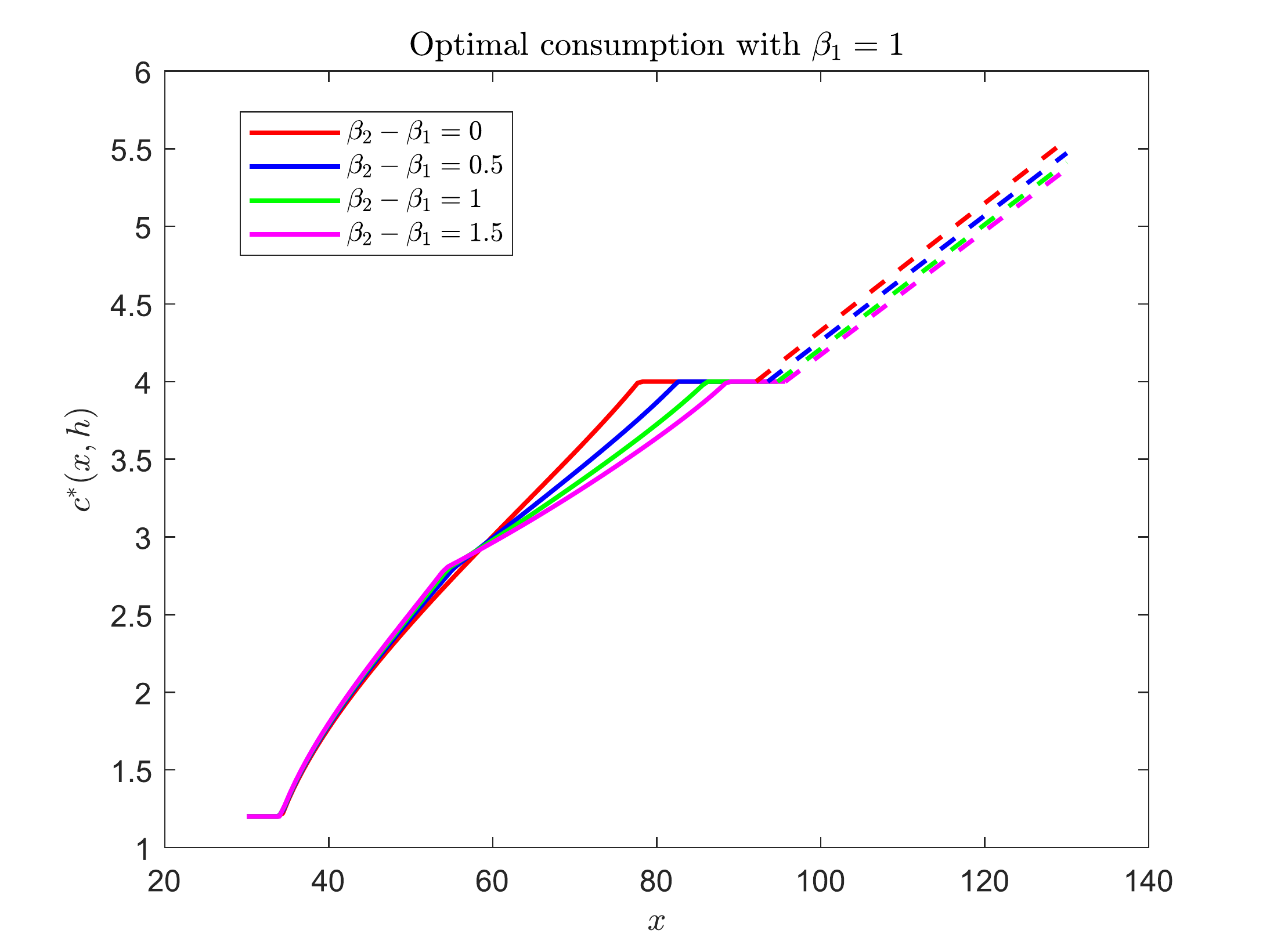}\quad
    \includegraphics[width=.31\linewidth]{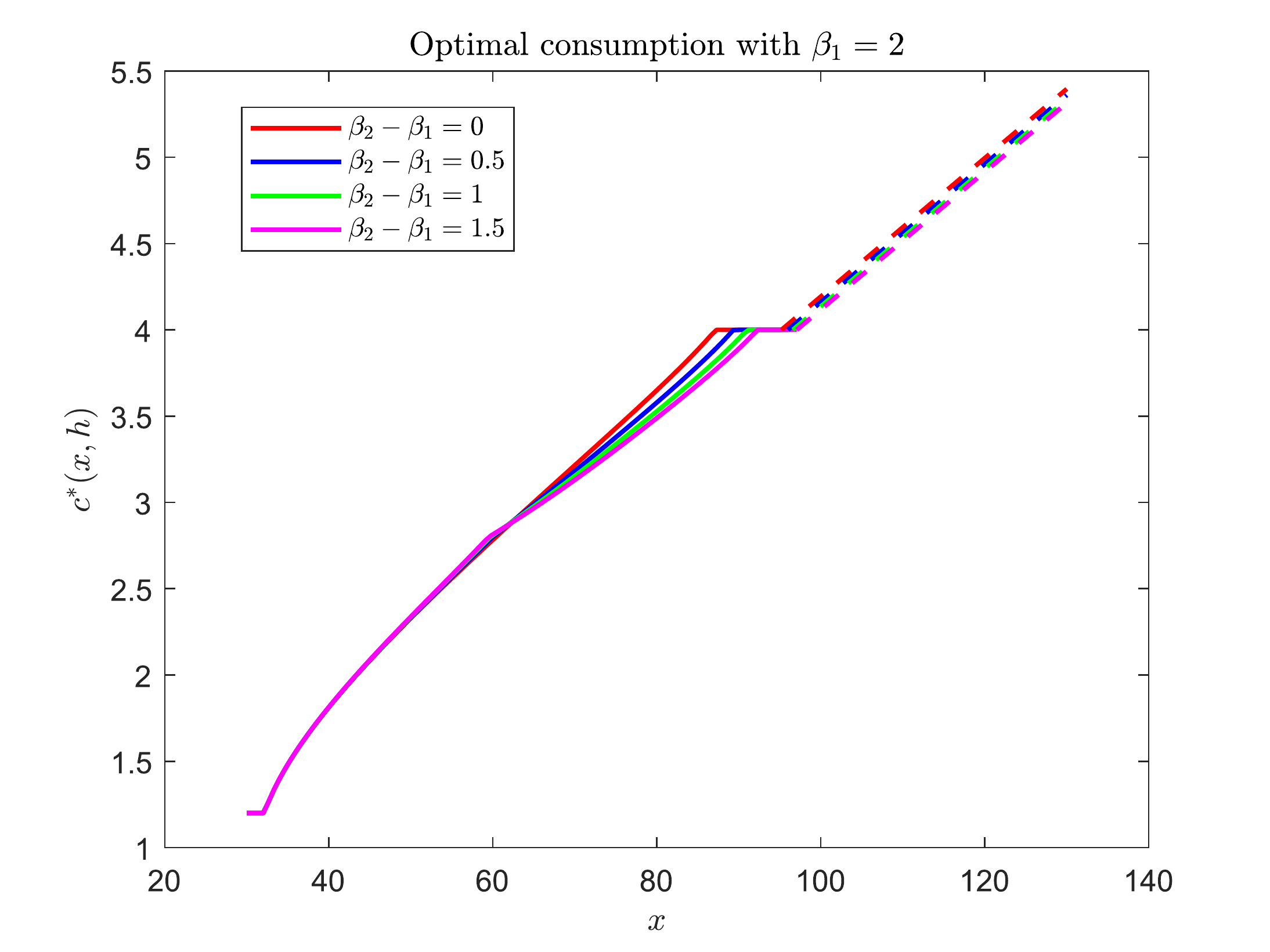}
    }
    \caption{With fixed parameters $\alpha=0.7,\ r=0.04,\ \mu=0.12,\ \sigma=0.3, \ h=4$, impact of $\beta_{2}-\beta_{1}$ on optimal consumption with fixed $\beta_{1}$.}
    \label{betatocon}
\end{figure}

\begin{figure}[ht]
    \centering
    {%
    \includegraphics[width=.31\linewidth]{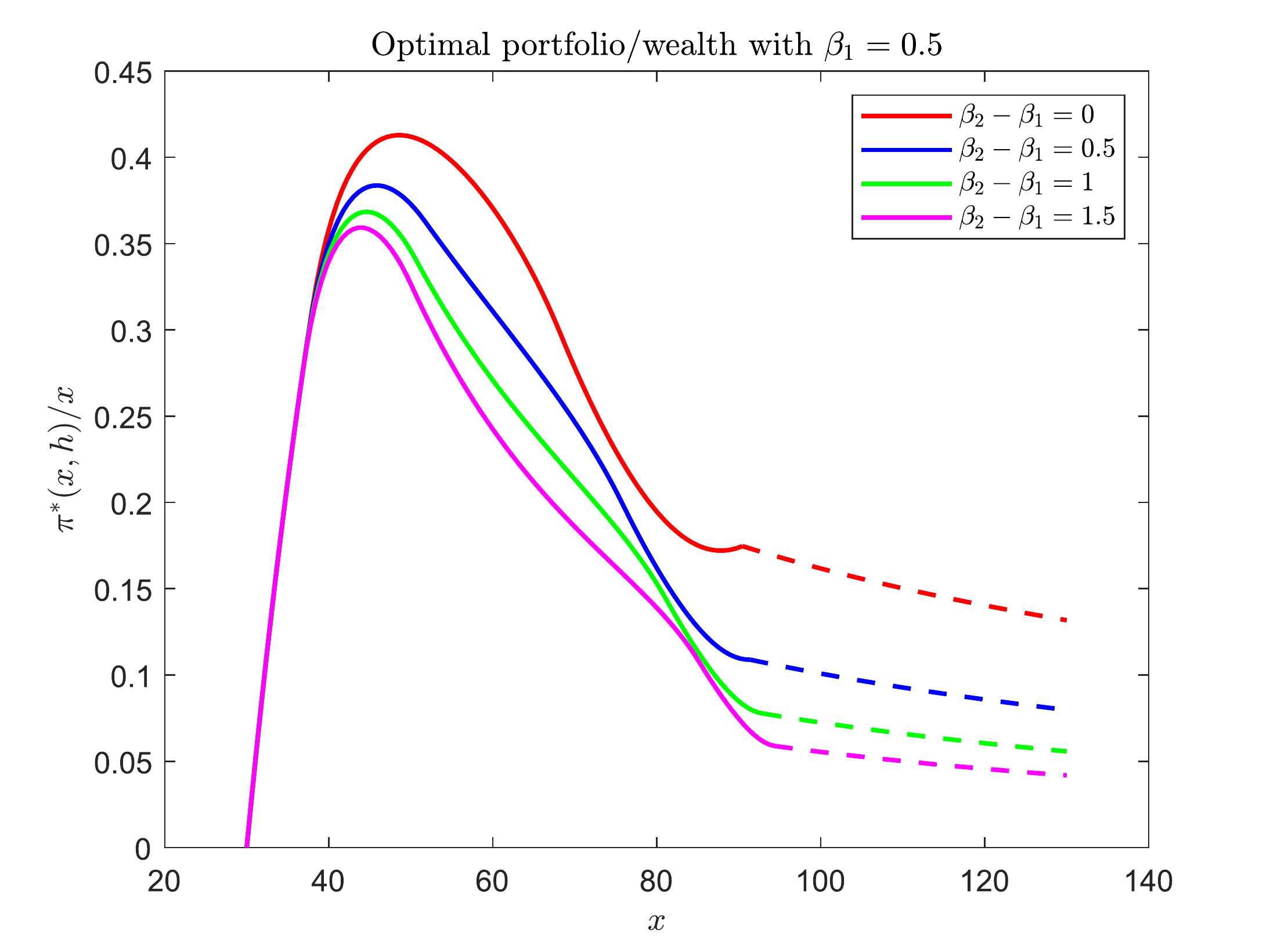}\quad
    \includegraphics[width=.31\linewidth]{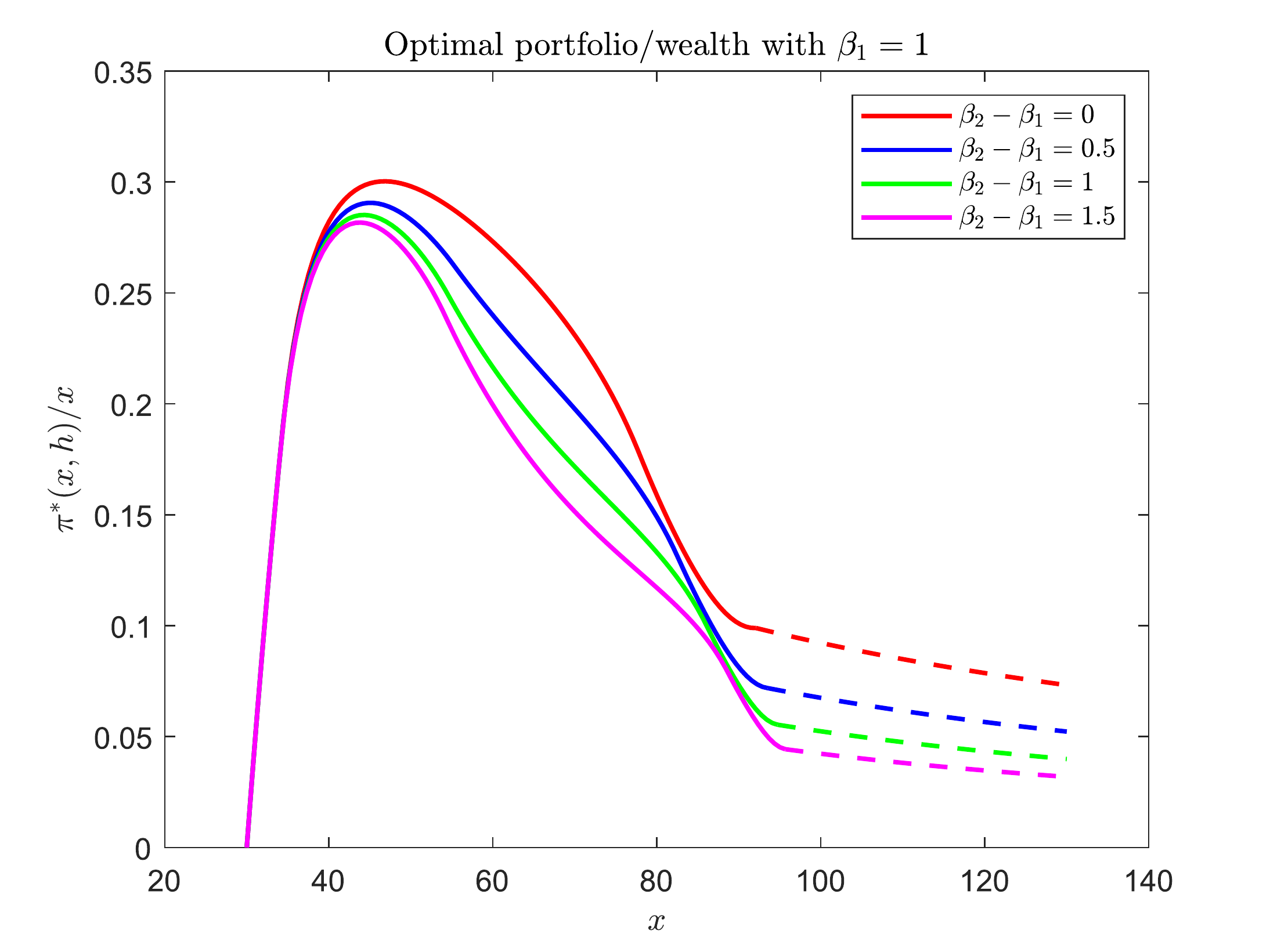}\quad
    \includegraphics[width=.31\linewidth]{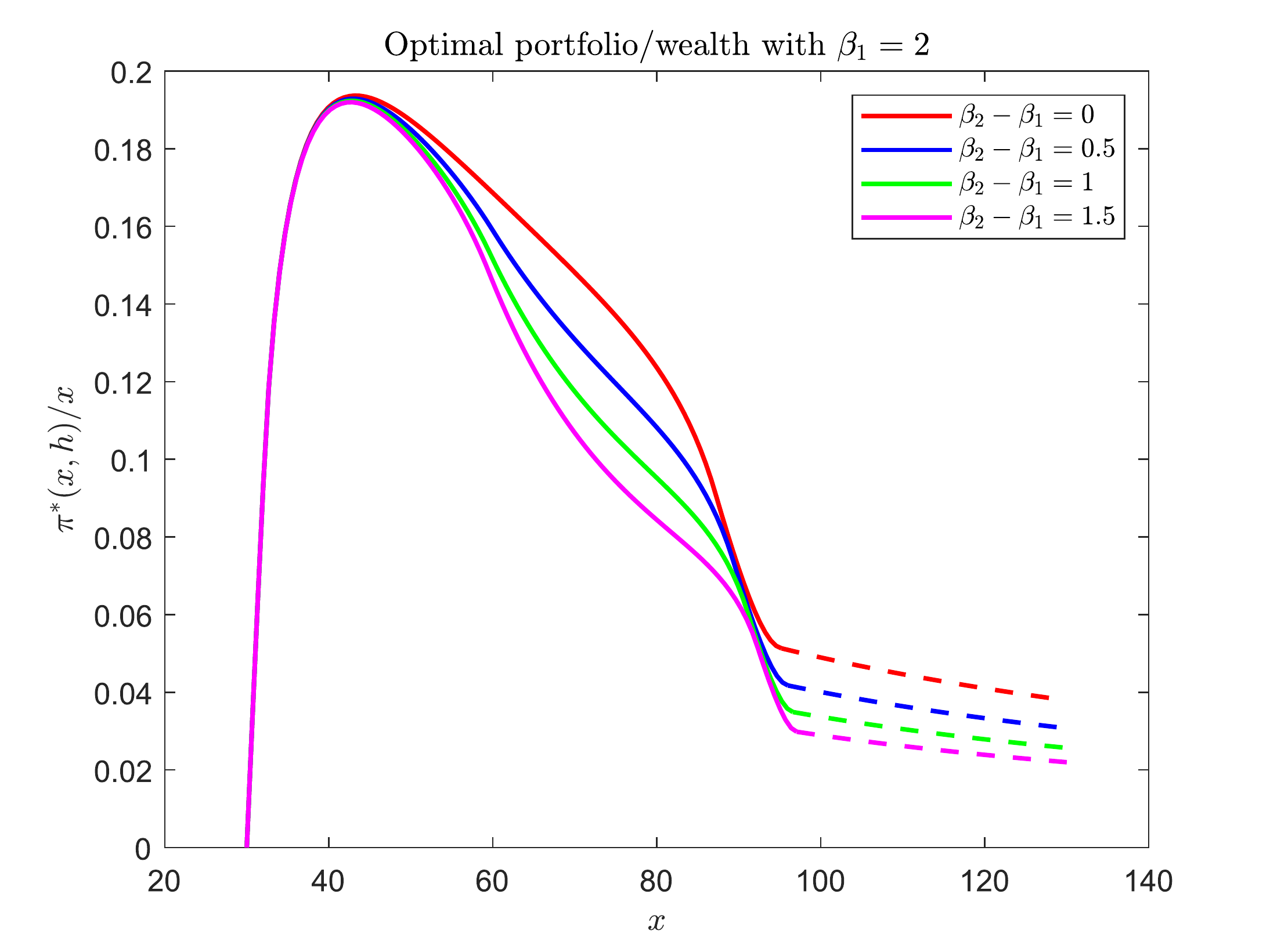}
    }
    \caption{With fixed parameters $\alpha=0.7,\ r=0.04,\ \mu=0.12,\ \sigma=0.3, \ h=4$, impact of $\beta_{2}-\beta_{1}$ on optimal risky investment proportion with fixed $\beta_{1}$.}
    \label{betatopor}
\end{figure}

\subsection{Discussion for $\beta_{1}$ and $\beta_2$}

In this subsection, we focus on the sensitivity analysis of the parameters $\beta_{1}$ and $\beta_{2}$, which are the risk aversion  coefficients below and above the reference $\alpha h$. The difference $\beta_{2}-\beta_{1}$ reflects the magnitude of the risk aversion change over the reference. We mainly illustrate the results with $\beta_{1}\le\beta_{2}$. The differences in results with $\beta_{1}\ge\beta_{2}$ are briefly discussed in \ref{debe}. Two limiting cases of interest are investigated in \ref{lc}.

\subsubsection{$\beta_{1}\ge\beta_{2}$}
\quad\\
In this part, we analyze the influence of $\beta_{2}-\beta_{1}$ when $\beta_{1}\ge\beta_{2}$.

To investigate the influence of $\beta_{2}-\beta_{1}$, we have to fix either $\beta_{1}$ or $\beta_{2}$. One approach is fixing $\beta_{1}$, then $\beta_{2}-\beta_{1}$ is the increase in risk aversion when the consumption exceeds the reference $\alpha h$.

We first investigate the influence on optimal strategies. When $\beta_2-\beta_1$ is enlarges, the main influence on consumption is the decelerating of consumption increase once across the boundary $x=W_{\rm ref}(h)$ (see Figure \ref{betatocon}), which delays the arrival of consumption peak $x=W_{\rm peak}(h)$. The decelerating effect in consumption increase can be explained by the conservative consumption behavior due to higher risk aversion above the reference. Meanwhile, change of $\beta_2$ even influence the consumption in the region $x<W_{\rm ref}(h)$, where only $\beta_1$ seems to be relative. This can be regarded as an important risk allocation behavior: people suppress the consumption when he can tolerate less risks to compensate the consumption elsewhere.

We now consider the risky assets allocation. As $\beta_2-\beta_1$ enlarges, there is a decrease in risky investment proportion once over certain thresholds around $W_{\rm low}(h)$. The amount of decrease varies with different wealth levels $x$ and initial risk aversion $\beta_{1}$ (see Figure \ref{betatopor} for more details). It is reasonably expected that risky investment proportion should decrease due to higher risk aversion above the reference. Again, the decrease is not constrained to those wealth levels in recovery region ($x>W_{\rm ref}(h)$). It instead occurs prior to the threshold $W_{\rm ref}(h)$.

We also investigate the influence of $\beta_2-\beta_1$ on the wealth thresholds. As is shown in Figure \ref{betapeak}, the increase of $\beta_2-\beta_1$ leads to higher threshold $W_{\rm peak}(h)$, especially for small $h$. The influence on other thresholds is negligible.

\begin{figure}[ht]
    \centering
    {%
    \includegraphics[width=.31\linewidth]{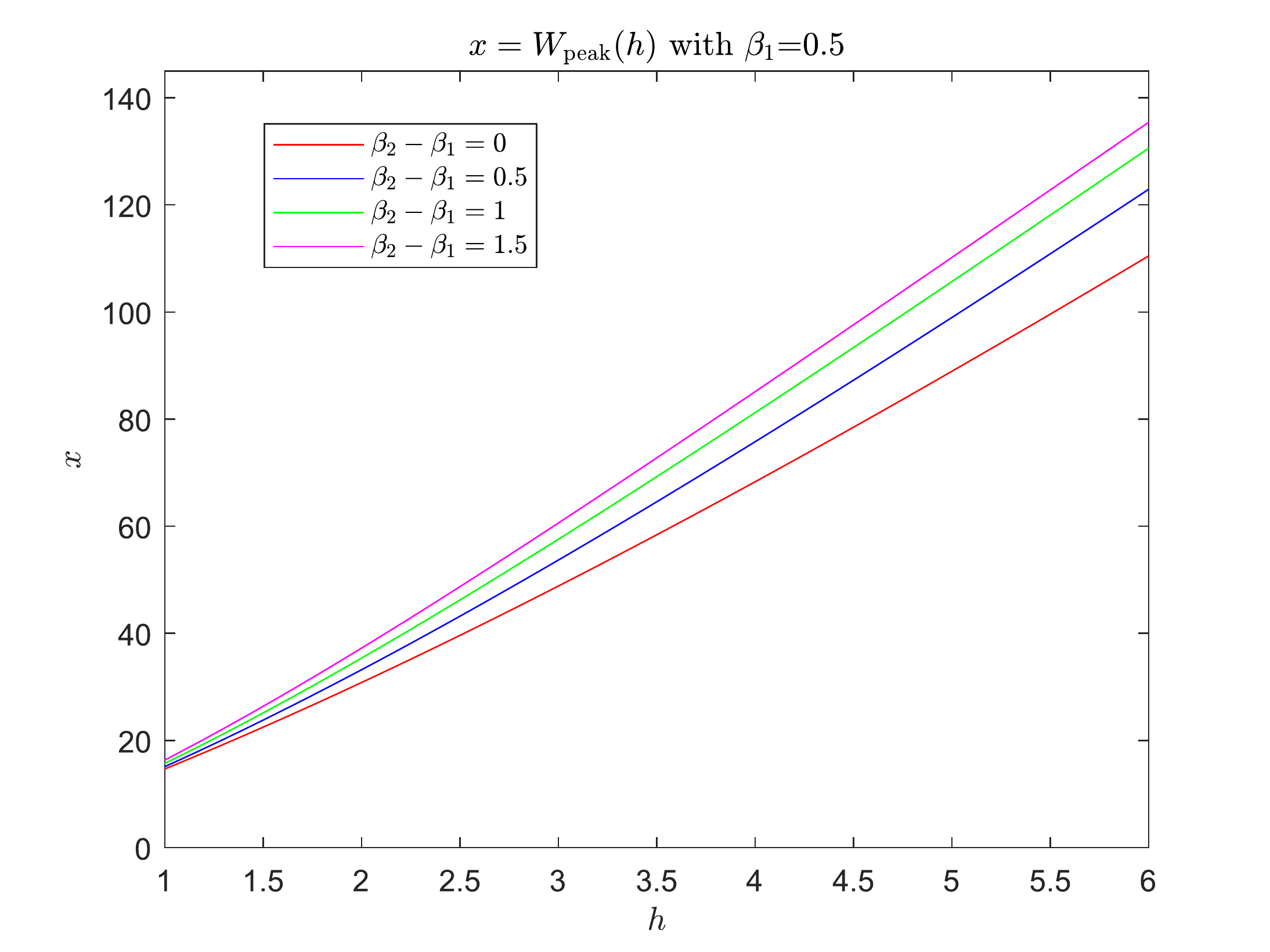}\quad
    \includegraphics[width=.31\linewidth]{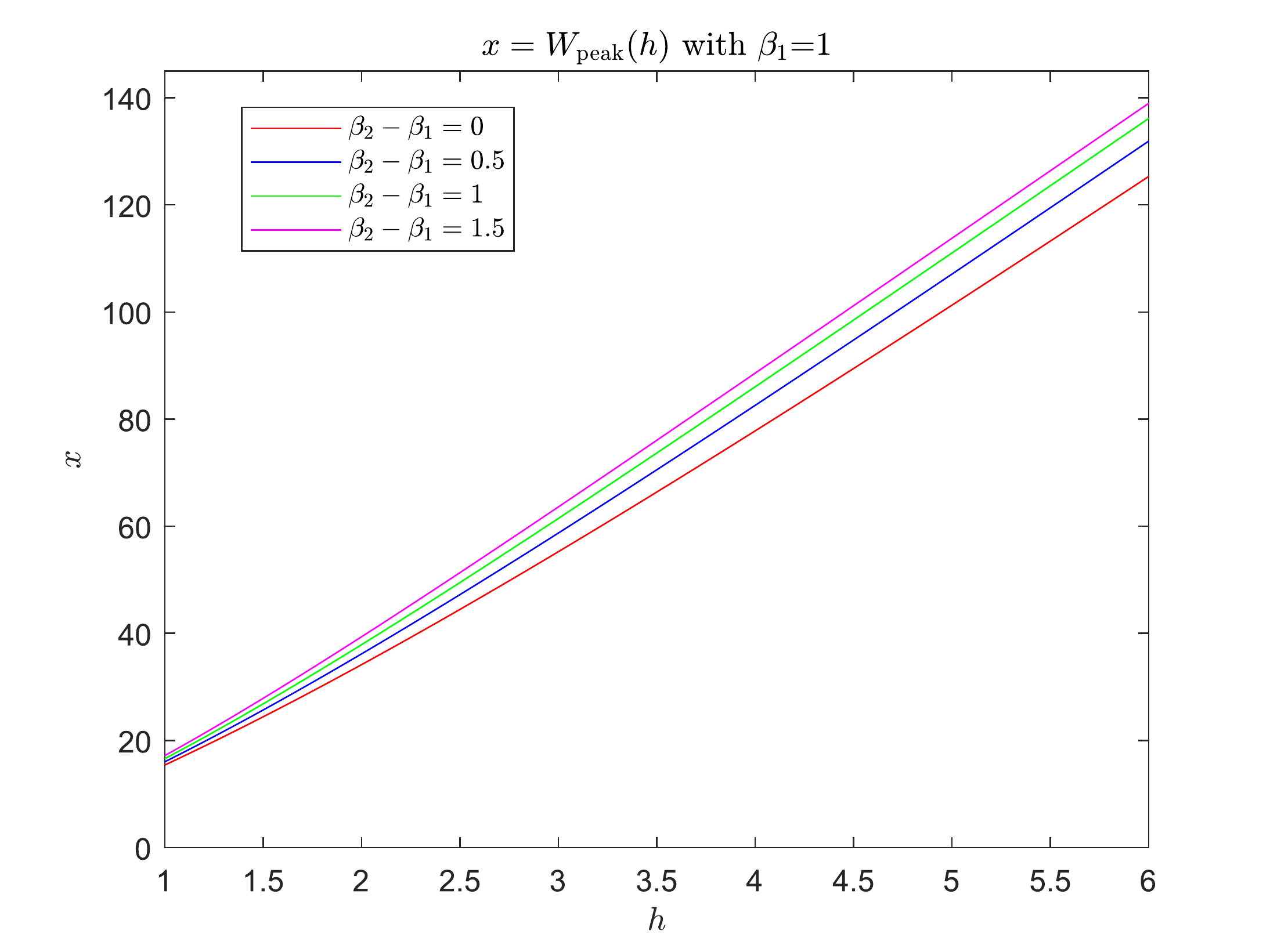}\quad
    \includegraphics[width=.31\linewidth]{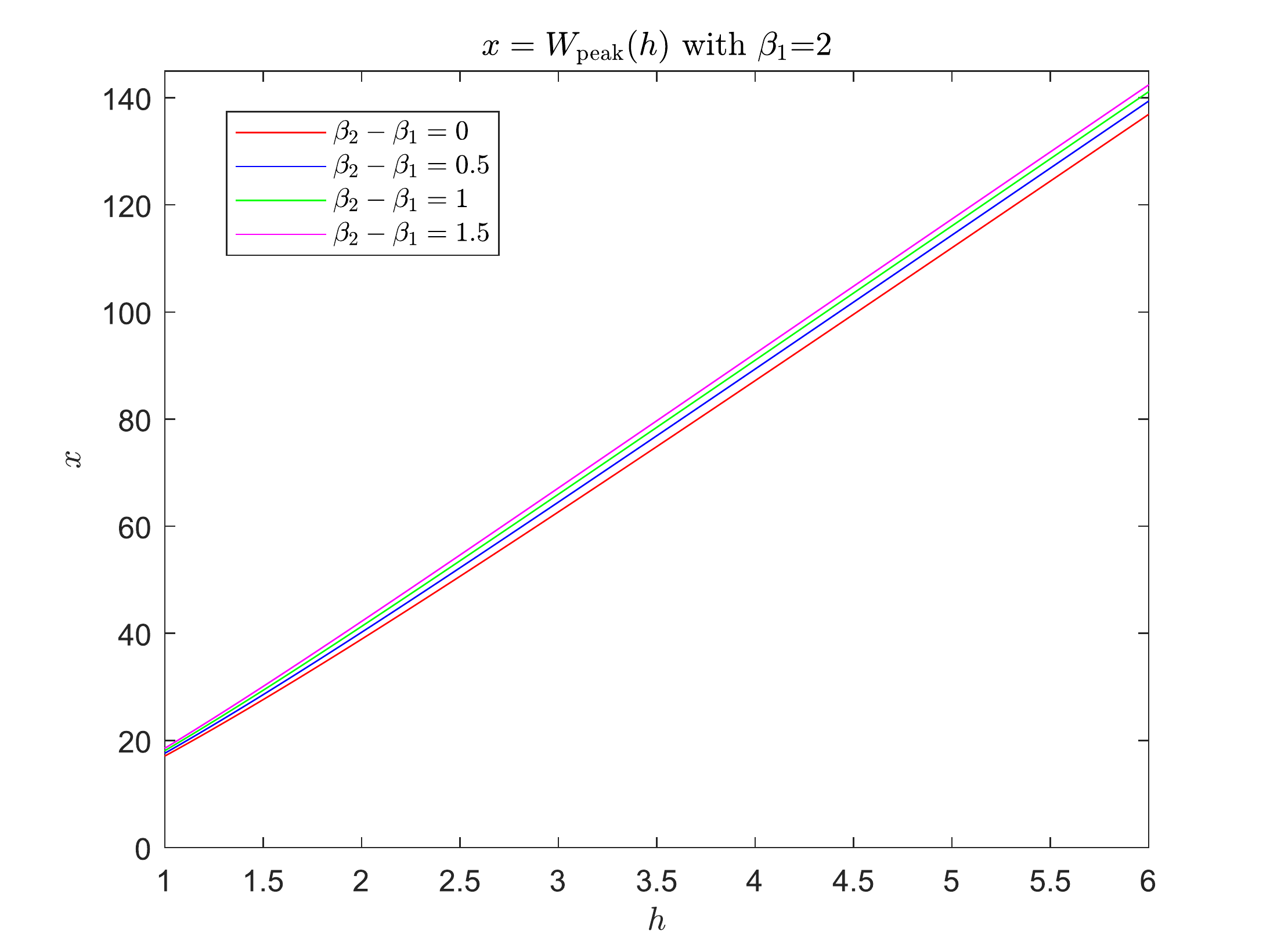}
    }
    \caption{With fixed parameters $\alpha=0.7,\ r=0.04,\ \mu=0.12,\ \sigma=0.3$, impact of $\beta_{2}-\beta_{1}$ on $x=W_{\rm peak}(h)$ with fixed $\beta_{1}$.}
    \label{betapeak}
\end{figure}

Another approach to examine the effect of risk aversion change $\beta_{2}-\beta_{1}$ is fixing $\beta_{2}$ and viewing $\beta_{2}-\beta_{1}$ as the decrease in risk aversion when consumption falls below the reference $\alpha h$.

As $\beta_{2}-\beta_{1}$ enlarges, the decrease of consumption decelerates above $W_{\rm ref}(h)$ but accelerates below $W_{\rm ref}(h)$ (see Figure \ref{beta2tocon}). The acceleration can be reasonably explained by aggressive consumption behavior due to lower risk aversion below the reference. Again, the effect of $\beta_1$ is not limited to $x<W_{\rm ref}(h)$.

The risky investment proportion, on the other hand, sees a significant increase when $\beta_2-\beta_1$ increases for wealth in the depression and recovery regions $[W_{\rm low}(h),W_{\rm peak}(h)]$, especially around $W_{\rm ref}(h)$ (see Figure \ref{beta2topor}). This is due to lower risk aversion below the reference. 

However, the change in $\beta_{2}-\beta_{1}$ for fixed $\beta_{2}$ does not have a significant impact on wealth thresholds. We merely observe a decrease in $W_{\rm ref}(h)$ for small values of $\beta_{2}$ as $\beta_{2}-\beta_{1}$ enlarges (see Figure \ref{betalore}).

\begin{figure}[ht]
    \centering
    {%
    \includegraphics[width=.31\linewidth]{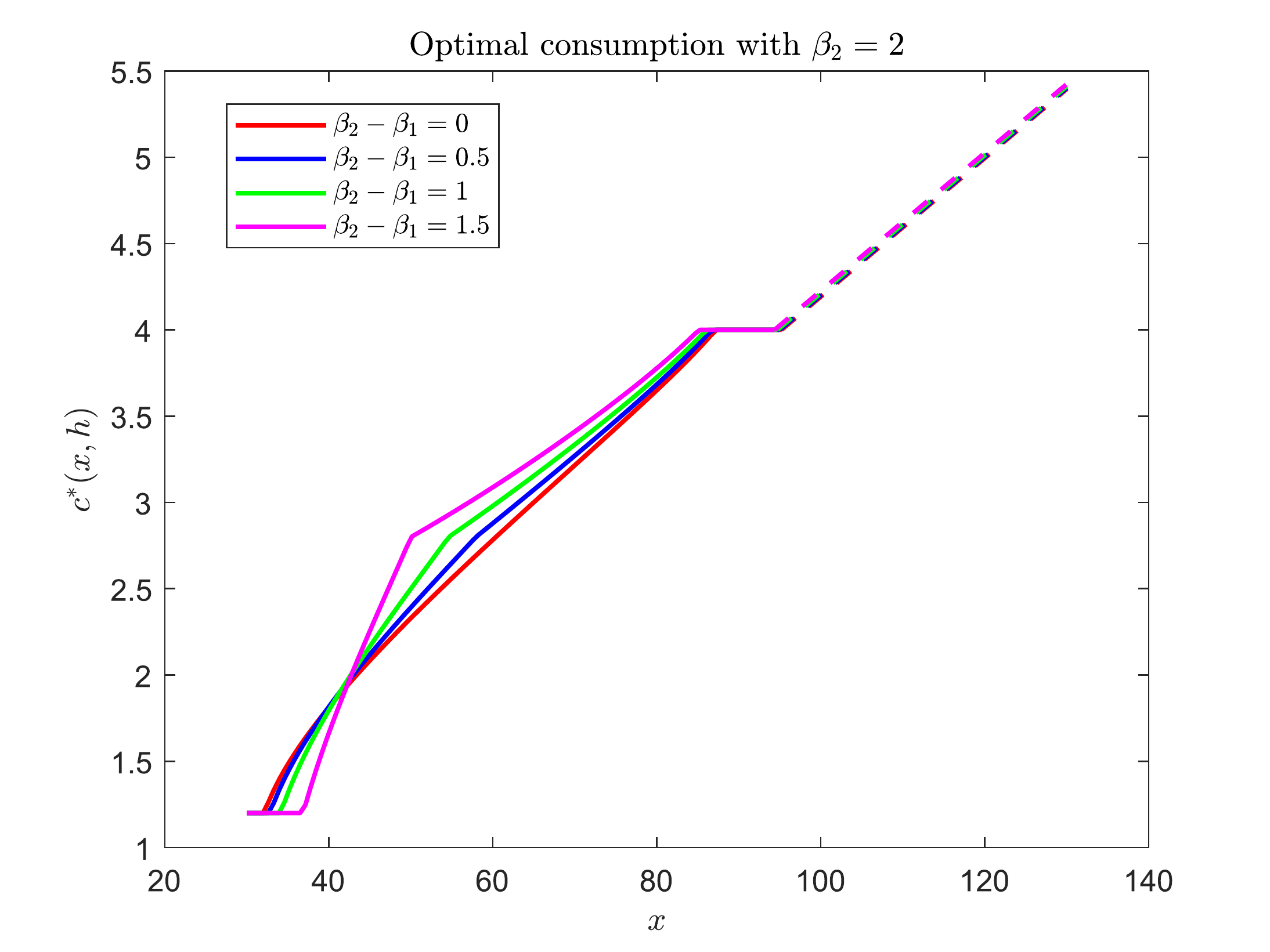}\quad
    \includegraphics[width=.31\linewidth]{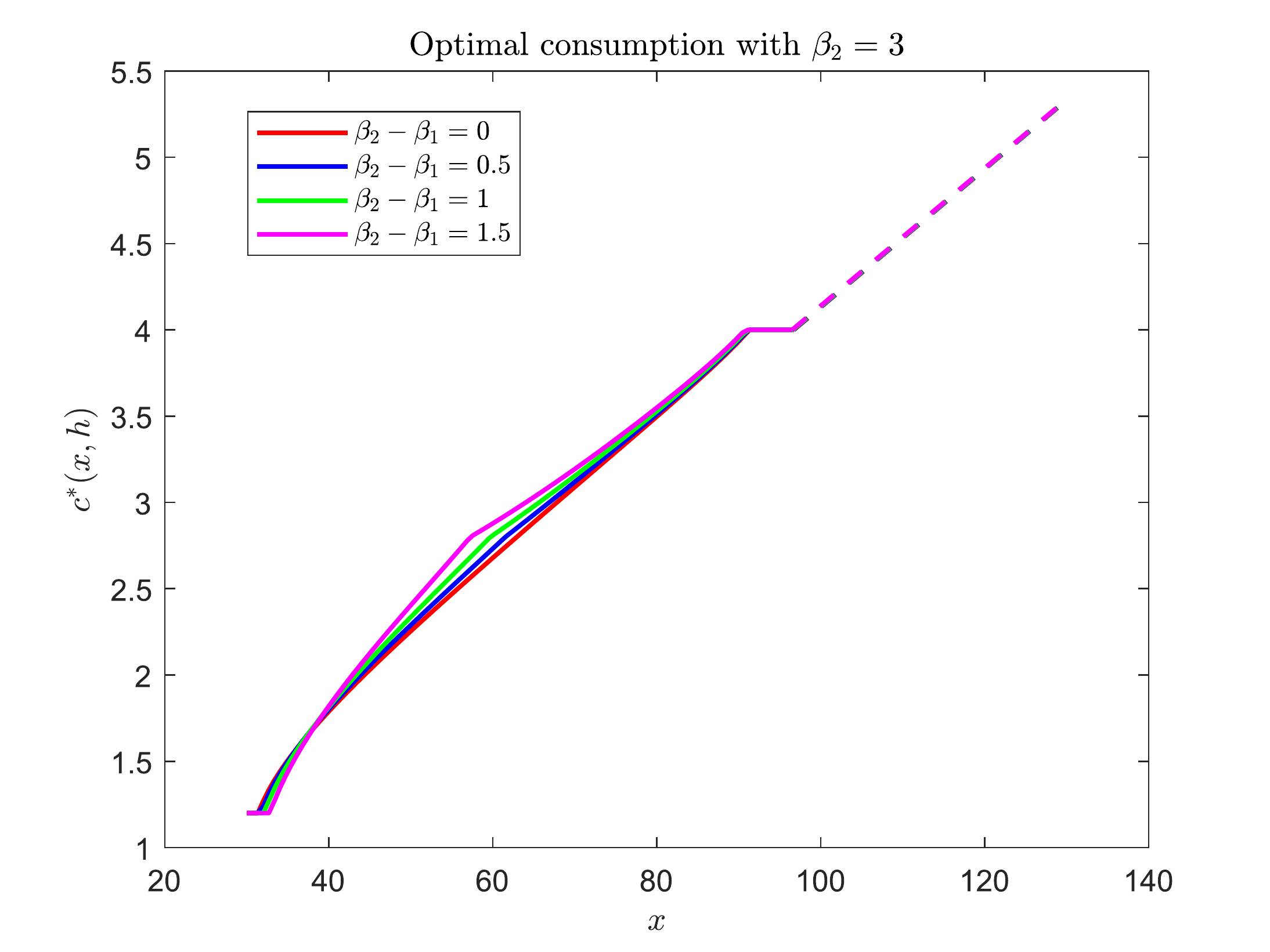}\quad
    \includegraphics[width=.31\linewidth]{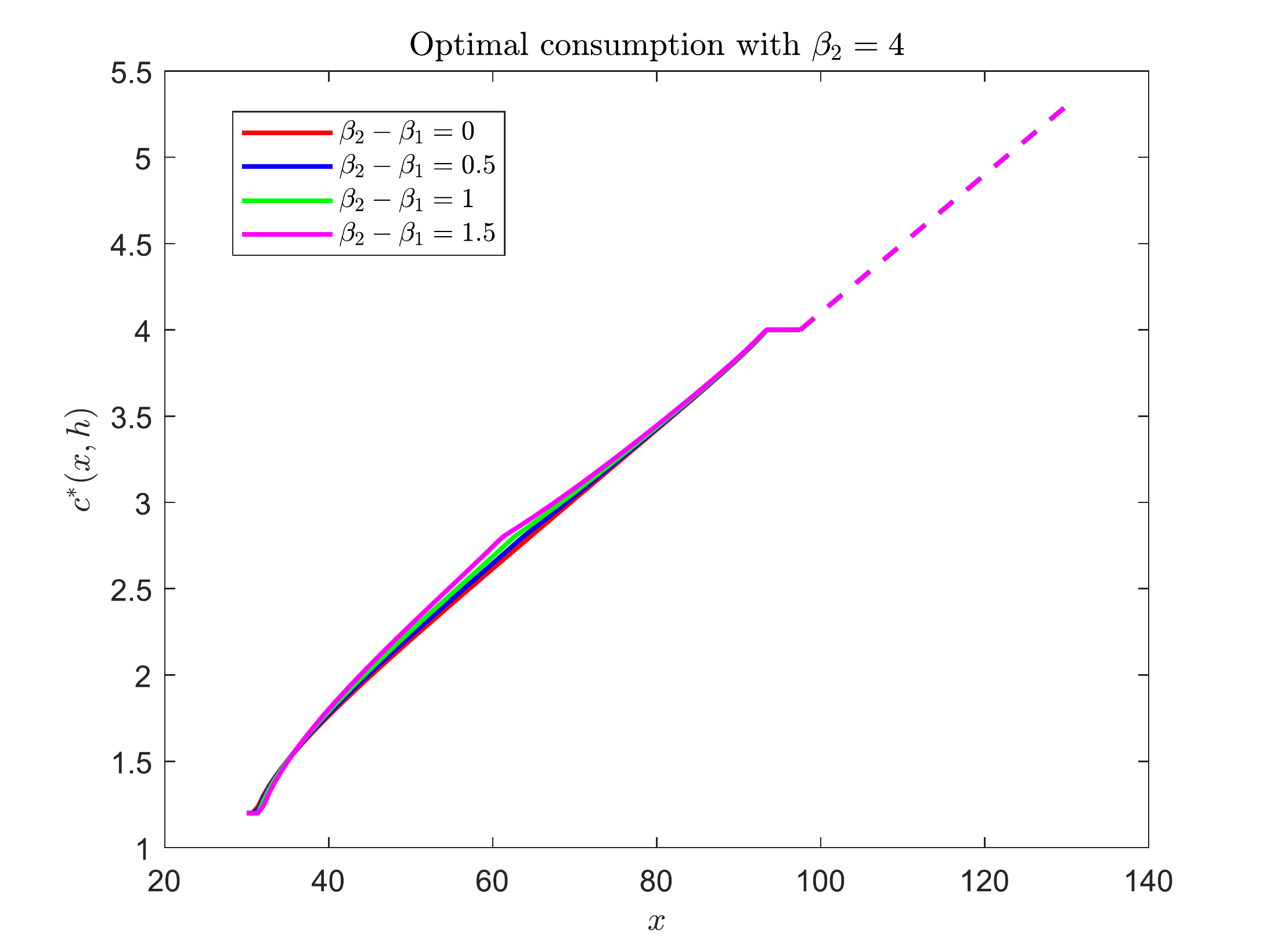}
    }
    \caption{With fixed parameters $\alpha=0.7,\ r=0.04,\ \mu=0.12,\ \sigma=0.3, \ h=4$, impact of $\beta_{2}-\beta_{1}$ on optimal consumption with fixed $\beta_{2}$.}
    \label{beta2tocon}
\end{figure}

\begin{figure}[ht]
    \centering
    {%
    \includegraphics[width=.31\linewidth]{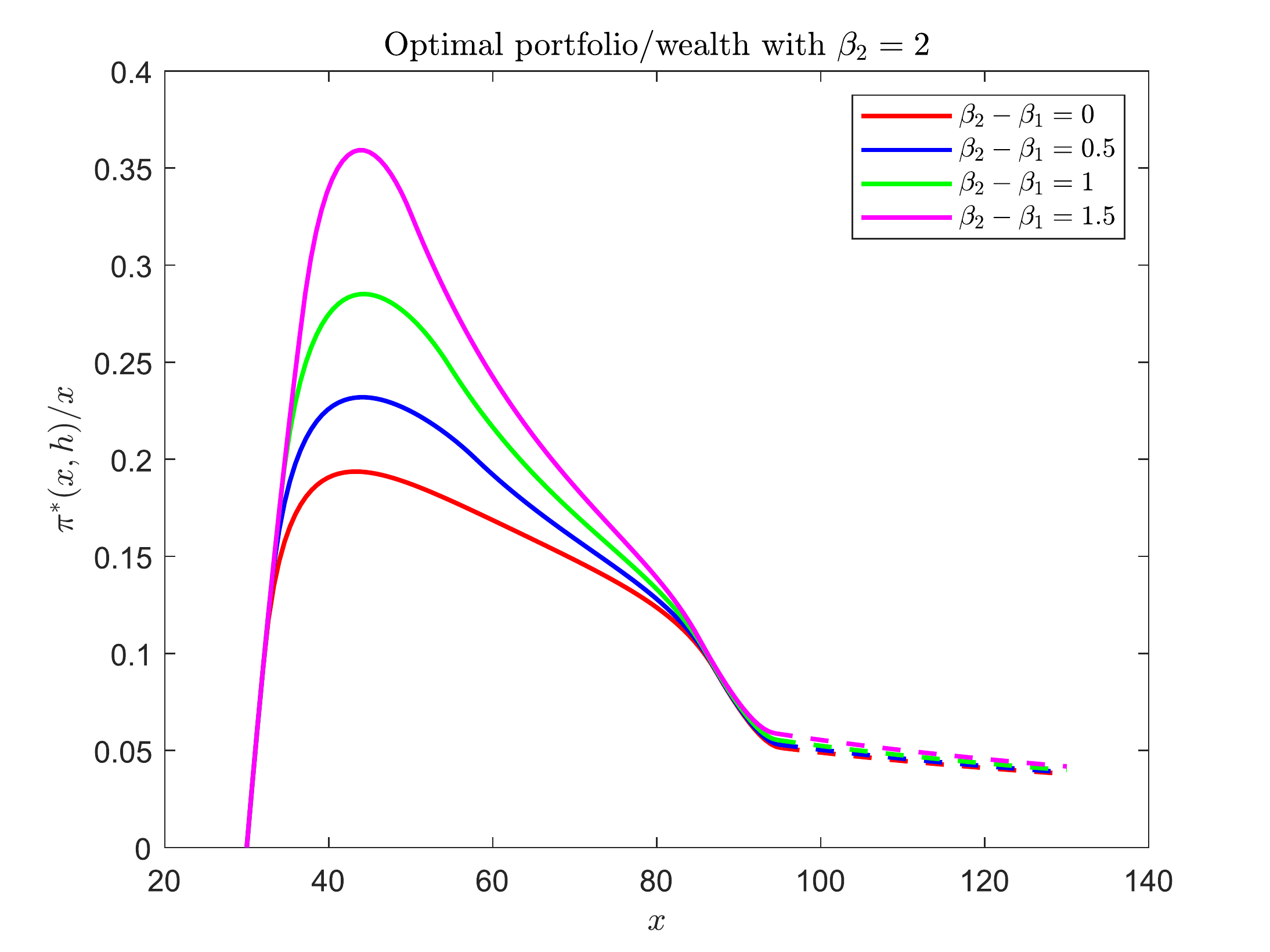}\quad
    \includegraphics[width=.31\linewidth]{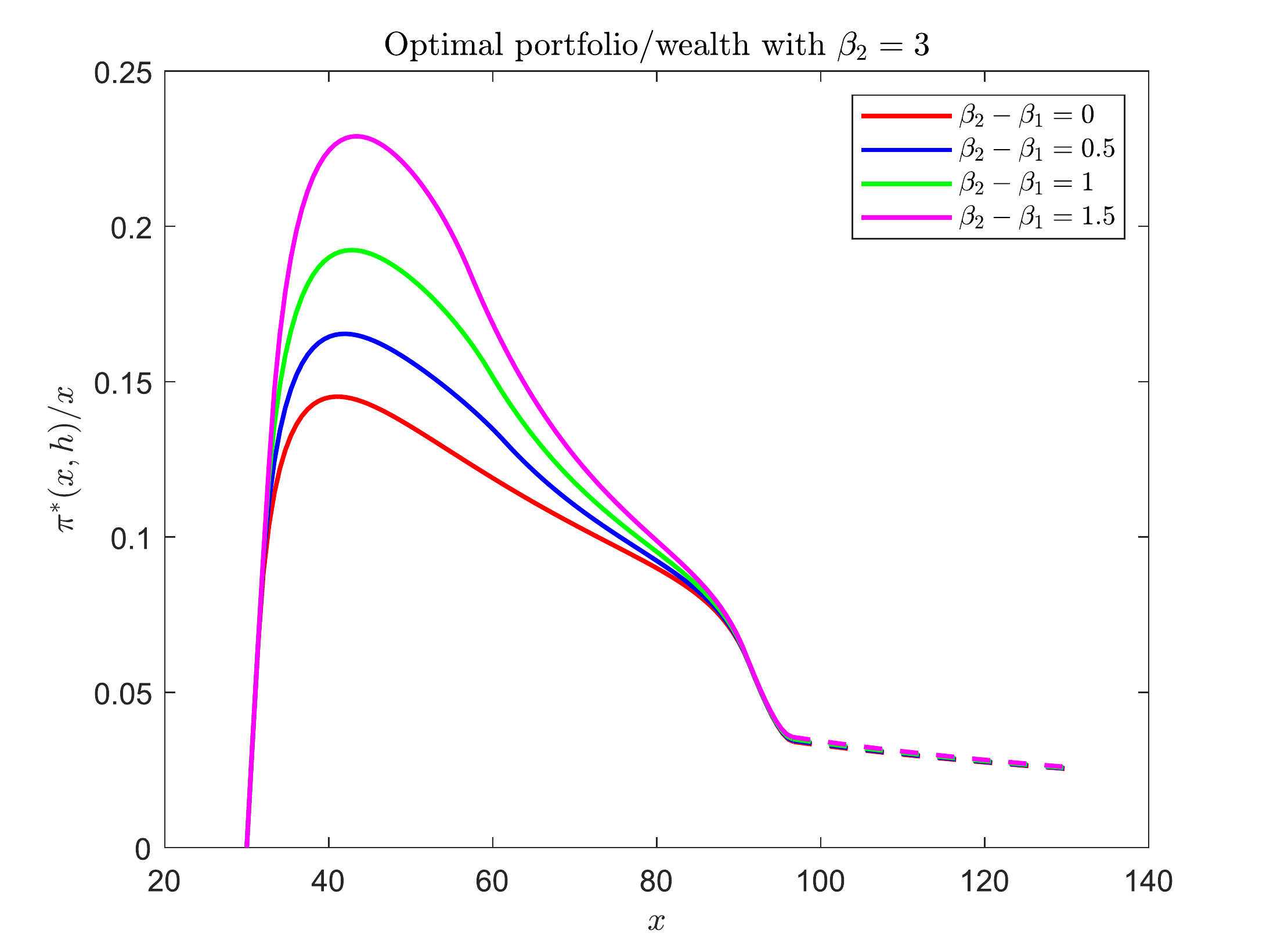}\quad
    \includegraphics[width=.31\linewidth]{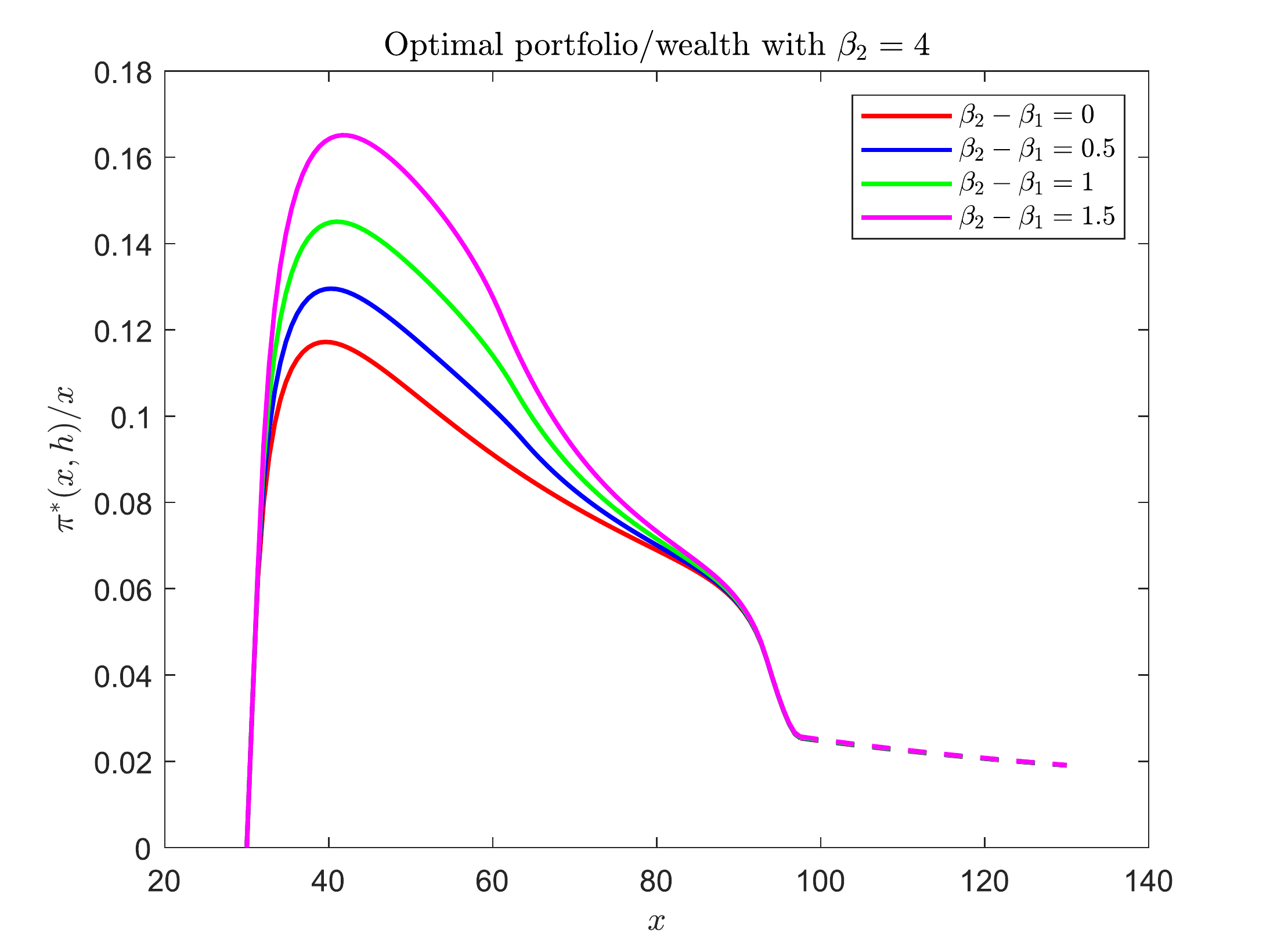}
    }
    \caption{With fixed parameters $\alpha=0.7,\ r=0.04,\ \mu=0.12,\ \sigma=0.3, \ h=4$, impact of $\beta_{2}-\beta_{1}$ on optimal risky investment proportion with fixed $\beta_{2}$.}
    \label{beta2topor}
\end{figure}

\begin{remark}
     The influence of $\beta_{2}-\beta_{1}$ can only be studied fixing either $\beta_{1}$ or $\beta_{2}$. The sensitivity analyses for the influence of $\beta_{2}-\beta_{1}$ on the optimal controls as well as the thresholds are established for relatively small $\beta_{1}$ or $\beta_{2}$ (the fixed one). As numerical results illustrate, When the fixed $\beta_{i}$ is sufficiently large, the optimal controls and the thresholds are nearly not affected by $\beta_{2}-\beta_{1}$. This phenomenon can already be well observed in the right panels in Figure \ref{betatocon}, Figure \ref{betapeak}, Figure \ref{beta2tocon} and Figure \ref{betalore}.
\end{remark}

\begin{figure}[ht]
    \centering
    {%
    \includegraphics[width=.31\linewidth]{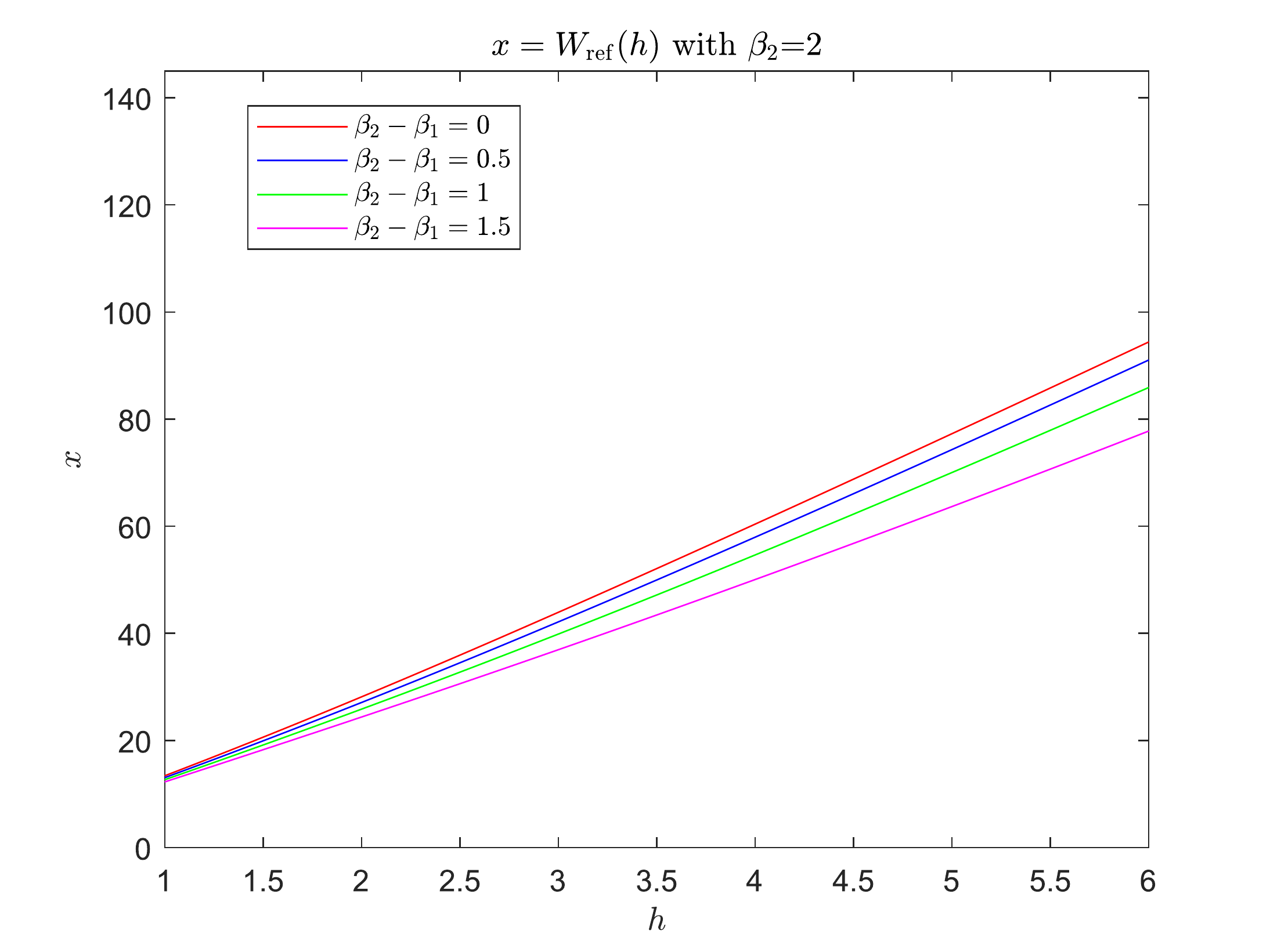}\quad
    \includegraphics[width=.31\linewidth]{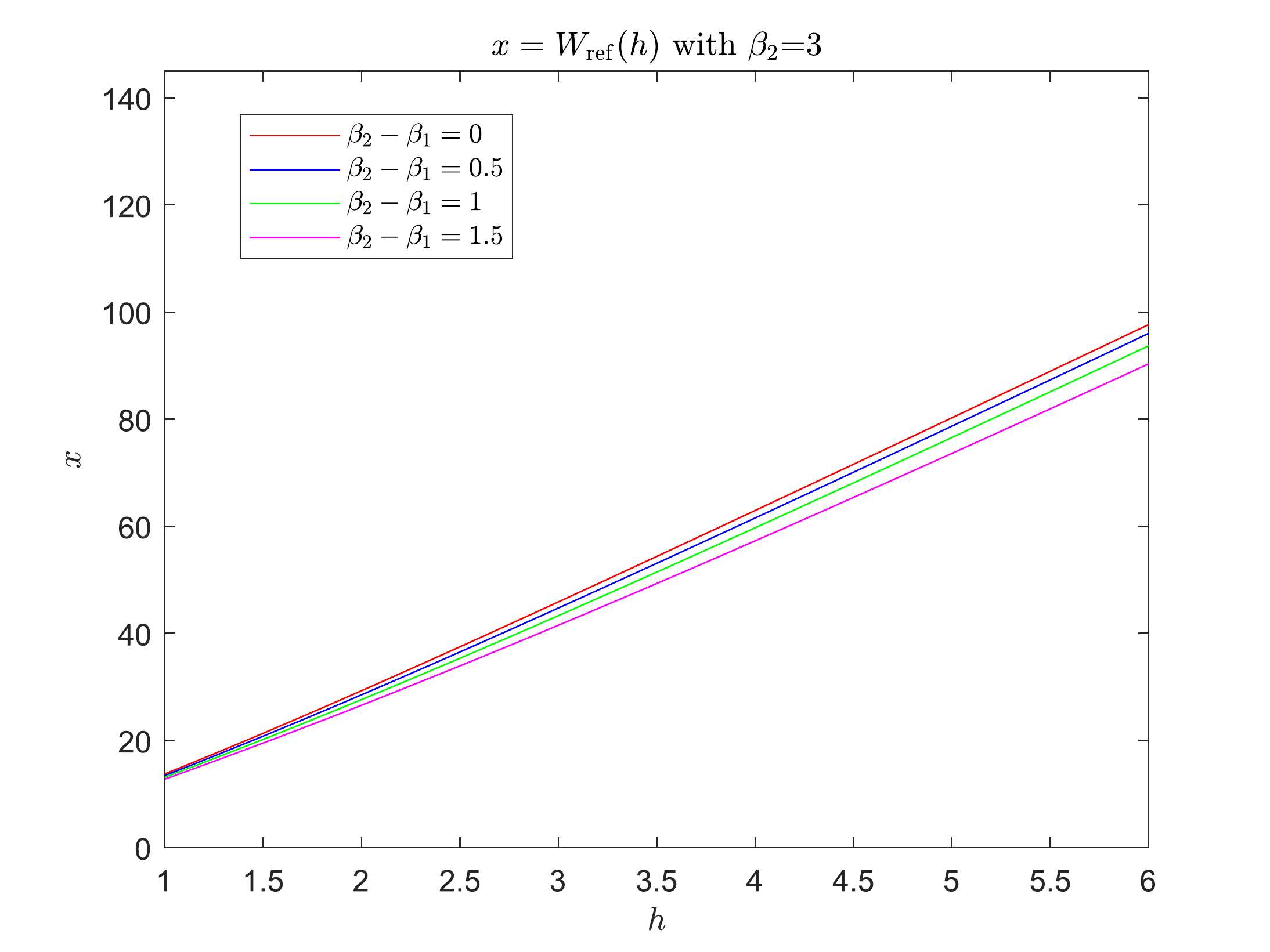}\quad
    \includegraphics[width=.31\linewidth]{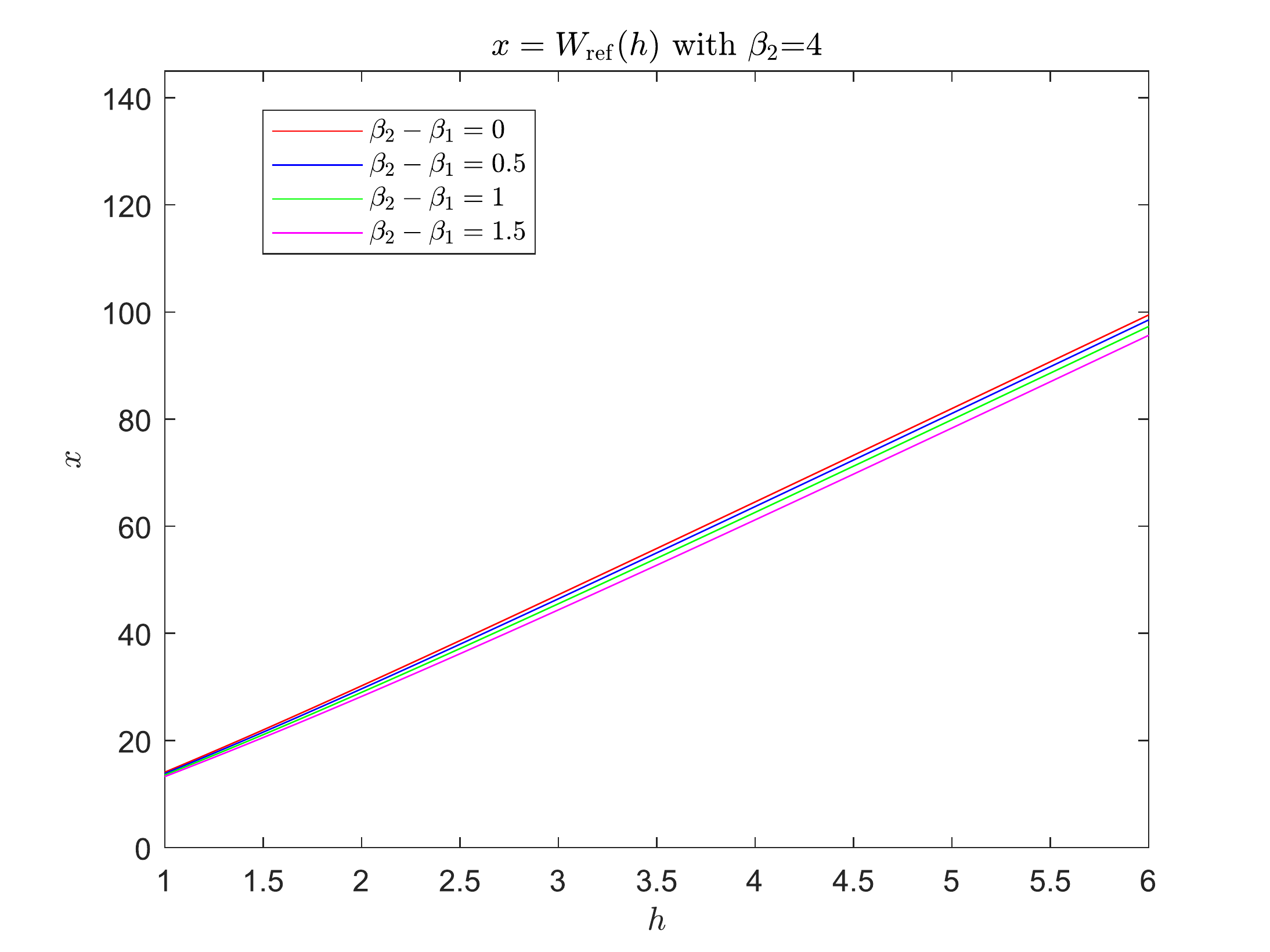}
    }
    \caption{With fixed parameters $\alpha=0.7,\ r=0.04,\ \mu=0.12,\ \sigma=0.3$, impact of $\beta_{2}-\beta_{1}$ on  $x=W_{\rm ref}(h)$ with fixed $\beta_{2}$.}
    \label{betalore}
\end{figure}
A special case in our model is $\beta_{2}-\beta_{1}=0$, which suggests that risk aversion does not change over the reference $\alpha h$. In this case, $W_{\rm ref}$ just exists symbolically but has no economical significance. The depression region ($x\in[W_{\rm low}(h),W_{\rm ref}(h)]$) and recovery region ($x\in[W_{\rm ref}(h),W_{\rm peak}(h)]$) share strategies and one region would merge into another.

\subsubsection{$\beta_{1}>\beta_{2}$}
\label{debe}
\quad\\
In this part, we consider the case $\beta_{1}>\beta_{2}$ and briefly discuss on the impact of $\beta_{1}-\beta_{2}$ on optimal strategies and boundaries. The impact for $\beta_{1}>\beta_{2}$ mainly corresponds to that for $\beta_{1}<\beta_{2}$ and can be well interpreted. However, it is beyond our expectation to observe a new peak around $W_{\rm peak}(h)$ in risky investment proportion that overtakes the peak around $W_{\rm ref}(h)$.

\begin{figure}[ht]
    \centering
    {%
    \includegraphics[width=.48\linewidth]{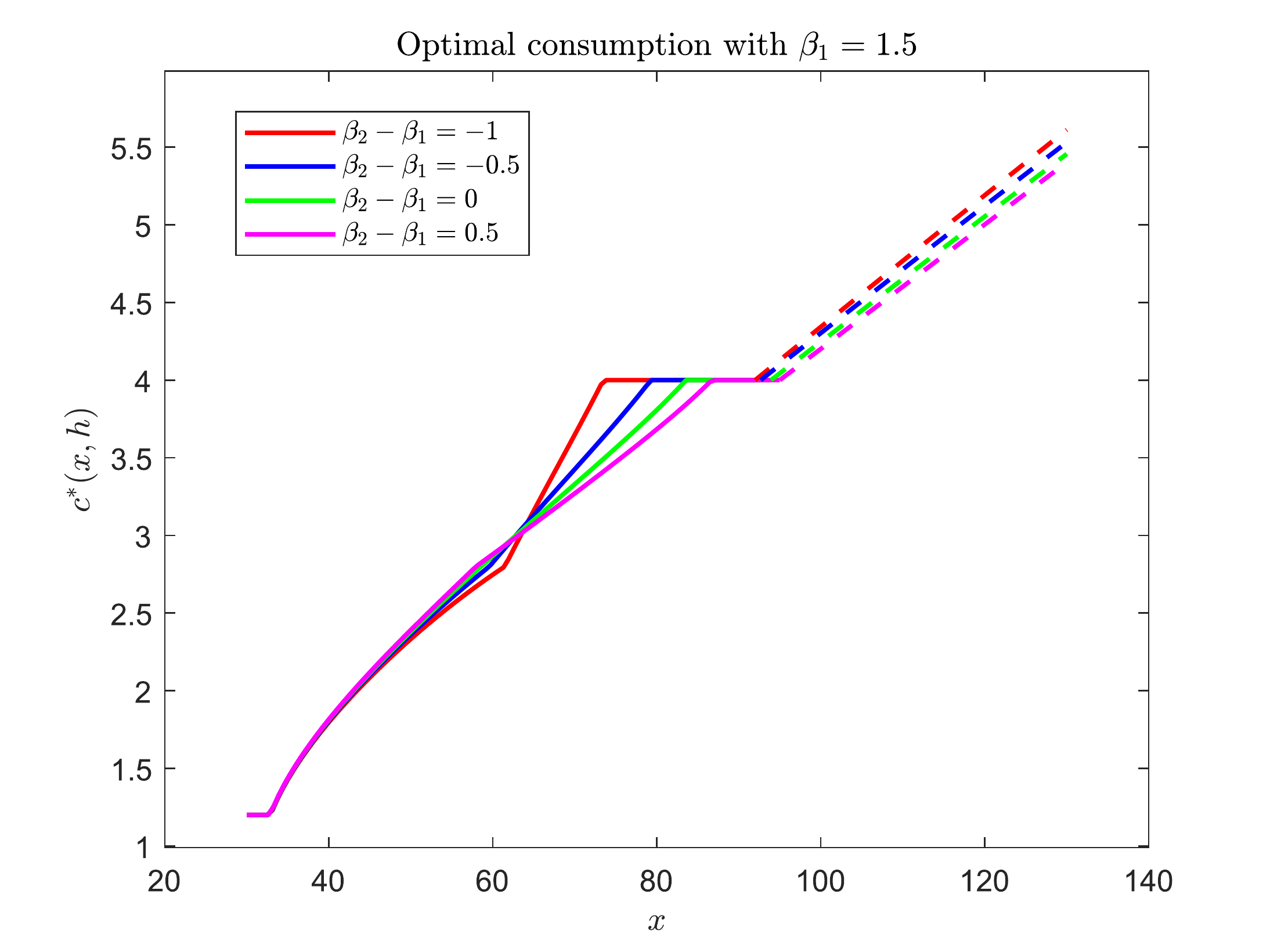}\quad
    \includegraphics[width=.48\linewidth]{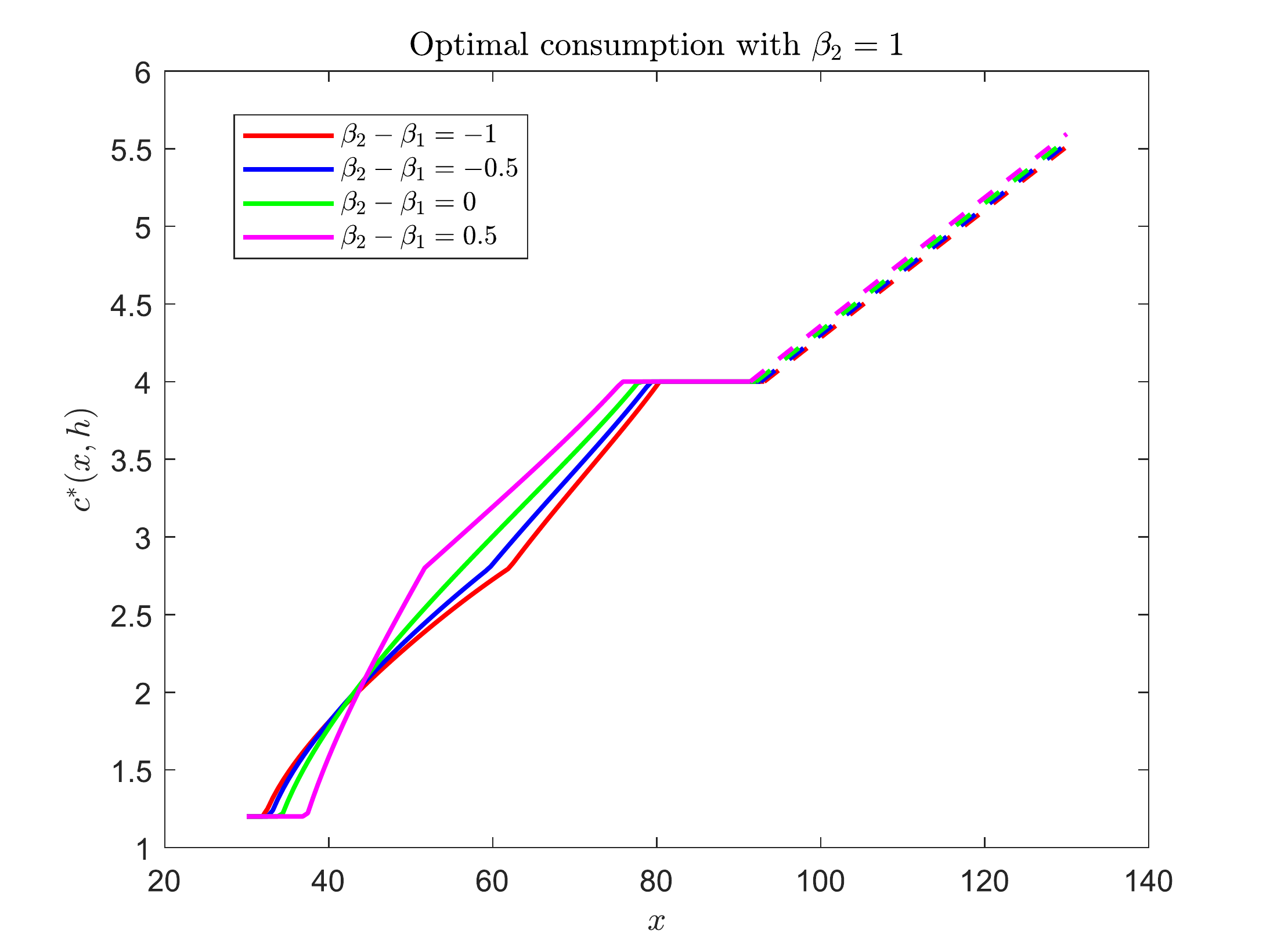}
    }
    \caption{With fixed parameters $\alpha=0.7,\ r=0.04,\ \mu=0.12,\ \sigma=0.3,\ h=4$, impact of $\beta_{2}-\beta_{1}$ on optimal consumption with fixed $\beta_{1}=1.5$ (left) or fixed $\beta_{2}=1$ (right).}
    \label{dibe1co}
\end{figure}

The optimal consumption for $\beta_{1}>\beta_{2}$ is still non-decreasing in $x$. The main difference occurs between $[W_{\rm low}(h),W_{\rm peak}(h)]$. As shown in Figure \ref{dibe1co}, when $\beta_{1}>\beta_{2}$, the risk aversion is lower in the recovery region and the MPC out of wealth is thus generally higher in the recovery region, which results in an upward turn at $W_{\rm ref}(h)$ (it is a downward turn when $\beta_{1}<\beta_{2}$). When fixing $\beta_{1}$, as $\beta_{1}-\beta_{2}$ enlarges, the increase of consumption accelerates over $W_{\rm ref}(h)$, bringing the arrival of consumption peak $W_{\rm peak}(h)$ forward. When fixing $\beta_{2}$, as $\beta_{1}-\beta_{2}$ enlarges, the decrease of consumption accelerates above $W_{\rm ref}(h)$ and decelerates below $W_{\rm ref}(h)$.

\begin{figure}[ht]
    \centering
    {%
    \includegraphics[width=.48\linewidth]{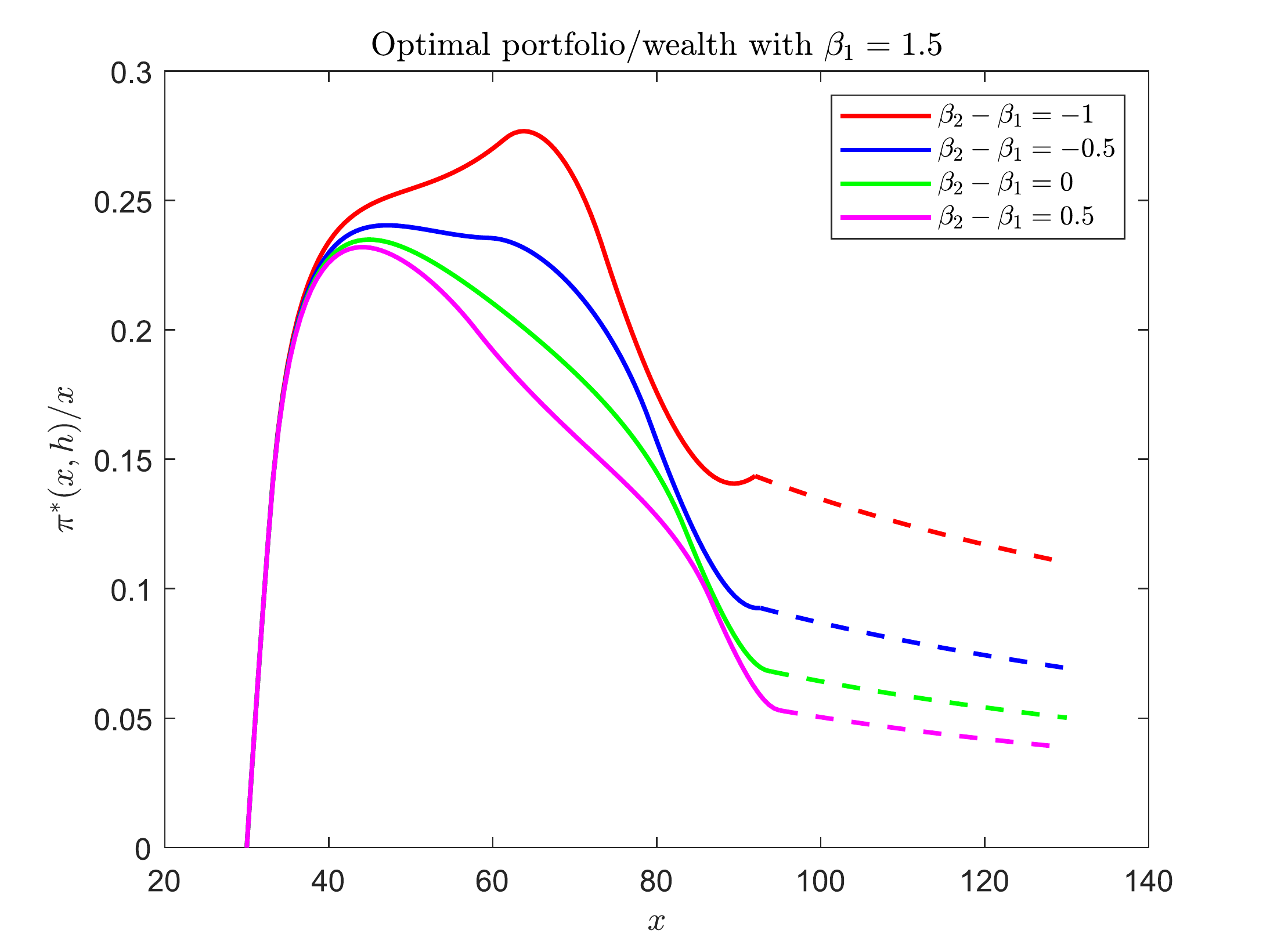}\quad
    \includegraphics[width=.48\linewidth]{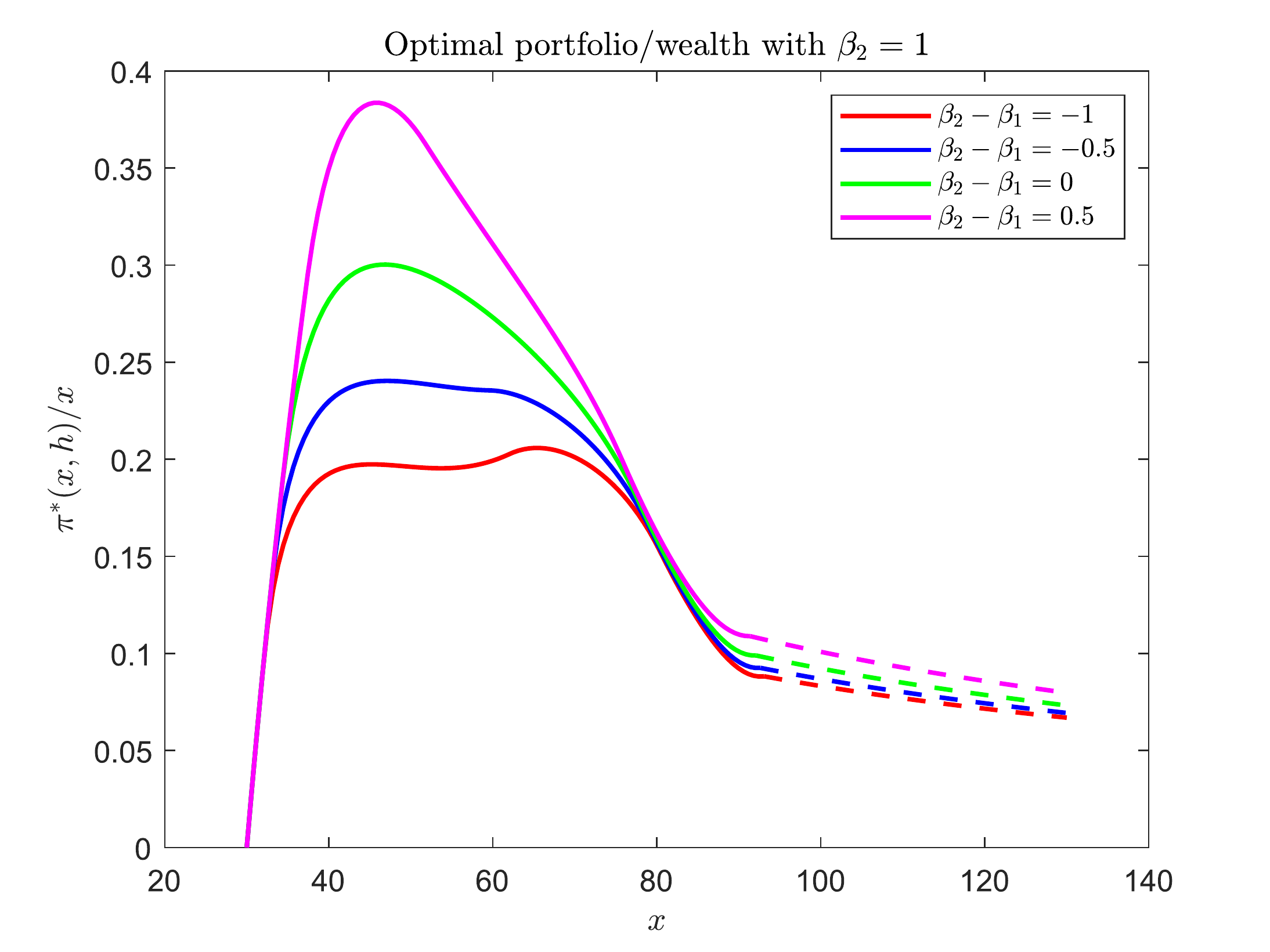}
    }
    \caption{With fixed parameters $\alpha=0.7,\ r=0.04,\ \mu=0.12,\ \sigma=0.3,\ h=4$, impact of $\beta_{2}-\beta_{1}$ on optimal risky investment proportion with fixed $\beta_{1}=1.5$ (left) or fixed $\beta_{2}=1$ (right).}
    \label{dibe1po}
\end{figure}

As shown in Figure \ref{dibe1po}, when $\beta_{1}-\beta_{2}$ enlarges, the optimal investment proportion sees an apparent increase over certain threshold around $W_{\rm low}(h)$ for fixed $\beta_{1}$ while there is an apparent decrease for wealth in depression and recovery regions for fixed $\beta_{2}$. Besides, we observe a new peak around $W_{\rm peak}(h)$ that overtakes the peak around $W_{\rm ref}(h)$ (see the red line in Figure \ref{dibe1po}). It suggests that high proportion of risky investment is recommended for wealth in the recovery region, especially around $W_{\rm ref}(h)$ and $W_{\rm peak}(h)$.

The impact of increasing $\beta_{1}-\beta_{2}$ on the boundaries coincides with the case $\beta_{1}<\beta_{2}$: lower level of $W_{\rm peak}(h)$ with fixed $\beta_{1}$, higher level of $W_{\rm ref}(h)$ with fixed $\beta_{2}$ and negligible effect on others.

\subsubsection{Limiting cases}
\label{lc}
In this subsection, we briefly discuss two limiting cases when one of the $\beta_i$ takes extreme value: fix $\beta_{1}$ and let $\beta_{2}\rightarrow 0$; fix $\beta_{2}$ and let $\beta_{1}\rightarrow 0$.

Fixing $\beta_{1}$ and letting $\beta_{2}\rightarrow 0$ indicates that the agent becomes risk neutral when $c>\alpha h$. Since $C_{1}(h)\rightarrow\infty$ as $\beta_{2}\rightarrow 0$, it is surprising to see that even the lowest constrained threshold $W_{\rm low}(h)$ tends to infinity. The limiting consumption curve is a horizontal line $c\equiv\lambda h$. It suggests that the agent always consumes at the lowest constrained level $\frac{\lambda}{r}h$. He might be saving money from consuming as less as possible in order to aggressively invest in risk asset once he reaches $W_{\rm ref}(h)$. However, the day never comes because $W_{\rm ref}(h)$ tends to infinity as $\beta_{2}\rightarrow 0$. 

Fixing $\beta_{2}$ and letting $\beta_{1}\rightarrow 0$, however, has completely different consequences. This limit corresponds to the case that the agent becomes risk neutral when $c<\alpha h$. A dedicated analysis (see Remark \ref{lims}) shows that all the thresholds have finite limits as $\beta_{1}\rightarrow 0$. In particular, $W_{\rm low}(h)$ and $W_{\rm ref}(h)$ share the same limit. It indicates that the agent never consumes between the lowest constrained level $\frac{\lambda h}{r}$ and the reference level $\alpha h$ in this limiting case. The result is similar to that of an S-shaped utility in \cite{Li2021} where the agent is risk seeking below the reference. However, in S-shaped utility, the optimal consumption jumps from 0 to a level that is strictly higher than reference point $\alpha h$ (see (3.25) in \cite{Li2021}), while our limiting optimal consumption jumps from the lowest level to exactly the reference point. 

\begin{remark}\label{lims}
From (\ref{C1})$\sim$(\ref{C7}), we have, as $\beta_{1}\rightarrow 0$,

\begin{align*}
C_{4}(h)\rightarrow&-\infty,\\
C_{6}(h)\rightarrow&-\frac{k}{\gamma^{2}\beta_{2}}\frac{1-q_{1}}{q_{2}-q_{1}}+\frac{k}{\gamma^{2}}\frac{1-q_{1}}{q_{2}-q_{1}}(\alpha-\lambda)(q_{2}-1)h,\\
C_{8}(h)\rightarrow&\frac{k}{\gamma^{2}\beta_{2}}\frac{1-q_{1}}{q_{2}-q_{1}}(e^{(1-\alpha)(q_{2}-1)\beta_{2}h}-1)+\frac{k}{\gamma^{2}}\frac{1-q_{1}}{q_{2}-q_{1}}(\alpha-\lambda)(q_{2}-1)h,\\
C_{7}(h)\rightarrow&(1-\alpha)^{q_{2}-q_{1}-1}(\alpha-\lambda)\frac{k}{\gamma^{2}}\frac{1-q_{1}}{q_{2}-q_{1}}\frac{q_{2}-1}{q_{2}-q_{1}}e^{-(1-\alpha)(q_{2}-q_{1})\beta_{2}h}\\&+(1-\alpha)^{q_{2}-q_{1}}\frac{k}{\gamma^{2}\beta_{2}}\frac{q_{2}-1}{q_{2}-q_{1}}e^{-(1-\alpha)(1-q_{1})\beta_{2}h},\\
C_{5}(h)\rightarrow&(1-\alpha)^{q_{2}-q_{1}-1}(\alpha-\lambda)\frac{k}{\gamma^{2}}\frac{1-q_{1}}{q_{2}-q_{1}}\frac{q_{2}-1}{q_{2}-q_{1}}e^{-(1-\alpha)(q_{2}-q_{1})\beta_{2}h}\\&+\left[(1-\alpha)^{q_{2}-q_{1}}-1\right]\frac{k}{\gamma^{2}\beta_{2}}\frac{q_{2}-1}{q_{2}-q_{1}}e^{-(1-\alpha)(1-q_{1})\beta_{2}h},\\
C_{3}(h)\rightarrow&-\infty,\\
C_{1}(h)\rightarrow&(1-\alpha)^{q_{2}-q_{1}-1}(\alpha-\lambda)\frac{k}{\gamma^{2}}\frac{1-q_{1}}{q_{2}-q_{1}}\frac{q_{2}-1}{q_{2}-q_{1}}e^{-(1-\alpha)(q_{2}-q_{1})\beta_{2}h}\\&+(1-\alpha)^{q_{2}-q_{1}}\frac{k}{\gamma^{2}\beta_{2}}\frac{q_{2}-1}{q_{2}-q_{1}}e^{-(1-\alpha)(1-q_{1})\beta_{2}h}\\&+\frac{k}{\gamma^{2}\beta_{2}}\frac{q_{2}-1}{q_{2}-q_{1}}+\frac{k}{\gamma^{2}}\frac{q_{2}-1}{q_{2}-q_{1}}(\alpha-\lambda)(1-q_{1})h.\\
\end{align*}
Denote the finite limits of $C_{i}(h),i=1,5,6,7,8$ by $C^{L}_{i}(h)$. Then we have, as $\beta_{1}\rightarrow 0$,

\begin{align*}
W_{\rm low}(h)\rightarrow&-C^{L}_{1}(h)q_{1}+\frac{\lambda h}{\gamma},\\
W_{\rm peak}(h)\rightarrow&-C^{L}_{5}(h)q_{1}e^{(1-\alpha)(1-q_{1})\beta_{2}h}-C^{L}_{6}(h)q_{2}e^{-(1-\alpha)(q_{2}-1)\beta_{2}h}-\frac{k}{\gamma^{2}\beta_{2}}+\frac{h}{\gamma},\\
W_{\rm updt}(h)\rightarrow&-C^{L}_{7}(h)q_{1}(1-\alpha)^{q_{1}-1}e^{(1-\alpha)(1-q_{1})\beta_{2}h}-C^{L}_{8}(h)q_{2}(1-\alpha)^{q_{2}-1}e^{-(1-\alpha)(q_{2}-1)\beta_{2}h}+\frac{h}{\gamma}.
\end{align*}
To show that $W_{\rm ref}(h)$ tends to the same limit $-C^{L}_{1}(h)q_{1}+\frac{\lambda h}{\gamma}$ as $W_{\rm low}(h)$, we prove $W_{\rm ref}(h)-W_{\rm low}(h)\rightarrow 0$. In fact,

\begin{align*}
W_{\rm ref}(h)-W_{\rm low}(h)=&q_{1}C_{1}(h)\left[e^{-(\alpha-\lambda)(1-q_{1})\beta_{1}h}-1\right]+\frac{(\alpha-\lambda)h}{\gamma}\\&+\frac{k}{\gamma^{2}\beta_{1}}\left[q_{1}\frac{q_{2}-1}{q_{2}-q_{1}}e^{(\alpha-\lambda)(1-q_{1})\beta_{1}h}+q_{2}\frac{1-q_{1}}{q_{2}-q_{1}}e^{-(\alpha-\lambda)(q_{2}-1)\beta_{1}h}-1\right]\\
\rightarrow&\frac{(\alpha-\lambda)h}{\gamma}+\frac{k}{\gamma^{2}}\left[q_{1}\frac{q_{2}-1}{q_{2}-q_{1}}(\alpha-\lambda)(1-q_{1})h-q_{2}\frac{1-q_{1}}{q_{2}-q_{1}}(\alpha-\lambda)(q_{2}-1)h\right]\\
=&0.
\end{align*}
\end{remark}

\vskip 15pt
\section{Conclusion}\label{concl}
We establish a new theoretical model focusing on the risky investment and consumption behavior of a sophisticated decision maker. We solve the optimal consumption and investment problem which maximizes the expected total discounted utility with running maximum related reference point and drawdown constraint. Mathematical analysis and computation illustrate that the optimal consumption and investment policy are of semi-explicit forms with five important thresholds classifying different ranks of people. Theoretical and numerical analysis of the solution and sensitivity analysis of the parameters are conducted as well. The results are of economic significance in the following aspects: the MPC out of wealth is generally decreasing but increasing with certain intermediate wealth levels, and it jumps inversely proportional to the risk aversion at the reference point; both DRRA and IRRA are possible and the implied relative risk aversion is roughly a smile in wealth; wealth shocks are more influential on the welfare of the poorer people. As a special feature of our model, risk aversion change results in significant changes in optimal strategies and the impact of local risk aversion change turns out to be global.
\vskip 15pt

{\bf Acknowledgements.}
The authors acknowledge the support from the National Natural Science Foundation of China (Grant No.11871036, and No.12271290). The authors also thank the members of the group of Actuarial Sciences and Mathematical Finance at the Department of Mathematical Sciences, Tsinghua University for their feedback and useful conversations. We are also particularly grateful to the two
anonymous reviewers and the associated editor whose suggestions helped us to greatly improve the quality of the manuscript.
\vskip 20pt
{\bf Data availability statement.} 
Data sharing not applicable to this article as no datasets were generated or analyzed during the current study.

\appendix
  \renewcommand{\appendixname}{Appendix~\Alph{section}}
	\section{Proof of the Verification Theorem}
	\label{appa}
\begin{proof}[Proof of Theorem \ref{verificationth}]
	Define
	
	\begin{equation*}
	\hat{H}_{t}(y)\triangleq h_{0}\vee\frac{1}{(1-\alpha)\beta_{2}}\ln(\frac{1-\alpha}{\inf\limits_{s\le t}Y_{s}(y)}).
	\end{equation*}
		Then, for any $(x_{0},h_{0})\in\mathcal{C}$ and any $y>0$, we have
\begin{eqnarray}
		\mathbb{E}_{x_{0},h_{0}}\int_{0}^{\infty}e^{-\gamma t}U(c_{t},h_{t})dt&=&\mathbb{E}_{x_{0},h_{0}}\int_{0}^{\infty}e^{-\gamma t}\big(U(c_{t},h_{t})-Y_{t}(y)c_{t}\big)dt+y\mathbb{E}_{x_{0},h_{0}}\int_{0}^{\infty}c_{t}M_{t}dt\notag
		\end{eqnarray}
\begin{eqnarray}
		&\le &\mathbb{E}_{x_{0},h_{0}}\int_{0}^{\infty}e^{-\gamma t}\tilde{U}\big(Y_{t}(y),H^{\dagger}_{t}(y)\big)dt+yx_{0}\notag\\
		&=&\mathbb{E}_{x_{0},h_{0}}\int_{0}^{\infty}e^{-\gamma t}\tilde{U}\big(Y_{t}(y),\hat{H}_{t}(y)\big)dt+yx_{0}\notag\\
		&=&\tilde{V}(y,h_{0})+yx_{0},\label{iqofv}
\end{eqnarray}
where the second, the third and the last line hold thanks to Lemma \ref{lemma1}, Lemma \ref{lemma2} and Lemma \ref{lemma4}. And equality holds with $c_{t}=c^{*}(Y_{t}(y),H^{\dagger}_{t}(y))$ and $y=y^{*}$ .
\vskip 5pt		
 Using the explicit expressions of $c^{*}(\cdot,\cdot)$, $Y_{t}(\cdot)$ and $\hat{H}_{t}(\cdot)$, we know that $c^{*}\big(Y_{t}(y),\hat{H}_{t}(y)\big)$ is strictly decreasing in $y$ with $\lim\limits_{y\rightarrow 0^{+}}c^{*}\big(Y_{t}(y),\hat{H}_{t}(y)\big)=\infty$ and $\lim\limits_{y\rightarrow \infty}c^{*}\big(Y_{t}(y),\hat{H}_{t}(y)\big)=\lambda h_{0}$. As such, there exists a unique $y$ to solve
 
		\begin{equation*}
		   \mE_{x_{0},h_{0}}\int_{0}^{\infty}c^{*}\big(Y_{t}(y),\hat{H}_{t}(y)\big)M_{t}dt=x_{0}.
		\end{equation*}
Hence, we deduce from Lemma \ref{lemma2}  that Eq.(\ref{ystar}) has a unique solution $y^{*}$. Using (\ref{iqofv}) yields

\begin{equation*}
		\inf\limits_{y>0}\big\{\tilde{V}(y,h_{0})+yx_{0}\big\}=\sup\limits_{(c,\pi)\in\mathcal{A}}\mathbb{E}_{x_{0},h_{0}}\int_{0}^{\infty}e^{-\gamma t}U(c_{t},h_{t})dt.
\end{equation*}
		To make the left hand side in which $\tilde{V}(\cdot,\cdot)$ is the solution of (\ref{ode}) equal the value function $V(x_{0},h_{0})$ which satisfies (\ref{HJB}), $\tilde{V}(\cdot,\cdot)$ must satisfy (\ref{dHJB}). Recall that (\ref{dHJB}) leads to (\ref{ode}) for $c=c^{*}(y,h)$ and $\pi=\pi^{*}(y,h)$ where $c^{*}(\cdot,\cdot)$ is given by (\ref{c}) which leads to (\ref{cstar}) when restricted to $\mathcal{C}_{d}$ and $\pi^{*}(\cdot,\cdot)$ is given by (\ref{pi}). Plugging the expression of $\tilde{V}(y,h)$ and $q_{i}^{2}-q_{i}=\frac{r}{k},\ i=1,2$ into (\ref{pi}) leads to (\ref{pistar}).
		
		Thus, the dual of $\tilde{V}(\cdot,\cdot)$ given by (\ref{tilv}) is actually the value function of (\ref{problem}) and that $\mathbb{E}_{x_{0},h_{0}}\int_{0}^{\infty}e^{-\gamma t}U(c_{t},h_{t})dt$ attains its maximum at $(c^{*},\pi^{*})$ given in Theorem \ref{verificationth}.
\end{proof}
\vskip 15pt
\section{Auxiliary Lemmas for Theorem \ref{verificationth} and  Their Proofs}
\label{appb}
The following four lemmas are needed in proving Theorem \ref{verificationth}.

		\begin{lemma}
		\label{lemma1}
		The inequality holds in (\ref{iqofv}) and it becomes equality with $c_{t}=c^{*}\big(Y_{t}(y),H^{\dagger}_{t}(y)\big)$ and $y=y^{*}$.
	\end{lemma}
	
		\begin{lemma}
	    \label{lemma2}
	    For any $y>0$ and any $t>0$, we have $H^{\dagger}_{t}(y)=\hat{H}_{t}(y)$.
	\end{lemma}
	
	\begin{lemma}[Transversality Condition]
		\label{lemma3}
		For any $y>0$,
		
		\begin{equation*}
		\lim\limits_{T\rightarrow\infty}\mathbb{E}_{x_{0},h_{0}}\Big[e^{-\gamma T}\tilde{V}\big(Y_{T}(y),\hat{H}_{T}(y)\big)\Big]=0.
		\end{equation*}
	\end{lemma}

	\begin{lemma}
		\label{lemma4}\quad
		
		\begin{equation*}
		\tilde{V}(y,h_{0})=\mathbb{E}_{x_{0},h_{0}}\int_{0}^{\infty}e^{-\gamma t}\tilde{U}\big(Y_{t}(y),\hat{H}_{t}(y)\big)dt.
		\end{equation*}
	\end{lemma}
As the proofs of the first two lemmas are similar to the proofs of Lemma 5.2 and Lemma 5.3 in \cite{deng2021}, we omit them here and   only prove Lemmas \ref{lemma3} and \ref{lemma4}.

\begin{proof}[Proof of Lemma \ref{lemma3}]
		Based on the definition of $Y_{t}(\cdot)$ and $\ \hat{H}_{t}(\cdot)$, we have
		
		\begin{align}
		&\lim\limits_{T\rightarrow\infty}Y_{T}(y)=\lim\limits_{T\rightarrow\infty}ye^{-\frac{(\mu-r)^{2}}{2\sigma^{2}}T-\frac{\mu-r}{\sigma}B_{T}}=0,\ a.s.,\label{YT}\\
		&\lim\limits_{T\rightarrow\infty}\hat{H}_{T}(y)=\lim\limits_{T\rightarrow\infty}h_{0}\vee\frac{1}{(1-\alpha)\beta_{2}}\ln(\frac{1-\alpha}{\inf\limits_{s\le T}Y_{s}(y)})=\infty,\ a.s..\label{HT}
		\end{align}
Using the expression of $\tilde{V}(\cdot,\cdot)$ yields

\begin{align}
\!\!\!\!\!\!\!\!		\lim\limits_{T\rightarrow\infty}&\mathbb{E}_{x_{0},h_{0}}\Big[e^{-\gamma T}\tilde{V}\big(Y_{T}(y),\hat{H}_{T}(y)\big)\Big]\notag\\
		&=\lim\limits_{T\rightarrow\infty}\mathbb{E}_{x_{0},h_{0}}\bigg\{e^{-\gamma T}\Big[C_{7}\big(\hat{H}_{T}(y)\big)Y_{T}(y)^{q_{1}}
+C_{8}\big(\hat{H}_{T}(y)\big)Y_{T}(y)^{q_{2}}-\frac{1}{\gamma}\hat{H}_{T}(y)Y_{T}(y)\notag\\
&+\frac{1}{\gamma\beta_{2}}\big(1-e^{-(1-\alpha)\beta_{2}\hat{H}_{T}(y)}\big)\Big]
1_{\big\{(1-\alpha)e^{-(1-\alpha)\beta_{2}\hat{H}_{T}(y)}\le Y_{T}(y)<e^{-(1-\alpha)\beta_{2}\hat{H}_{T}(y)}\big\}}\notag\\
		&+e^{-\gamma T}\Big[C_{5}\big(\hat{H}_{T}(y)\big)Y_{T}(y)^{q_{1}}+
C_{6}\big(\hat{H}_{T}(y)\big)Y_{T}(y)^{q_{2}}+\frac{k}{\gamma^{2}\beta_{2}}Y_{T}(y)\notag\\
&+\frac{1}{\gamma\beta_{2}}\Big(1-Y_{T}(y)+Y_{T}(y)\ln\big(Y_{T}(y)\big)\Big)-\frac{1}{\gamma}\alpha \hat{H}_{T}(y)Y_{T}(y)\Big]1_{\big\{e^{-(1-\alpha)\beta_{2}\hat{H}_{T}(y)}\le Y_{T}(y)<1\big\}}\bigg\}.
\label{traneq}
		\end{align}
By Proposition \ref{hinfty}, we have

$C_{7}\big(\hat{H}_{T}(y)\big)=\mathcal{O}(e^{-(1-\alpha)(1-q_{1})\beta_{2}\hat{H}_{T}(y)})$, a.s. as $T\rightarrow\infty$, as such,
		\begin{align*}
		   &\mathbb{E}_{x_{0},h_{0}}e^{-\gamma T}C_{7}\big(\hat{H}_{T}(y)\big)Y_{T}(y)^{q_{1}}1_{\big\{(1-\alpha)e^{-(1-\alpha)\beta_{2}\hat{H}_{T}(y)}\le Y_{T}(y)<e^{-(1-\alpha)\beta_{2}\hat{H}_{T}(y)}\big\}}\\
		   &=\mathcal{O}\bigg(e^{-\gamma T}\mathbb{E}_{x_{0},h_{0}}e^{-(1-\alpha)(1-q_{1})\beta_{2}\hat{H}_{T}(y)}Y_{T}(y)^{q_{1}}
1_{\big\{(1-\alpha)e^{-(1-\alpha)\beta_{2}\hat{H}_{T}(y)}\le Y_{T}(y)<e^{-(1-\alpha)\beta_{2}\hat{H}_{T}(y)}\big\}}\bigg)\\
		   &=\mathcal{O}\bigg(e^{-\gamma T}\mathbb{E}_{x_{0},h_{0}}Y_{T}(y)^{1-q_{1}+q_{1}}\bigg)\\
		   &=\mathcal{O}\bigg(e^{-\gamma T}\mathbb{E}_{x_{0},h_{0}}e^{-\frac{(\mu-r)^{2}}{2\sigma^{2}}T-\frac{\mu-r}{\sigma}B_{T}}\bigg)\\
		   &=\mathcal{O}\bigg(e^{-\gamma T}\bigg),\ \mbox{ as $T\rightarrow\infty$,}
		\end{align*}
then
\begin{equation}
		\label{term1}		   \lim\limits_{T\rightarrow\infty}\mathbb{E}_{x_{0},h_{0}}e^{-\gamma T}C_{7}\big(\hat{H}_{T}(y)\big)Y_{T}(y)^{q_{1}}1_{\big\{(1-\alpha)e^{-(1-\alpha)\beta_{2}\hat{H}_{T}(y)}\le Y_{T}(y)<e^{-(1-\alpha)\beta_{2}\hat{H}_{T}(y)}\big\}}=0.
		\end{equation}
Similarly, 	

		\begin{align*}
		   &\mathbb{E}_{x_{0},h_{0}}e^{-\gamma T}C_{8}\big(\hat{H}_{T}(y)\big)Y_{T}(y)^{q_{2}}1_{\big\{(1-\alpha)e^{-(1-\alpha)\beta_{2}\hat{H}_{T}(y)}\le Y_{T}(y)<e^{-(1-\alpha)\beta_{2}\hat{H}_{T}(y)}\big\}}\\
		   &=\mathcal{O}\bigg(e^{-\gamma T}\mathbb{E}_{x_{0},h_{0}}e^{(1-\alpha)(q_{2}-1)\beta_{2}\hat{H}_{T}(y)}Y_{T}(y)^{q_{2}}1_{\big\{(1-\alpha)e^{-(1-\alpha)\beta_{2}\hat{H}_{T}(y)}\le Y_{T}(y)<e^{-(1-\alpha)\beta_{2}\hat{H}_{T}(y)}\big\}}\bigg)\\
		   &=\mathcal{O}\bigg(e^{-\gamma T}\mathbb{E}_{x_{0},h_{0}}Y_{T}(y)^{-(q_{2}-1)+q_{2}}\bigg)\\
&=\mathcal{O}\bigg(e^{-\gamma T}\mathbb{E}_{x_{0},h_{0}}e^{-\frac{(\mu-r)^{2}}{2\sigma^{2}}T-\frac{\mu-r}{\sigma}B_{T}}\bigg)
		   =\mathcal{O}\bigg(e^{-\gamma T}\bigg), \ \mbox{as $T\rightarrow\infty$, }
		\end{align*}
		then
		\begin{equation}
		\label{term2}
		   \lim\limits_{T\rightarrow\infty}\mathbb{E}_{x_{0},h_{0}}e^{-\gamma T}C_{8}\big(\hat{H}_{T}(y)\big)Y_{T}(y)^{q_{2}}1_{\big\{(1-\alpha)e^{-(1-\alpha)\beta_{2}\hat{H}_{T}(y)}\le Y_{T}(y)<e^{-(1-\alpha)\beta_{2}\hat{H}_{T}(y)}\big\}}=0.
		\end{equation}
Using (\ref{HT}), as $T\rightarrow\infty$,

		\begin{equation*}
		   e^{-\gamma T}\frac{1}{\gamma\beta_{2}}\big(1-e^{-(1-\alpha)\beta_{2}\hat{H}_{T}(y)}\big)
1_{\big\{(1-\alpha)e^{-(1-\alpha)\beta_{2}\hat{H}_{T}(y)}\le Y_{T}(y)<e^{-(1-\alpha)\beta_{2}\hat{H}_{T}(y)}\big\}}=\mathcal{O}\big(e^{-\gamma T}\big),\ a.s.,
		\end{equation*}
from which we get
			\begin{equation}
			\label{term3}
		   \lim\limits_{T\rightarrow\infty}\mathbb{E}_{x_{0},h_{0}}e^{-\gamma T}\frac{1}{\gamma\beta_{2}}\big(1-e^{-(1-\alpha)\beta_{2}\hat{H}_{T}(y)}\big)1_{\big\{(1-\alpha)e^{-(1-\alpha)\beta_{2}\hat{H}_{T}(y)}\le Y_{T}(y)<e^{-(1-\alpha)\beta_{2}\hat{H}_{T}(y)}\big\}}=0.
		\end{equation}
By the same way as in  deriving  (\ref{term1}), we have
		\begin{equation}
		\label{term4}
		  \lim\limits_{T\rightarrow\infty}\mathbb{E}_{x_{0},h_{0}}e^{-\gamma T}C_{5}\big(\hat{H}_{T}(y)\big)Y_{T}(y)^{q_{1}}1_{\big\{e^{-(1-\alpha)\beta_{2}\hat{H}_{T}(y)}\le Y_{T}(y)<1\big\}}=0.
		\end{equation}
By Proposition \ref{hinfty}, we have $C_{6}\big(\hat{H}_{T}(y)\big)=\mathcal{O}\big(1\big)$, a.s. as $T\rightarrow\infty$. The facts (\ref{YT}) and $q_{2}>1$ yield
		\begin{equation}
		\label{term5}
		  \lim\limits_{T\rightarrow\infty}\mathbb{E}_{x_{0},h_{0}}e^{-\gamma T}C_{6}\big(\hat{H}_{T}(y)\big)Y_{T}(y)^{q_{2}}1_{\big\{e^{-(1-\alpha)\beta_{2}\hat{H}_{T}(y)}\le Y_{T}(y)<1\big\}}=0.
		\end{equation}
and		
\begin{equation}
		\label{term6}
		  \lim\limits_{T\rightarrow\infty}\mathbb{E}_{x_{0},h_{0}}e^{-\gamma T}\frac{k}{\gamma^{2}\beta_{2}}Y_{T}(y)1_{\big\{e^{-(1-\alpha)\beta_{2}\hat{H}_{T}(y)}\le Y_{T}(y)<1\big\}}=0.
		\end{equation}
Using (\ref{YT}) again  and the fact  $\lim\limits_{y\rightarrow 0^{+}}(1-y+y\ln y)=1$, we obtain
		\begin{equation}
		\label{term7}
		  \lim\limits_{T\rightarrow\infty}\mathbb{E}_{x_{0},h_{0}}e^{-\gamma T}\frac{1}{\gamma\beta_{2}}\Big[1-Y_{T}(y)+Y_{T}(y)\ln\big(Y_{T}(y)\big)\Big]1_{\big\{e^{-(1-\alpha)\beta_{2}\hat{H}_{T}(y)}\le Y_{T}(y)<1\big\}}=0.
		\end{equation}
Based on Girsanov's theorem, we have

		\begin{align*}
		    &e^{-\gamma T}\mathbb{E}_{x_{0},h_{0}}\hat{H}_{T}(y)Y_{T}(y)\\
		    &=\mathcal{O}\bigg(e^{-\gamma T}\mathbb{E}_{x_{0},h_{0}}\Big[\sup\limits_{s\le T}\big\{\frac{(\mu-r)^{2}}{2\sigma^{2}}s+\frac{\mu-r}{\sigma}B_{s}\big\}
e^{-\frac{(\mu-r)^{2}}{2\sigma^{2}}T-\frac{\mu-r}{\sigma}B_{T}}\Big]\bigg)\\
		   &=\mathcal{O}\bigg(\frac{\mu-r}{\sigma}e^{-\gamma T}\Big\{\sqrt{\frac{T}{2\pi}}e^{-\frac{(\mu-r)^{2}}{8\sigma^{2}}T}-\frac{\mu-r}{2\sigma}T\Phi(-\frac{\mu-r}{2\sigma}\sqrt{T})\\&\qquad+\frac{\sigma}{\mu-r}\big[\Phi(\frac{\mu-r}{2\sigma}\sqrt{T})-\Phi(-\frac{\mu-r}{2\sigma}\sqrt{T})\big]\Big\}\bigg).
		\end{align*}
		It follows that
		\begin{equation}
		\label{term8}
		   \lim\limits_{T\rightarrow\infty}\mathbb{E}_{x_{0},h_{0}}e^{-\gamma T}\hat{H}_{T}(y)Y_{T}(y)=0.
		\end{equation}
 Thus, using (\ref{traneq})-(\ref{term8}), we have
 
		\begin{equation*}
		  \lim\limits_{T\rightarrow\infty}\mathbb{E}_{x_{0},h_{0}}\Big[e^{-\gamma T}\tilde{V}\big(Y_{T}(y),\hat{H}_{T}(y)\big)\Big]=0.
		\end{equation*}
\end{proof}
\begin{proof}[Proof of Lemma \ref{lemma4}]
			Applying Eq. (\ref{ode}) and It\^{o}'s rule, we obtain
			
			\begin{align}
			d\left\{e^{-\gamma t}\tilde{V}\big[Y_{t}(y),\hat{H}_{t}(y)\big]\right\}=&-e^{-\gamma t}\tilde{U}\big[Y_{t}(y),\hat{H}_{t}(y)\big]dt\notag\\
&-\frac{\mu-r}{\sigma}e^{-\gamma t}\tilde{V}_{y}\big[Y_{t}(y),\hat{H}_{t}(y)\big]Y_{t}(y)dB_{t}\notag\\
			&+e^{-\gamma t}\tilde{V}_{h}\big[Y_{t}(y),\hat{H}_{t}(y)\big]d\hat{H}_{t}(y).
			\label{eq}
			\end{align}
Define the stopping times: $\forall n\ge 1$,

			\begin{equation*}
			\tau_{n}\triangleq\inf\Big\{t\ge 0\Big|Y_{t}(y)\ge n\ {\rm or}\ \hat{H}_{t}(y)\ge \frac{1}{(1-\alpha)\beta_{2}}\ln\big[(1-\alpha)n\big]\Big\}.
			\end{equation*}
 It follows that $\lim\limits_{n\rightarrow\infty}\tau_{n}=\infty$ and for $ \forall$ $n>\frac{1}{1-\alpha}e^{(1-\alpha)\beta_{2}h_{0}}$, $ \forall$ $\kappa\ge 1$ and  $ \forall$  $T>0$,
 
			\begin{align}
			\mathbb{E}_{x_{0},h_{0}}1_{\{\tau_{n}\le T\}}&\le \mathbb{P}_{x_{0},h_{0}}\Big(\big\{\sup\limits_{t\in [0,T]}Y_{t}(y)\ge n\big\}\bigcup\big\{\sup\limits_{t\in [0,T]}\hat{H}_{t}(y)\ge \frac{1}{(1-\alpha)\beta_{2}}\ln\big[(1-\alpha)n\big]\big\}\Big)\notag\\
			&\le \mathbb{P}_{x_{0},h_{0}}\big(\sup\limits_{t\in [0,T]}Y_{t}(y)\ge n\big)+\mathbb{P}_{x_{0},h_{0}}\big(\inf\limits_{t\in [0,T]}Y_{t}(y)\le \frac{1}{n}\big)\notag\\
			&=\mathbb{P}_{x_{0},h_{0}}\big(\sup\limits_{t\in [0,T]}Y_{t}(y)\ge n\big)+\mathbb{P}_{x_{0},h_{0}}\big(\sup\limits_{t\in [0,T]}Y_{t}(y)^{-1}\ge n\big)\notag\\
			&\le n^{-2\kappa}\mathbb{E}_{x_{0},h_{0}}\sup_{t\in [0,T]}Y_{t}(y)^{2\kappa}+n^{-2\kappa}\mathbb{E}_{x_{0},h_{0}}\sup_{t\in [0,T]}Y_{t}(y)^{-2\kappa}\notag\\
			&=\mathcal{O}\big(n^{-2\kappa}(1+y^{2\kappa})e^{MT}\big)
			\label{esttaun}
			\end{align}
for some constants $M$.
			
Integrating (\ref{eq}) from $0$ to $T\wedge\tau_{n}$ and taking expectation on both sides, we obtain

			\begin{align}
			\tilde{V}(y,h_{0})=&\mathbb{E}_{x_{0},h_{0}}\Big[e^{-\gamma T\wedge\tau_{n}}\tilde{V}\big(Y_{T\wedge\tau_{n}}(y),\hat{H}_{T\wedge\tau_{n}}(y)\big)\Big]\notag\\
&+\mathbb{E}_{x_{0},h_{0}}\int_{0}^{T\wedge\tau_{n}}e^{-\gamma t}\tilde{U}\big(Y_{t}(y),\hat{H}_{t}(y)\big)dt\notag\\
			&+\mathbb{E}_{x_{0},h_{0}}\int_{0}^{T\wedge\tau_{n}}\frac{\mu-r}{\sigma}e^{-\gamma t}\tilde{V}_{y}\big(Y_{t}(y),\hat{H}_{t}(y)\big)Y_{t}(y)dB_{t}\notag\\
			&-\mathbb{E}_{x_{0},h_{0}}\int_{0}^{T\wedge\tau_{n}}e^{-\gamma t}\tilde{V}_{h}\big(Y_{t}(y),\hat{H}_{t}(y)\big)d\hat{H}_{t}(y).
			\label{formular1}
			\end{align}
			The four terms on the right hand side of (\ref{formular1}) can be evaluated as follows:
			
			The first term in (\ref{formular1}) can be split into two parts:
			
\begin{align}
			&\mathbb{E}_{x_{0},h_{0}}\Big[e^{-\gamma T\wedge\tau_{n}}\tilde{V}\big(Y_{T\wedge\tau_{n}}(y),\hat{H}_{T\wedge\tau_{n}}(y)\big)\Big]\notag\\
			&=\mathbb{E}_{x_{0},h_{0}}\Big[e^{-\gamma T}\tilde{V}\big(Y_{T}(y),\hat{H}_{T}(y)\big)1_{\{T\le\tau_{n}\}}\Big]\notag\\
&+\mathbb{E}_{x_{0},h_{0}}\Big[e^{-\gamma \tau_{n}}\tilde{V}\big(Y_{\tau_{n}}(y),\hat{H}_{\tau_{n}}(y)\big)1_{\{T>\tau_{n}\}}\Big].
			\label{firstterm}
\end{align}
Based on Lemma \ref{lemma3}, as $n\uparrow\infty$, the first term in (\ref{firstterm}) tends to $\mathbb{E}_{x_{0},h_{0}}\Big[e^{-\gamma T}\tilde{V}\big(Y_{T}(y),\hat{H}_{T}(y)\big)\Big]$ , which converges to zero as $T\uparrow\infty$. To handle the second term in (\ref{firstterm}), we first observe from the definition of $\tau_{n}$ that either

$$Y_{\tau_{n}}(y)\le n,\ \hat{H}_{\tau_{n}}(y)=\frac{1}{(1-\alpha)\beta_{2}}\ln\big((1-\alpha)n\big)$$
or

$$Y_{\tau_{n}}(y)=n,\ \hat{H}_{\tau_{n}}(y)\le \frac{1}{(1-\alpha)\beta_{2}}\ln\big((1-\alpha)n\big)$$ holds. In addition, it follows from the fact  $\big(Y_{\tau_{n}}(y),\hat{H}_{\tau_{n}}(y)\big)\in \mathcal{C}_{d}$ that either

$$\frac{1}{n}\le Y_{\tau_{n}}(y)\le n,\ \hat{H}_{\tau_{n}}(y)=\frac{1}{(1-\alpha)\beta_{2}}\ln\big((1-\alpha)n\big)$$
or

$$Y_{\tau_{n}}(y)=n,\ 0<\hat{H}_{\tau_{n}}(y)\le \frac{1}{(1-\alpha)\beta_{2}}\ln\big((1-\alpha)n\big)$$ holds.\\

Applying Proposition \ref{hinfty}  with (\ref{tilv}),  we obtain the  order estimate of $\tilde{V}\big(Y_{\tau_{n}}(y),\hat{H}_{\tau_{n}}(y)\big)$ as follows:

For $\frac{1}{n}\le Y_{\tau_{n}}(y)\le n,\ \hat{H}_{\tau_{n}}(y)=\frac{1}{(1-\alpha)\beta_{2}}\ln\big[(1-\alpha)n\big]$, we have, for sufficiently large $n$,

$$\hat{H}_{\tau_{n}}(y)=\mathcal{O}\big(\ln n\big),\ e^{\hat{H}_{\tau_{n}}(y)}=\mathcal{O}\big(n^{\frac{1}{(1-\alpha)\beta_{2}}}\big).$$
If $(1-\alpha)e^{-(1-\alpha)\beta_{2}\hat{H}_{\tau_{n}}(y)}\le Y_{\tau_{n}}(y)<e^{-(1-\alpha)\beta_{2}\hat{H}_{\tau_{n}}(y)}$, then we have,  for sufficiently large $n$, 

$$\tilde{V}\big(Y_{\tau_{n}}(y),\hat{H}_{\tau_{n}}(y)\big)=\mathcal{O}\big(n^{-q_{1}}\big).$$
If $e^{-(1-\alpha)\beta_{2}\hat{H}_{\tau_{n}}(y)}\le Y_{\tau_{n}}(y)<1$, then we have, for sufficiently large $n$,

$$\tilde{V}\big(Y_{\tau_{n}}(y),\hat{H}_{\tau_{n}}(y)\big)=\mathcal{O}\big(n^{-q_{1}}\big).$$
 If $1\le Y_{\tau_{n}}(y)<e^{(\alpha-\lambda)\beta_{1}\hat{H}_{\tau_{n}}(y)}$, then we have, for sufficiently large $n$,
 
 $$\tilde{V}\big(Y_{\tau_{n}}(y),\hat{H}_{\tau_{n}}(y)\big)=\mathcal{O}\big(n^{q_{2}}\big).$$
 If $e^{(\alpha-\lambda)\beta_{1}\hat{H}_{\tau_{n}}(y)}\le Y_{\tau_{n}}(y)$, then we have, for sufficiently large $n$,
 
$$\tilde{V}\big(Y_{\tau_{n}}(y),\hat{H}_{\tau_{n}}(y)\big)=\mathcal{O}
\big(n^{\frac{(\alpha-\lambda)\beta_{1}}{(1-\alpha)\beta_{2}}}\big).$$
			
For $Y_{\tau_{n}}(y)=n,\ 0<\hat{H}_{\tau_{n}}(y)\le \frac{1}{(1-\alpha)\beta_{2}}\ln\big[(1-\alpha)n\big]$, we have,  for sufficiently large $n$, either 

$$1\le Y_{\tau_{n}}(y)<e^{(\alpha-\lambda)\beta_{1}\hat{H}_{\tau_{n}}(y)}$$ or

$$e^{(\alpha-\lambda)\beta_{1}\hat{H}_{\tau_{n}}(y)}\le Y_{\tau_{n}}(y).$$
 If $1\le Y_{\tau_{n}}(y)<e^{(\alpha-\lambda)\beta_{1}\hat{H}_{\tau_{n}}(y)}$, then we have, for sufficiently large $n$,
 
 $$\tilde{V}\big(Y_{\tau_{n}}(y),\hat{H}_{\tau_{n}}(y)\big)=\mathcal{O}\big(n^{q_{2}}\big).$$
If $e^{(\alpha-\lambda)\beta_{1}\hat{H}_{\tau_{n}}(y)}\le Y_{\tau_{n}}(y)$, then we have, for sufficiently large $n$,

$$\tilde{V}\big(Y_{\tau_{n}}(y),\hat{H}_{\tau_{n}}(y)\big)=
\mathcal{O}\big(n^{q_{1}+(1-q_{1})\frac{(\alpha-\lambda)\beta_{1}}{(1-\alpha)\beta_{2}}}\big).$$
			
In summary, we have $\tilde{V}\big(Y_{\tau_{n}}(y),\hat{H}_{\tau_{n}}(y)\big)=\mathcal{O}\big(n^{(-q_{1})\vee q_{2}\vee(1-q_{1})\frac{(\alpha-\lambda)\beta_{1}}{(1-\alpha)\beta_{2}}}\big)$. Applying (\ref{esttaun}) with $\kappa\geq \frac{1}{2}\big((-q_{1})\vee q_{2}\vee(1-q_{1})\frac{(\alpha-\lambda)\beta_{1}}{(1-\alpha)\beta_{2}}\big)$, we deduce that the second term in (\ref{firstterm}) converges to zero as $n\uparrow\infty$. Hence the first term in (\ref{formular1}) tends to zero by first letting $n\uparrow\infty$ and then $T\uparrow\infty$.
Because $\tau_{n}$ tends to $\infty$ as $n\rightarrow \infty$,  based on monotone convergence theorem, we obtain  that the second term in (\ref{formular1}) tends to $\mathbb{E}_{x_{0},h_{0}}\int_{0}^{\infty}e^{-\gamma t}\tilde{U}\big(Y_{t}(y),\hat{H}_{t}(y)\big)dt$, as $n\rightarrow \infty$ and $T\rightarrow \infty$. The third term in (\ref{formular1}) vanishes because the integral is a martingale. If $\hat{H}_{t}(y)$ strictly increases, then $Y_{t}(y)$ must be strictly decreasing, hence $\big(Y_{t}(y),\hat{H}_{t}(y)\big)$ is on the boundary. Using
the boundary condition, we have $\tilde{V}_{h}\big(Y_{t}(y),\hat{H}_{t}(y)\big)=0$, as such, the last term in (\ref{formular1}) vanishes. Thus, the proof follows.
\end{proof}
\vskip 15pt		
		\section{Proofs of Other Results in Section 4}	\label{appc}
\begin{proof}[Proof of Lemma \ref{vyypo}]
	
\vskip 5pt
If $e^{(\alpha-\lambda)\beta_{1}h}\le y$, then

	    \begin{equation*}
	        \tilde{V}_{yy}(y,h)=C_{1}(h)\frac{r}{k}y^{q_{1}-2}+C_{2}(h)\frac{r}{k}y^{q_{2}-2}.
	    \end{equation*}
As $C_{1}(h)>0$ and $C_{2}(h)=0$, we have $\tilde{V}_{yy}(y,h)>0$ for $e^{(\alpha-\lambda)\beta_{1}h}\le y$.
	
If $1\le y<e^{(\alpha-\lambda)\beta_{1}h}$, then

	    \begin{equation*}
	        y\tilde{V}_{yy}(y,h)=C_{3}(h)\frac{r}{k}y^{q_{1}-1}+C_{4}(h)\frac{r}{k}y^{q_{2}-1}+\frac{1}{\gamma\beta_{1}}.
	    \end{equation*}
	    Let $\psi(y)=y\tilde{V}_{yy}(y,h)$, then
	    
	    $$\psi'(y)=C_{3}(h)\frac{r}{k}(q_{1}-1)y^{q_{1}-2}+C_{4}(h)\frac{r}{k}(q_{2}-1)y^{q_{2}-2}.$$ Noting that $C_{4}(h)<0$,  $\psi(y)$ is either increasing, decreasing or first increasing then decreasing, we only need to show $\psi(1)>0$ and $\psi(e^{(\alpha-\lambda)\beta_{1}h})>0$. Precisely, using $C_{7}(h)>0$,
	    
\begin{align*}
	        \psi(1)=&C_{3}(h)\frac{r}{k}+C_{4}(h)\frac{r}{k}+\frac{1}{\gamma\beta_{1}}\\
	        =&\frac{r}{k}C_{7}(h)+\frac{1}{\gamma\beta_{1}}
\frac{1-q_{1}}{q_{2}-q_{1}}\left [1-e^{-(\alpha-\lambda)(q_{2}-1)\beta_{1}h}\right ]\\
&+\frac{1}{\gamma\beta_{2}}\frac{q_{2}-1}{q_{2}-q_{1}}\left[1-e^{-(1-\alpha)(1-q_{1})\beta_{2}h}\right]\\
	       >&0
\end{align*}
and

\begin{align*}	        \psi(e^{(\alpha-\lambda)\beta_{1}h})=&C_{3}(h)\frac{r}{k}e^{(q_{1}-1)
(\alpha-\lambda)\beta_{1}h}+C_{4}(h)\frac{r}{k}e^{(q_{2}-1)(\alpha-\lambda)
\beta_{1}h}+\frac{1}{\gamma\beta_{1}}\\
=&\frac{r}{k}C_{7}(h)e^{(q_{1}-1)(\alpha-\lambda)\beta_{1}h}+
\frac{1}{\gamma\beta_{1}}\frac{q_{2}-1}{q_{2}-q_{1}}\left[1-e^{-(\alpha-\lambda)(1-q_{1})\beta_{1}h}\right]\\
&+\frac{1}{\gamma\beta_{2}}\frac{q_{2}-1}{q_{2}-q_{1}}
\left[1-e^{-(1-\alpha)(1-q_{1})\beta_{2}h}\right]e^{(q_{2}-1)(\alpha-\lambda)\beta_{1}h}\\
>&0.
\end{align*}
Thus $\tilde{V}_{yy}(y,h)>0$ for $1\le y<e^{(\alpha-\lambda)\beta_{1}h}$.
\vskip 5pt	
If $e^{-(1-\alpha)\beta_{2}h}\le y<1$, then

\begin{align*}
	        y\tilde{V}_{yy}(y,h)=&C_{5}(h)\frac{r}{k}y^{q_{1}-1}+C_{6}(h)\frac{r}{k}y^{q_{2}-1}+\frac{1}{\gamma\beta_{2}}\\
	       =&\left[C_{5}(h)\frac{r}{k}y^{q_{1}-1}+\frac{1}{\gamma\beta_{2}}\frac{q_{2}-1}{q_{2}-q_{1}}\right]
+\left[C_{6}(h)\frac{r}{k}y^{q_{2}-1}+\frac{1}{\gamma\beta_{2}}\frac{1-q_{1}}{q_{2}-q_{1}}\right].
\end{align*}
For any fixed $h>0$, if $C_{5}(h)\ge 0$, then $C_{5}(h)\frac{r}{k}y^{q_{1}-1}+\frac{1}{\gamma\beta_{2}}\frac{q_{2}-1}{q_{2}-q_{1}}>0$; If $C_{5}(h)<0$, then $C_{5}(h)\frac{r}{k}y^{q_{1}-1}+\frac{1}{\gamma\beta_{2}}\frac{q_{2}-1}{q_{2}-q_{1}}$ is increasing in $y$ and

\begin{align*}
	        &C_{5}(h)\frac{r}{k}y^{q_{1}-1}+\frac{1}{\gamma\beta_{2}}\frac{q_{2}-1}{q_{2}-q_{1}}\\\ge& C_{5}(h)\frac{r}{k}e^{(1-\alpha)(1-q_{1})\beta_{2}h}+\frac{1}{\gamma\beta_{2}}\frac{q_{2}-1}{q_{2}-q_{1}}\\
	        =&\frac{(1-\alpha)^{q_{2}-q_{1}}\frac{1}{\gamma}\frac{1-q_{1}}{q_{2}-q_{1}}(\alpha-\lambda)(q_{2}-1)}{(1-\alpha)(q_{2}-q_{1})\beta_{2}+(\alpha-\lambda)(q_{2}-1)\beta_{1}}e^{-[(1-\alpha)(q_{2}-1)\beta_{2}+(\alpha-\lambda)(q_{2}-1)\beta_{1}]h}\\
	        &+(1-\alpha)^{q_{2}-q_{1}}\frac{1}{\gamma\beta_{2}}\frac{q_{2}-1}{q_{2}-q_{1}}\\
	        >&0.
	    \end{align*}
Similarly,  we have

	    \begin{equation*}
	        C_{6}(h)\frac{r}{k}y^{q_{2}-1}+\frac{1}{\gamma\beta_{2}}\frac{1-q_{1}}{q_{2}-q_{1}}>0.
	    \end{equation*}
It follows that $\tilde{V}_{yy}(y,h)>0$ for $e^{-(1-\alpha)\beta_{2}h}\le y<1$.
\vskip 5pt
Finally, If $(1-\alpha)e^{-(1-\alpha)\beta_{2}h}\le y<e^{-(1-\alpha)\beta_{2}h}$, then

	    \begin{equation*}
	        \tilde{V}_{yy}(y,h)=C_{7}(h)\frac{r}{k}y^{q_{1}-2}+C_{8}(h)\frac{r}{k}y^{q_{2}-2}.
	    \end{equation*}
As $C_{7}(h)>0$ and

	    \begin{equation*}
	        C_{8}(h)=\frac{k}{\gamma^{2}\beta_{1}}\frac{1-q_{1}}{q_{2}-q_{1}}
\big[1-e^{-(\alpha-\lambda)(q_{2}-1)\beta_{1}h}\big]+\frac{k}{\gamma^{2}\beta_{2}}
\frac{1-q_{1}}{q_{2}-q_{1}}\big[e^{(1-\alpha)(q_{2}-1)\beta_{2}h}-1\big]>0,
	    \end{equation*}
we have  $\tilde{V}_{yy}(y,h)>0$ for $(1-\alpha)e^{-(1-\alpha)\beta_{2}h}\le y<e^{-(1-\alpha)\beta_{2}h}$. Thus the proof is complete.
\end{proof}

\begin{proof}[Proof of Lemma \ref{lipsc}]
	    By Lemma \ref{fc1}, $f$ is $C^{1}$ and continuous at the boundaries, as such, using the fact that $C_{i}(h),\ 1\le i\le 8$, are $C^{1}$, we conclude that $c^{*}_{\rm primal}$ and $\pi^{*}_{\rm primal}$ given in Theorem \ref{optpolicy} are locally Lipschitz on $\mathcal{C}$.
		
		Now to prove that $\pi^{*}_{\rm primal}$ is Lipschitz, we only need to show that $\frac{\partial\pi^{*}_{\rm primal}}{\partial x}$ and $\frac{\partial\pi^{*}_{\rm primal}}{\partial h}$ are both bounded.
		
		Step 1: $\frac{\partial\pi^{*}_{\rm primal}}{\partial x}$ is bounded:
		
		By definition of $\pi^{*}_{\rm primal}$ given in Theorem \ref{optpolicy}, we have
		
		\begin{equation*}
		   \frac{\partial\pi^{*}_{\rm primal}}{\partial x}=\frac{\mu-r}{\sigma^{2}}\frac{r}{k}\left\{
		\begin{array}{l}
		C_{1}(h)(q_{1}-1)\big[f_{1}(x,h)\big]^{q_{1}-2}\frac{\partial f_{1}(x,h)}{\partial x}\\
+C_{2}(h)(q_{2}-1)\big[f_{1}(x,h)\big]^{q_{2}-2}\frac{\partial f_{1}(x,h)}{\partial x},\hskip 10pt \frac{\lambda h}{\gamma}\le x\le W_{\rm low}(h),\\
	C_{3}(h)(q_{1}-1)\big[f_{2}(x,h)\big]^{q_{1}-2}\frac{\partial f_{2}(x,h)}{\partial x}\\
+C_{4}(h)(q_{2}-1)\big[f_{2}(x,h)\big]^{q_{2}-2}\frac{\partial f_{2}(x,h)}{\partial x},\hskip 10pt W_{\rm low}(h)<x\le W_{\rm ref}(h),\\
	C_{5}(h)(q_{1}-1)\big[f_{3}(x,h)\big]^{q_{1}-2}\frac{\partial f_{3}(x,h)}{\partial x}\\
+C_{6}(h)(q_{2}-1)\big[f_{3}(x,h)\big]^{q_{2}-2}\frac{\partial f_{3}(x,h)}{\partial x},\hskip 10pt W_{\rm ref}(h)<x\le W_{\rm peak}(h),\\
	C_{7}(h)(q_{1}-1)\big[f_{4}(x,h)\big]^{q_{1}-2}\frac{\partial f_{4}(x,h)}{\partial x}\\
+C_{8}(h)(q_{2}-1)\big[f_{4}(x,h)\big]^{q_{2}-2}\frac{\partial f_{4}(x,h)}{\partial x},\hskip 10pt W_{\rm peak}(h)<x\le W_{\rm updt}(h).
	\end{array}
	\right.
		\end{equation*}
Differentiating (\ref{f1}), we obtain
\begin{eqnarray}
    1&=&-C_{1}(h)q_{1}(q_{1}-1)\big[f_{1}(x,h)\big]^{q_{1}-2}\frac{\partial f_{1}(x,h)}{\partial x}\nonumber\\
    &&-C_{2}(h)q_{2}(q_{2}-1)\big[f_{1}(x,h)\big]^{q_{2}-2}\frac{\partial f_{1}(x,h)}{\partial x}.\label{pf1}
\end{eqnarray}
It follows from (\ref{fx}) and (\ref{pf1}) that,  for $\frac{\lambda h}{\gamma}\le x\le W_{\rm low}(h)$,

\begin{eqnarray*}
		   \frac{\partial\pi^{*}_{\rm primal}}{\partial x}&=&\frac{\mu-r}{\sigma^{2}}\left[(1-q_{2})+\frac{r}{k}(q_{1}-q_{2})C_{1}(h)
\big[f_{1}(x,h)\big]^{q_{1}-2}\frac{\partial f_{1}(x,h)}{\partial x}\right]\\
		   &=&\frac{\mu-r}{\sigma^{2}}\left[(1-q_{2})+\frac{A_{1}(x,h)}{B_{1}(x,h)}\right],
		\end{eqnarray*}
		where
\begin{eqnarray*}
A_{1}(x,h)&=&\frac{r}{k}(q_{1}-q_{2})C_{1}(h)\left[f_{1}(x,h)\right]^{q_{1}-2},\\		 B_{1}(x,h)&=&\frac{r}{k}\left\{-C_{1}(h)\big[f_{1}(x,h)\big]^{q_{1}-2}-C_{2}(h)\big[f_{1}(x,h)\big]^{q_{2}-2}\right\}.
\end{eqnarray*}
 As $C_{2}(h)=0$, $\frac{\partial\pi^{*}_{\rm primal}}{\partial x}$ is constant for $\frac{\lambda h}{\gamma}\le x\le W_{\rm low}(h)$. Differentiating (\ref{f2}), we obtain
\begin{eqnarray}
1&=&-C_{3}(h)q_{1}(q_{1}-1)\big[f_{2}(x,h)\big]^{q_{1}-2}\frac{\partial f_{2}(x,h)}{\partial x}\nonumber\\
&&-C_{4}(h)q_{2}(q_{2}-1)\big[f_{2}(x,h)\big]^{q_{2}-2}\frac{\partial f_{2}(x,h)}{\partial x}-\frac{1}{\gamma\beta_{1}f_{2}(x,h)}\frac{\partial f_{2}(x,h)}{\partial x}.\label{pf2}
\end{eqnarray}
Using (\ref{fx}) and (\ref{pf2}), we have, for $W_{\rm low}(h)\le x\le W_{\rm ref}(h)$,
\begin{eqnarray*}
\!\!\!\!\!\!\frac{\partial\pi^{*}_{\rm primal}}{\partial x}&=&\frac{\mu-r}{\sigma^{2}}\left\{(1-q_{2})+\frac{r}{k}(q_{1}-q_{2})C_{3}(h)\big[f_{2}(x,h)\big]^{q_{1}-2}
\frac{\partial f_{2}(x,h)}{\partial x}\right.\\
&&\left.+(1-q_{2})\frac{1}{\gamma\beta_{1}f_{2}(x,h)}\frac{\partial f_{2}(x,h)}{\partial x}\right\}\\
		   &=&\frac{\mu-r}{\sigma^{2}}\left[(1-q_{2})+\frac{A_{2}(x,h)}{B_{2}(x,h)}\right],
\end{eqnarray*}
		where
\begin{eqnarray*}
A_{2}(x,h)&=&\frac{r}{k}(q_{1}-q_{2})C_{3}(h)\left[f_{2}(x,h)\right]^{q_{1}-1}+(1-q_{2})\frac{1}{\gamma\beta_{1}},\\
B_{2}(x,h)&=&-\frac{r}{k}C_{3}(h)\big[f_{2}(x,h)\big]^{q_{1}-1}-\frac{r}{k}C_{4}(h)\big[f_{2}(x,h)\big]^{q_{2}-1}
-\frac{1}{\gamma\beta_{1}}.
\end{eqnarray*}
		As $e^{(\alpha-\lambda)\beta_{1}h}>f_{2}(x,h)\ge 1$ for $W_{\rm low}(h)\le x\le W_{\rm ref}(h)$ and $C_{3}(h)=\mathcal{O}(1)$,  there exists a constant $\bar{A}_{2}$ such that $|A_{2}(x,h)|\le \bar{A}_{2}$ for $W_{\rm low}(h)\le x\le W_{\rm ref}(h)$.
 As $\tilde{V}_{yy}(y,h)>0$, we know that $B_{2}(x,h)<0$ for any $h\ge\bar{h}$. Using $C_{4}<0$ and $q_{1}<0<1<q_{2}$ implies that $\frac{r}{k}C_{3}(h)y^{q_{1}-1}+\frac{r}{k}C_{4}(h)y^{q_{2}-1}+\frac{1}{\gamma\beta_{1}}$ as a function of $y$ is either decreasing, or increasing, or first increasing then decreasing. Hence for $W_{\rm low}(h)\le x\le W_{\rm ref}(h)$,
\begin{eqnarray}
-B_{2}(x,h)&\ge & \min\left\{\frac{r}{k}C_{3}(h)+\frac{r}{k}C_{4}(h)+\frac{1}{\gamma\beta_{1}}, \frac{r}{k}C_{3}(h)e^{-(1-q_{1})(\alpha-\lambda)\beta_{1}h}\right.\notag\\
&&\left.+\frac{r}{k}C_{4}(h)
e^{(q_{2}-1)(\alpha-\lambda)\beta_{1}h}+\frac{1}{\gamma\beta_{1}}\right\}\label{b2ineq}\\
		    &>& 0.\notag
\end{eqnarray}
Plugging the expressions of  $C_{3}(h)$ and $C_{4}(h)$ in (\ref{C1}) and (\ref{C6}) into (\ref{b2ineq}), we obtain that the expression in (\ref{b2ineq}) is continuous in $h$ with a limit $\frac{1}{\gamma\beta_{1}}\frac{q_{2}-1}{q_{2}-q_{1}}>0$ as $h\rightarrow\infty$. Hence there exists a constant $\bar{B}_{2}>0$ such that $-B_{2}(x,h)>\bar{B}_{2}>0$, i.e., $\frac{\partial\pi^{*}_{\rm primal}}{\partial x}$ is bounded for $W_{\rm low}(h)\le x\le W_{\rm ref}(h)$. Similarly, $\frac{\partial\pi^{*}_{\rm primal}}{\partial x}$ is also bounded for the rest two regions $W_{\rm ref}(h)\le x\le W_{\rm peak}(h)$ and $W_{\rm peak}(h)\le x\le W_{\rm updt}(h)$.
		
Step 2: $\frac{\partial\pi^{*}_{\rm primal}}{\partial h}$ is bounded:\\
By definition of $\pi^{*}_{\rm primal}$, denote $ c_{11}\triangleq \frac{\mu-r}{\sigma^{2}}\frac{r}{k} $, we obtain

\begin{equation*}
\!\!\!\!\!\!\!\!\!		   \frac{\partial\pi^{*}_{\rm primal}}{\partial h}=c_{11}\left\{
		   \begin{array}{l}
		   C_{1}(h)(q_{1}-1)\big[f_{1}(x,h)\big]^{q_{1}-2}\frac{\partial f_{1}(x,h)}{\partial h}+C_{1}'(h)\big[f_{1}(x,h)\big]^{q_{1}-1}\\+C_{2}(h)(q_{2}-1)\big[f_{1}(x,h)\big]^{q_{2}-2}\frac{\partial f_{1}(x,h)}{\partial h}+C_{2}'(h)\big[f_{1}(x,h)\big]^{q_{2}-1},\ \frac{\lambda h}{\gamma}\le x\le W_{\rm low}(h),\\
		   C_{3}(h)(q_{1}-1)\big[f_{2}(x,h)\big]^{q_{1}-2}\frac{\partial f_{2}(x,h)}{\partial h}+C_{3}'(h)\big[f_{2}(x,h)\big]^{q_{1}-1}\\+C_{4}(h)(q_{2}-1)\big[f_{2}(x,h)\big]^{q_{2}-2}\frac{\partial f_{2}(x,h)}{\partial h}+C_{4}'(h)\big[f_{2}(x,h)\big]^{q_{2}-1},\ W_{\rm low}(h)\le x\le W_{\rm ref}(h),\\
		   C_{5}(h)(q_{1}-1)\big[f_{3}(x,h)\big]^{q_{1}-2}\frac{\partial f_{3}(x,h)}{\partial h}+C_{5}'(h)\big[f_{3}(x,h)\big]^{q_{1}-1}\\+C_{6}(h)(q_{2}-1)\big[f_{3}(x,h)\big]^{q_{2}-2}\frac{\partial f_{3}(x,h)}{\partial h}+C_{6}'(h)\big[f_{3}(x,h)\big]^{q_{2}-1},\ W_{\rm ref}(h)\le x\le W_{\rm peak}(h),\\
		   C_{7}(h)(q_{1}-1)\big[f_{4}(x,h)\big]^{q_{1}-2}\frac{\partial f_{4}(x,h)}{\partial h}+C_{7}'(h)\big[f_{4}(x,h)\big]^{q_{1}-1}\\+C_{8}(h)(q_{2}-1)\big[f_{4}(x,h)\big]^{q_{2}-2}\frac{\partial f_{4}(x,h)}{\partial h}+C_{8}'(h)\big[f_{4}(x,h)\big]^{q_{2}-1},\ W_{\rm peak}(h)\le x\le W_{\rm updt}(h).
		   \end{array}
		   \right.
		\end{equation*}
Differentiating (\ref{f1}),

		\begin{eqnarray}
		   0&=&-C_{1}(h)q_{1}(q_{1}-1)\big[f_{1}(x,h)\big]^{q_{1}-2}\frac{\partial f_{1}(x,h)}{\partial h}-C_{1}'(h)q_{1}\big[f_{1}(x,h)\big]^{q_{1}-1}\notag\\
&&-C_{2}(h)q_{2}(q_{2}-1)\big[f_{1}(x,h)\big]^{q_{2}-2}\frac{\partial f_{1}(x,h)}{\partial h}-C_{2}'(h)q_{2}\big[f_{1}(x,h)\big]^{q_{2}-1}+\frac{\lambda}{\gamma}.\label{phf1}
		\end{eqnarray}
Then, using  $C_{2}(h)=0$, we have for $\frac{\lambda h}{\gamma}\le x\le W_{\rm low}(h)$,

		\begin{eqnarray*}
		   \frac{\partial\pi^{*}_{\rm primal}}{\partial h}&=&c_{11}\frac{\lambda}{\gamma q_{1}}.
		\end{eqnarray*}
		Hence $\frac{\partial\pi^{*}_{\rm primal}}{\partial h}$ is constant for $\frac{\lambda h}{\gamma}\le x\le W_{\rm low}(h)$.\\
Differentiating (\ref{f2}),

		\begin{eqnarray}
		   0&=&-C_{3}(h)q_{1}(q_{1}-1)\big[f_{2}(x,h)\big]^{q_{1}-2}\frac{\partial f_{2}(x,h)}{\partial h}-C_{3}'(h)q_{1}\big[f_{2}(x,h)\big]^{q_{1}-1}\notag\\
&&-C_{4}(h)q_{2}(q_{2}-1)\big[f_{2}(x,h)\big]^{q_{2}-2}\frac{\partial f_{2}(x,h)}{\partial h}-C_{4}'(h)q_{2}\big[f_{2}(x,h)\big]^{q_{2}-1}\notag\\&&-\frac{1}{\gamma\beta_{1}f_{2}(x,h)}\frac{\partial f_{2}(x,h)}{\partial h}+\frac{\alpha}{\gamma}.\label{phf2}
		\end{eqnarray}
Using (\ref{phf2}), we have for $W_{\rm low}(h)\le x\le W_{\rm ref}(h)$,

		\begin{eqnarray*}
		   \frac{\partial\pi^{*}_{\rm primal}}{\partial h}&=&c_{11}\left\{C_{3}'(h)\big[f_{2}(x,h)\big]^{q_{1}-1}(1-\frac{q_{1}}{q_{2}})+\frac{\alpha}{\gamma q_{2}}\right.\\&&\left.+\Big[C_{3}(h)\big[f_{2}(x,h)\big]^{q_{1}-1}(q_{1}-q_{2})-\frac{1}{\gamma \beta_{1}q_{2}}\Big]\frac{1}{f_{2}}\frac{\partial f_{2}(x,h)}{\partial h}\right\}\\
		   &=&\frac{\mu-r}{\sigma^{2}}\frac{r}{k}\left\{C_{3}'(h)\big[f_{2}(x,h)
\big]^{q_{1}-1}(1-\frac{q_{1}}{q_{2}})+\frac{\alpha}{\gamma q_{2}}\right.\\
&&\left.+\frac{k}{r}A_{2}(x,h)
\frac{C_{3}'(h)q_{1}[f_{2}(x,h)]^{q_{1}-1}+
C_{4}'(h)q_{2}[f_{2}(x,h)]^{q_{2}-1}-\frac{\alpha}{\gamma}}{B_{2}(x,h)}\right\}.
		\end{eqnarray*}
Thus, using the expression of $C_{3}(h)$ and $C_{4}(h)$ in (\ref{C1}) and (\ref{C6}), we obtain  that $C_{3}'(h)$ and $C_{4}'(h)$ are both bounded. Moreover, we have $1\le f_{2}(x,h)<e^{(\alpha-\lambda)\beta_{1}h}$ for $W_{\rm low}(h)<x\le W_{\rm ref}(h)$. Applying the estimates of $A_{2}(x,h)$ and $B_{2}(x,h)$ again, we know that $\frac{\partial\pi^{*}_{\rm primal}}{\partial h}$ is bounded for $W_{\rm low}(h)<x\le W_{\rm ref}(h)$. 	For the rest two regions $W_{\rm ref}(h)\le x\le W_{\rm peak}(h)$ and $W_{\rm peak}(h)\le x\le W_{\rm updt}(h)$, the proof is similar and omitted.
\end{proof}

\vskip 15pt		
		\section{Results with general reference}
		\label{genpre}
		We consider the alternative endogenous reference point $\alpha[\varphi(h)c+(1-\varphi(h))h]$, which is a fraction $\alpha$ of the convex combination of the current consumption and consumption peak. $\varphi(h)$ is the proportion assigned to current consumption (it is a function of $h$) and we assume that the proportion function $\varphi$ is non-decreasing and smooth with values in $[0,1]$. The non-decreasing property suggests that once the maximum is updated, its weight decreases. This assumption aims to capture the insight that, upon updating the consumption peak, the agent tends to put more emphasis on the current consumption $c$ instead of the past peak. We further assume that $\varphi'(h)h+\varphi(h)\le 1$ for all $h\ge h_{0}$, which implies the non-increasing property of the utility $U(c,h)$ w.r.t $h$. 
		\begin{remark}
		$\varphi(h)=0$ reduces to the case in the main part of the paper. $\varphi(h)=1$ reduces to a non-habit model. A non-trivial choice of $\varphi(h)$ satisfying the mentioned assumptions might be the fractional function $\varphi(h)=\bar{\varphi}\frac{h}{h+\hat{h}}$ where constant $\hat{h}>0$ is a benchmark level and $\bar{\varphi}\in[0,1]$ is a scaling constant.
		\end{remark}
		The optimal dual feedback form $c^{*}(y,h)$ is replaced by
		
		\begin{equation*}		
	c^{*}(y,h)=\left\{
	\begin{array}{l}
	\lambda h,\hskip 5.8cm 1\vee(1-\alpha\varphi(h))e^{[(\alpha-\lambda)-(1-\lambda)\alpha \varphi(h)]\beta_{1}h}\le y,\\
	-\frac{1}{\beta_{1}}\frac{1}{1-\alpha\varphi(h)}\ln(\frac{y}{1-\alpha\varphi(h)})+\alpha h\frac{1-\varphi(h)}{1-\alpha\varphi(h)} ,\ \ 1\le y<1\vee(1-\alpha\varphi(h))e^{[(\alpha-\lambda)-(1-\lambda)\alpha \varphi(h)]\beta_{1}h},\\
	-\frac{1}{\beta_{2}}\frac{1}{1-\alpha\varphi(h)}\ln(\frac{y}{1-\alpha\varphi(h)})+\alpha h\frac{1-\varphi(h)}{1-\alpha\varphi(h)} ,\ \  (1-\alpha\varphi(h))e^{-(1-\alpha)\beta_{2}h}\le y<1,\\
	h,\ \ \hskip 5.8cm (1-\alpha)e^{-(1-\alpha)\beta_{2}h}\le y<(1-\alpha\varphi(h))e^{-(1-\alpha)\beta_{2}h}.
	\end{array}
	\right.
	\end{equation*}
	We need the following assumption on the upper bound of $\varphi(h)$.
	\begin{assumption}
	\label{asvar}
	$\varphi(\infty)<\frac{\alpha-\lambda}{\alpha(1-\lambda)}$.
	\end{assumption}
	This assumption is reasonable. On the one hand, we give a lower bound of the weight given to $h$ in the reference point, so that consumption peak is always taken into consideration by the agent. On the other hand, under Assumption \ref{asvar}, for large $h$ we always have $(1-\alpha\varphi(h))e^{[(\alpha-\lambda)-(1-\lambda)\alpha \varphi(h)]\beta_{1}h}>1$ so that every region in the expression of optimal consumption is not null, indicating that people with higher standard of living have more complicated behavior. Finally, it is interesting that Assumption \ref{asvar} also serves as a convenient sufficient condition for verification theorem. See Remark \ref{c8po}.
	 
   The general solution to the dual HJB equation becomes
   
   \begin{equation*}
		\!\!\!  \!\!\! \!\!\!  \tilde{V}(y,h)\!\!=\!\!\left\{
		    \begin{array}{l}
		   C_{1}(h)y^{q_{1}}\!\!+\!\!C_{2}(h)y^{q_{2}}\!\!-\!\!\frac{1}{\gamma}\lambda hy\!\!+\!\!\frac{1}{\gamma\beta_{1}}\big[1-e^{[(\alpha-\lambda)-(1-\lambda)\alpha \varphi(h)]\beta_{1}h}\big]\\,\  1\vee(1-\alpha\varphi(h))e^{[(\alpha-\lambda)-(1-\lambda)\alpha \varphi(h)]\beta_{1}h}\le y,\\
		   C_{3}(h)y^{q_{1}}\!\!+\!\!C_{4}(h)y^{q_{2}}\!\!+\!\!\frac{1}{\gamma\beta_{1}}(1\!\!-\!\!\frac{y}{1-\alpha\varphi(h)})\!\!+\!\!\frac{y}{\gamma\beta_{1}}\frac{1}{1-\alpha\varphi(h)}\ln(\frac{y}{1-\alpha\varphi(h)})\!\!-\!\!\frac{\alpha h y}{\gamma}\frac{1-\varphi(h)}{1-\alpha\varphi(h)} \\ \!\!+\!\!
\frac{k}{\gamma^{2}\beta_{1}}\frac{1}{1-\alpha\varphi(h)}y,\ 1\le y<1\vee(1-\alpha\varphi(h))e^{[(\alpha-\lambda)-(1-\lambda)\alpha \varphi(h)]\beta_{1}h},\\
		   C_{5}(h)y^{q_{1}}\!\!+\!\!C_{6}(h)y^{q_{2}}\!\!+\!\!\frac{1}{\gamma\beta_{2}}(1\!\!-\!\!\frac{y}{1-\alpha\varphi(h)})\!\!+\!\!\frac{y}{\gamma\beta_{2}}\frac{1}{1-\alpha\varphi(h)}\ln(\frac{y}{1-\alpha\varphi(h)})\!\!-\!\!\frac{\alpha h y}{\gamma}\frac{1-\varphi(h)}{1-\alpha\varphi(h)} \\ \!\!+\!\!
\frac{k}{\gamma^{2}\beta_{2}}\frac{1}{1-\alpha\varphi(h)}y,\ (1-\alpha\varphi(h))e^{-(1-\alpha)\beta_{2}h}\le y<1,\\
		   C_{7}(h)y^{q_{1}}\!\!+\!\!C_{8}(h)y^{q_{2}}\!\!-\!\!\frac{1}{\gamma}hy\!\!+
\!\!\frac{1}{\gamma\beta_{2}}\big[1\!-\!e^{-(1\!-\!\alpha)\beta_{2}h}\big],\ (1-\alpha)e^{-(1-\alpha)\beta_{2}h}\le y<(1-\alpha\varphi(h))e^{-(1-\alpha)\beta_{2}h}.
		    \end{array}
		    \right.
\end{equation*}
   By the same approach, we obtain the expressions of $C_{i}(h),1\le i\le 8$ in the following two cases:
  
   {\bf Case 1.} For $h$ such that $1<(1-\alpha\varphi(h))e^{[(\alpha-\lambda)-(1-\lambda)\alpha \varphi(h)]\beta_{1}h}$,
   
   \begin{align*}
	C_{2}(h)=&0, \ \ \   C_{4}(h)=-\frac{k}{\gamma^{2}\beta_{1}}\frac{1-q_{1}}{q_{2}-q_{1}}\frac{1}{(1-\alpha\varphi(h))^{q_{2}}}e^{-[(\alpha-\lambda)-(1-\lambda)\alpha \varphi(h)](q_{2}-1)\beta_{1}h},\\
	    C_{6}(h)=&C_{4}(h)+\frac{k}{\gamma^{2}}\frac{\beta_{2}-\beta_{1}}{\beta_{1}\beta_{2}}\frac{1-q_{1}}{q_{2}-q_{1}}\frac{1}{1-\alpha\varphi(h)}+\frac{1}{\gamma}\frac{\beta_{2}-\beta_{1}}{\beta_{1}\beta_{2}}\frac{q_{1}}{q_{2}-q_{1}}\frac{\alpha\varphi(h)}{1-\alpha\varphi(h)}\\&+\frac{1}{\gamma}\frac{\beta_{2}-\beta_{1}}{\beta_{1}\beta_{2}}\frac{1-q_{1}}{q_{2}-q_{1}}\frac{1}{1-\alpha\varphi(h)}\ln(\frac{1}{1-\alpha\varphi(h)}),
	    \end{align*}
	    \begin{align*}
	    C_{8}(h)=&C_{6}(h)+\frac{k}{\gamma^{2}\beta_{2}}\frac{1-q_{1}}{q_{2}-q_{1}}\frac{1}{(1-\alpha\varphi(h))^{q_{2}}}e^{(1-\alpha)(q_{2}-1)\beta_{2}h},\\
	    C_{7}(h)=&\frac{(1-\alpha)^{q_{2}-q_{1}}\frac{k}{\gamma^{2}}\frac{1-q_{1}}{q_{2}-q_{1}}
(\alpha-\lambda)(q_{2}-1)}{(1-\alpha)(q_{2}-q_{1})\beta_{2}
+(\alpha-\lambda)(q_{2}-1)\beta_{1}}e^{-\big[(1-\alpha)(q_{2}-q_{1})\beta_{2}
+(\alpha-\lambda)(q_{2}-1)\beta_{1}\big]h}\notag\\
	    &+(1-\alpha)^{q_{2}-q_{1}}\frac{k}{\gamma^{2}\beta_{2}}\frac{q_{2}-1}{q_{2}-q_{1}}
e^{-(1-\alpha)(1-q_{1})\beta_{2}h},\\
	    C_{5}(h)=&C_{7}(h)-\frac{k}{\gamma^{2}\beta_{2}}\frac{q_{2}-1}{q_{2}-q_{1}}\frac{1}{(1-\alpha\varphi(h))^{q_{1}}}e^{-(1-\alpha)(1-q_{1})\beta_{2}h},\\
	    C_{3}(h)=&C_{5}(h)-\frac{k}{\gamma^{2}}
\frac{\beta_{2}-\beta_{1}}{\beta_{1}\beta_{2}}\frac{q_{2}-1}{q_{2}-q_{1}}\frac{1}{1-\alpha\varphi(h)}+\frac{1}{\gamma}\frac{\beta_{2}-\beta_{1}}{\beta_{1}\beta_{2}}\frac{q_{2}}{q_{2}-q_{1}}\frac{\alpha\varphi(h)}{1-\alpha\varphi(h)}\\&-\frac{1}{\gamma}\frac{\beta_{2}-\beta_{1}}{\beta_{1}\beta_{2}}\frac{q_{2}-1}{q_{2}-q_{1}}\frac{1}{1-\alpha\varphi(h)}\ln(\frac{1}{1-\alpha\varphi(h)}),\\
	    C_{1}(h)=&C_{3}(h)+\frac{k}{\gamma^{2}\beta_{1}}
\frac{q_{2}-1}{q_{2}-q_{1}}\frac{1}{(1-\alpha\varphi(h))^{q_{1}}}e^{[(\alpha-\lambda)-(1-\lambda)\alpha \varphi(h)](1-q_{1})\beta_{1}h}.
	\end{align*}
	In this case, as $\varphi(h)\in[0,1]$ and $e^{-[(\alpha-\lambda)-(1-\lambda)\alpha \varphi(h)]\beta_{1}h}<1-\alpha\varphi(h)\le 1$, the estimates of $C_{i}(h),1\le i\le 8$ is the same as Section \ref{dersol} except that the order estimate of $C_{4}(h)$ is replaced by $C_{4}(h)=\mathcal{O}(1)$.
	
	{\bf Case 2.} For $h$ such that $1\ge (1-\alpha\varphi(h))e^{[(\alpha-\lambda)-(1-\lambda)\alpha \varphi(h)]\beta_{1}h}$ (the second region $1\le y<1\vee (1-\alpha\varphi(h))e^{[(\alpha-\lambda)-(1-\lambda)\alpha \varphi(h)]\beta_{1}h}$ is null), we have:
	
		\begin{align*}
	C_{2}(h)=&0, \\   
	    C_{6}(h)=&-\frac{1-q_{1}}{q_{2}-q_{1}}\frac{\lambda h}{\gamma}-\frac{q_{1}}{q_{2}-q_{1}}\frac{1}{\gamma\beta_{1}}(1-e^{[(\alpha-\lambda)-(1-\lambda)\alpha \varphi(h)]\beta_{1}h})-\frac{1-q_{1}}{q_{2}-q_{1}}\frac{1}{\gamma\beta_{2}}\frac{1}{1-\alpha \varphi(h)}\ln(\frac{1}{1-\alpha\varphi(h)})\\&+\frac{q_{1}}{q_{2}-q_{1}}\frac{1}{\gamma\beta_{2}}(1-\frac{1}{1-\alpha\varphi(h)})+\frac{1-q_{1}}{q_{2}-q_{1}}\frac{\alpha h}{\gamma}\frac{1-\varphi(h)}{1-\alpha\varphi(h)}-\frac{1-q_{1}}{q_{2}-q_{1}}\frac{k}{\gamma^{2}\beta_{2}}\frac{1}{1-\alpha\varphi(h)},\\
	    C_{8}(h)=&C_{6}(h)+\frac{k}{\gamma^{2}\beta_{2}}\frac{1-q_{1}}{q_{2}-q_{1}}\frac{1}{(1-\alpha\varphi(h))^{q_{2}}}e^{(1-\alpha)(q_{2}-1)\beta_{2}h},\\
	    C_{7}(h)=&\frac{(1-\alpha)^{q_{2}-q_{1}}\frac{k}{\gamma^{2}}\frac{1-q_{1}}{q_{2}-q_{1}}
(\alpha-\lambda)(q_{2}-1)}{(1-\alpha)(q_{2}-q_{1})\beta_{2}
+(\alpha-\lambda)(q_{2}-1)\beta_{1}}e^{-\big[(1-\alpha)(q_{2}-q_{1})\beta_{2}
+(\alpha-\lambda)(q_{2}-1)\beta_{1}\big]h}\notag\\
	    &+(1-\alpha)^{q_{2}-q_{1}}\frac{k}{\gamma^{2}\beta_{2}}\frac{q_{2}-1}{q_{2}-q_{1}}e^{-(1-\alpha)(1-q_{1})\beta_{2}h},\\
	    C_{5}(h)=&C_{7}(h)-\frac{k}{\gamma^{2}\beta_{2}}\frac{q_{2}-1}{q_{2}-q_{1}}\frac{1}{(1-\alpha\varphi(h))^{q_{1}}}e^{-(1-\alpha)(1-q_{1})\beta_{2}h},
	    \end{align*}
        \begin{align*}
	    C_{1}(h)=&C_{5}(h)+\frac{q_{2}-1}{q_{2}-q_{1}}\frac{\lambda h}{\gamma}-\frac{q_{2}}{q_{2}-q_{1}}\frac{1}{\gamma\beta_{1}}(1-e^{[(\alpha-\lambda)-(1-\lambda)\alpha \varphi(h)]\beta_{1}h})+\frac{q_{2}-1}{q_{2}-q_{1}}\frac{1}{\gamma\beta_{2}}\frac{1}{1-\alpha\varphi(h)}\ln(\frac{1}{1-\alpha\varphi(h)})\\&+\frac{q_{2}}{q_{2}-q_{1}}\frac{1}{\gamma\beta_{2}}(1-\frac{1}{1-\alpha\varphi(h)})-\frac{q_{2}-1}{q_{2}-q_{1}}\frac{\alpha h}{\gamma}\frac{1-\varphi(h)}{1-\alpha\varphi(h)}+\frac{q_{2}-1}{q_{2}-q_{1}}\frac{k}{\gamma^{2}\beta_{2}}\frac{1}{1-\alpha\varphi(h)}.
	\end{align*}
	Under Assumption \ref{asvar}, Case 2 does not happen for sufficiently large $h$ so that the asymptotic estimates are not necessary. In Case 2, we still have 
	
\begin{align*}
	        &C_{1}(h)>0,\ \ C_{7}(h)>0,\\
	        &C_{8}(h)>0.
	        \end{align*}
	   Assumption \ref{asvar} is actually a sufficient condition for $C_{8}(h)>0$.     
		\begin{remark}
		\label{c8po}
		 Under Assumption \ref{asvar}, the proof of $C_{8}(h)>0$ in Case 2 is as follows:\\
		We write $C_{8}(h)$ as the sum of three terms given by
		
		\begin{align*}
		    C_{8}(h)=&\frac{k}{\gamma^{2}\beta_{2}}\frac{1-q_{1}}{q_{2}-q_{1}}\frac{1}{(1-\alpha\varphi(h))^{q_{2}}}[e^{(1-\alpha)(q_{2}-1)\beta_{2}h}-(1-\alpha\varphi(h))^{q_{2}-1}]\\
		    &+[\frac{1-q_{1}}{q_{2}-q_{1}}\frac{h}{\gamma}\frac{(\alpha-\lambda)-(1-\lambda)\alpha\varphi(h)}{1-\alpha\varphi(h)}-\frac{q_{1}}{q_{2}-q_{1}}\frac{1}{\gamma\beta_{1}}(1-e^{[(\alpha-\lambda)-(1-\lambda)\alpha\varphi(h)]\beta_{1}h})]\\
		    &+[-\frac{1-q_{1}}{q_{2}-q_{1}}\frac{1}{\gamma\beta_{2}}\frac{1}{1-\alpha\varphi(h)}\ln(\frac{1}{1-\alpha\varphi(h)})+\frac{q_{1}}{q_{2}-q_{1}}\frac{1}{\gamma\beta_{2}}(1-\frac{1}{1-\alpha\varphi(h)})].
		\end{align*}
		The first term is positive due to $e^{(1-\alpha)(q_{2}-1)\beta_{2}h}>1\ge (1-\alpha\varphi(h))^{q_{2}-1}$.
		
		For the second term, it can be directly verified that
		
		\begin{equation*}
		    f(x):=\frac{1-q_{1}}{q_{2}-q_{1}}\frac{h}{\gamma}\frac{x}{1-\alpha\varphi(h)}-\frac{q_{1}}{q_{2}-q_{1}}\frac{1}{\gamma\beta_{1}}(1-e^{x\beta_{1}h})
		\end{equation*}
		is increasing for $x\ge 0$. Thus $f((\alpha-\lambda)-(1-\lambda)\alpha\varphi(h))>f(0)=0$, which implies that the second term is positive.
		
		For the last term, let
		
		\begin{equation*}
		    g(x):=-(1-q_{1})x\ln(x)+q_{1}(1-x),
		\end{equation*}
		then one can directly show that $g(x)$ is decreasing for $x\in[1,\frac{1}{1-\alpha}]$. Hence $g(\frac{1}{1-\alpha\varphi(h)})\ge g(1)=0$ and the last term is non-negative.
		\end{remark}
	
		 Under all the aforementioned assumptions on $\varphi$, we can establish the verification theorem and apply duality to obtain the optimal strategy given in the following theorem.

		\begin{theorem}	
		For $(x_{0},h_{0})\in\mathcal{C}$, where $\mathcal{C}$ is the effective region given by (\ref{effective}), let $c^{*}_{\rm primal}(\cdot,\cdot)$ and $ \pi^{*}_{\rm primal}(\cdot,\cdot)$ be the feedback functions in terms of primal variable given respectively by
		\[
		    c_{\rm primal}^{*}(x,h)=\left\{
	\begin{array}{l}
	\lambda h,\hskip 5.8cm \frac{\lambda h}{\gamma}\le x\le W_{\rm low}(h)\wedge W_{\rm ref}(h),\\
	-\frac{1}{\beta_{1}}\frac{1}{1-\alpha\varphi(h)}\ln(\frac{f_{2}(x,h)}{1-\alpha\varphi(h)})+\alpha h\frac{1-\varphi(h)}{1-\alpha\varphi(h)} ,\ \ W_{\rm low}(h)\wedge W_{\rm ref}(h)<x\le W_{\rm ref}(h),\\
	-\frac{1}{\beta_{2}}\frac{1}{1-\alpha\varphi(h)}\ln(\frac{f_{3}(x,h)}{1-\alpha\varphi(h)})+\alpha h\frac{1-\varphi(h)}{1-\alpha\varphi(h)} ,\ \  W_{\rm ref}(h)<x\le W_{\rm peak}(h),\\
	h,\ \ \hskip 5.8cm W_{\rm peak}(h)<x\le W_{\rm updt}(h).
	\end{array}
	\right.\\
	\]
	\[
	\pi_{\rm primal}^{*}(x,h)\!\!=\!\!\frac{\mu-r}{\sigma^{2}}\left\{
	\begin{array}{l}
	\frac{r}{k}\left\{C_{1}(h)\big[f_{1}(x,h)\big]^{q_{1}-1}\!\!+\!\!C_{2}(h)\big[f_{1}(x,h)\big]^{q_{2}-1}\right\},\hskip 0.4cm \frac{\lambda h}{\gamma}\!\!\le\!\! x\!\!\le\!\! W_{\rm low}(h)\!\!\wedge\!\! W_{\rm ref}(h),\\
	\frac{r}{k}\left\{C_{3}(h)\big[f_{2}(x,h)\big]^{q_{1}-1}\!\!+\!\!C_{4}(h)\big[f_{2}(x,h)\big]^{q_{2}-1}\right\}
\!\!+\!\!\frac{1}{\gamma\beta_{1}}\frac{1}{1-\alpha\varphi(h)},\ W_{\rm low}(h)\!\!\wedge\!\! W_{\rm ref}(h)\!\!<\!\!x\!\!\le\!\! W_{\rm ref}(h),\\
	\frac{r}{k}\left\{C_{5}(h)\big[f_{3}(x,h)\big]^{q_{1}-1}\!\!+\!\!C_{6}(h)\big[f_{3}(x,h)\big]^{q_{2}-1}\right\}\!\!+\!\!\frac{1}{\gamma\beta_{2}}\frac{1}{1-\alpha\varphi(h)},\ W_{\rm ref}(h)\!\!<\!\!x\!\!\le\!\! W_{\rm peak}(h),\\
	\frac{r}{k}\left\{C_{7}(h)\big[f_{4}(x,h)\big]^{q_{1}-1}\!\!+\!\!C_{8}(h)\big[f_{4}(x,h)\big]^{q_{2}-1}\right\},\hskip 0.3cm W_{\rm peak}(h)\!\!<\!\!x\!\!\le\!\! W_{\rm updt}(h).
	\end{array}
	\right.
	\]
	where $f_{i}(x,h),\ 1\le i\le 4$ are uniquely determined by 
	
	\begin{align*}
	    x=&-C_{1}(h)q_{1}\big[f_{1}(x,h)\big]^{q_{1}-1}-C_{2}(h)q_{2}\big[f_{1}(x,h)\big]^{q_{2}-1}+\frac{\lambda h}{\gamma},\\
	    x=&-C_{3}(h)q_{1}\big[f_{2}(x,h)\big]^{q_{1}-1}-C_{4}(h)q_{2}\big[f_{2}(x,h)\big]^{q_{2}-1}-\frac{1}{\gamma\beta_{1}}\frac{1}{1-\alpha\varphi(h)}\ln\big[\frac{f_{2}(x,h)}{1-\alpha\varphi(h)}\big]\\&-\frac{k}{\gamma^{2}\beta_{1}}\frac{1}{1-\alpha\varphi(h)}+\frac{\alpha h}{\gamma}\frac{1-\varphi(h)}{1-\alpha\varphi(h)},\\
	    x=&-C_{5}(h)q_{1}\big[f_{3}(x,h)\big]^{q_{1}-1}-C_{6}(h)q_{2}\big[f_{3}(x,h)\big]^{q_{2}-1}-\frac{1}{\gamma\beta_{2}}\frac{1}{1-\alpha\varphi(h)}\ln\big[\frac{f_{3}(x,h)}{1-\alpha\varphi(h)}\big]\\&-\frac{k}{\gamma^{2}\beta_{2}}\frac{1}{1-\alpha\varphi(h)}+\frac{\alpha h}{\gamma}\frac{1-\varphi(h)}{1-\alpha\varphi(h)},\\
	    x=&-C_{7}(h)q_{1}\big[f_{4}(x,h)\big]^{q_{1}-1}-C_{8}(h)q_{2}\big[f_{4}(x,h)\big]^{q_{2}-1}+\frac{h}{\gamma},
	    \end{align*}
	    and $W_{\rm low}(h),W_{\rm ref}(h),W_{\rm peak}(h)$ and $W_{\rm updt}(h)$ are given by
	    
	    \begin{align*}
	    W_{\rm low}(h)=&-C_{1}(h)q_{1}\big[(1-\alpha\varphi(h))e^{[(\alpha-\lambda)-(1-\lambda)\alpha \varphi(h)]\beta_{1}h}\big]^{q_{1}-1}\\&-C_{2}(h)q_{2}\big[(1-\alpha\varphi(h))e^{[(\alpha-\lambda)-(1-\lambda)\alpha \varphi(h)]\beta_{1}h}\big]^{q_{2}-1}+\frac{\lambda h}{\gamma},\\
	    W_{\rm ref}(h)=&-C_{3}(h)q_{1}-C_{4}(h)q_{2}-\frac{1}{\gamma\beta_{1}}\frac{1}{1-\alpha\varphi(h)}\ln\big[\frac{1}{1-\alpha\varphi(h)}\big]\\&-\frac{k}{\gamma^{2}\beta_{1}}\frac{1}{1-\alpha\varphi(h)}+\frac{\alpha h}{\gamma}\frac{1-\varphi(h)}{1-\alpha\varphi(h)},\\
	    W_{\rm peak}(h)=&-C_{5}(h)q_{1}\big[(1-\alpha\varphi(h))e^{-(1-\alpha)\beta_{2}h}\big]^{q_{1}-1}-C_{6}(h)q_{2}\big[(1-\alpha\varphi(h))e^{-(1-\alpha)\beta_{2}h}\big]^{q_{2}-1}\\&+\frac{h}{\gamma}\frac{1-\alpha}{1-\alpha\varphi(h)}-\frac{k}{\gamma^{2}\beta_{2}}\frac{1}{1-\alpha\varphi(h)}+\frac{\alpha h}{\gamma}\frac{1-\varphi(h)}{1-\alpha\varphi(h)},\\
	    W_{\rm updt}(h)=&-C_{7}(h)q_{1}\big[(1-\alpha)e^{-(1-\alpha)\beta_{2}h}\big]^{q_{1}-1}-C_{8}(h)q_{2}\big[(1-\alpha)e^{-(1-\alpha)\beta_{2}h}\big]^{q_{2}-1}+\frac{h}{\gamma}.
	\end{align*}
		Then SDE
		\begin{equation}
		\label{sde}
		\left\{
		\begin{array}{l}
		dX_{t}=rX_{t}dt+\pi^{*}_{\rm primal}(X_{t},H^{*}_{t})(\mu-r)dt+\pi^{*}_{\rm primal}(X_{t},H^{*}_{t})\sigma dW_{t}-c^{*}_{\rm primal}(X_{t},H^{*}_{t})dt,\\
		X_{0}=x_{0}
		\end{array}
		\right.
		\end{equation}
		with $H^{*}_{t}\triangleq h_{0}\vee \sup\limits_{s\le t}c^{*}_{\rm primal}(X_{s},H^{*}_{s})$ and $H^{*}_{0}=h_{0}$,  has a unique strong solution $\{X^{*}_{t},\ t\ge 0\}$. The optimal consumption and investment policy is
		
		\begin{equation*}
		\Big\{\big(c^{*}_{\rm primal}(X^{*}_{t},H^{*}_{t}),\pi^{*}_{\rm primal}(X^{*}_{t},H^{*}_{t})\big),\ t\ge 0\Big\}.
		\end{equation*}
	\end{theorem}
	\begin{remark}
	     It is interesting to notice that the depression region can possibly vanish under the current model. As has been mentioned, Assumption \ref{asvar} ensures that this will not happen for large $h$. Specific characterizations of scenarios depends crucially on the form of $\varphi$. There are similar phenomena in \cite{Li2021}, though due to completely different reasons. 
	\end{remark}
	 \begin{remark}
		  The proofs of main results in this generalization are similar and thus omitted. The difference in proofs mainly lies in the proof of Lemma \ref{vyypo} and Lemma \ref{lipsc}. In the proof of Lemma \ref{vyypo} for this generalization, we need to apply a similar decomposition as in Remark \ref{c8po} to show that $\frac{r}{k}C_{6}(h)y^{q_{2}-1}+\frac{1}{\gamma\beta_{2}}\frac{1-q_{1}}{q_{2}-q_{1}}\frac{1}{1-\alpha\varphi(h)}>0$ for $(1-\alpha\varphi(h))e^{-(1-\alpha)\beta_{2}h}\le y<1$ in case $1\ge (1-\alpha\varphi(h))e^{[(\alpha-\lambda)-(1-\lambda)\alpha \varphi(h)]\beta_{1}h}$. While in the proof of Lemma \ref{lipsc}, we need to apply the boundedness of $\varphi(h)$ and $\varphi'(h)$ ( $0\le\varphi'(h)<\frac{1-\varphi(h)}{h}\le\frac{1-\frac{\alpha-\lambda}{\alpha(1-\lambda)}}{h_{0}}$) to show that $C_{i}'(h),1\le i\le 8$ are bounded.
		 \end{remark}
	
\vskip 15pt		
		\section{Results for $r\neq\gamma$}
		\label{genpara}
For $r\neq\gamma$, equation (\ref{ode}) should be replaced by
\begin{equation}
    -\gamma\tilde{V}(y,h)+(\gamma-r)y\tilde{V}_{y}(y,h)+\frac{(r-\mu)^{2}}{2\sigma^{2}}y^{2}\tilde{V}_{yy}(y,h)=-\tilde{U}(y,h),
    \label{odegen}
\end{equation}
and $q_{i},i=1,2$ should be defined instead by $q_{1}=\frac{k+r-\gamma-\sqrt{(k+r-\gamma)^{2}+4k\gamma}}{2k}$, $q_{2}=\frac{k+r-\gamma+\sqrt{(k+r-\gamma)^{2}+4k\gamma}}{2k}$. It still holds that $q_{1}<0<1<q_{2}$.

The general solution to (\ref{odegen}) is 

\begin{equation*}
		\!\!\!  \!\!\! \!\!\!  \tilde{V}(y,h)\!\!=\!\!\left\{
		    \begin{array}{l}
		   C_{1}(h)y^{q_{1}}+C_{2}(h)y^{q_{2}}-\frac{1}{r}\lambda hy+\frac{1}{\gamma\beta_{1}}\big[1-e^{(\alpha-\lambda)\beta_{1}h}\big],\ \ \ \ \  e^{(\alpha-\lambda)\beta_{1}h}\le y,\\
    	   C_{3}(h)y^{q_{1}}\!\!+\!\!C_{4}(h)y^{q_{2}}\!\!+\!\!\frac{1}{r\beta_{1}}y\ln(y)\!\!+\!\!\frac{\gamma-2r+k}{r^{2}\beta_{1}}y\!\!-\!\!\frac{\alpha h}{r}y\!\!+\!\!\frac{1}{\gamma\beta_{1}},\ 1\le y<e^{(\alpha-\lambda)\beta_{1}h},\\
		   C_{5}(h)y^{q_{1}}\!\!+\!\!C_{6}(h)y^{q_{2}}\!\!+\!\!\frac{1}{r\beta_{2}}y\ln(y)\!\!+\!\!\frac{\gamma-2r+k}{r^{2}\beta_{2}}y\!\!-\!\!\frac{\alpha h}{r}y\!\!+\!\!\frac{1}{\gamma\beta_{2}},\ e^{-(1-\alpha)\beta_{2}h}\le y<1,\\
		   C_{7}(h)y^{q_{1}}\!\!+\!\!C_{8}(h)y^{q_{2}}\!\!-\!\!\frac{1}{r}hy\!\!+
\!\!\frac{1}{\gamma\beta_{2}}\big[1\!-\!e^{-(1\!-\!\alpha)\beta_{2}h}\big],\ (1\!-\!\alpha)e^{-(1\!-\!\alpha)\beta_{2}h}\!\le\! y\!<\!e^{-(1-\alpha)\beta_{2}h}.
		    \end{array}
		    \right.
\end{equation*}
We can obtain $C_{i}(h),1\le i\le 8$ in the same way as in Section \ref{dersol}.

\begin{align*}
    C_{2}(h)=&0, \ \ \   C_{4}(h)=\frac{1}{(q_{2}-q_{1})\beta_{1}}\Big[\frac{q_{1}}{\gamma}-\frac{1}{r}+\frac{\gamma-2r+k}{r^{2}}(q_{1}-1)\Big]e^{-(\alpha-\lambda)(q_{2}-1)\beta_{1}h},\\
	    C_{6}(h)=&C_{4}(h)+\frac{\beta_{2}-\beta_{1}}{(q_{2}-q_{1})\beta_{1}\beta_{2}}\Big[-\frac{q_{1}}{\gamma}+\frac{1}{r}-\frac{\gamma-2r+k}{r^{2}}(q_{1}-1)\Big],\\
	    C_{8}(h)=&C_{6}(h)+\frac{1}{(q_{2}-q_{1})\beta_{2}}\Big[-\frac{q_{1}}{\gamma}+\frac{1}{r}-\frac{\gamma-2r+k}{r^{2}}(q_{1}-1)\Big]e^{(1-\alpha)(q_{2}-1)\beta_{2}h},\\
	    C_{7}(h)=&\frac{1}{q_{2}-q_{1}}\frac{(1-\alpha)^{q_{2}-q_{1}}\big[-\frac{q_{1}}{\gamma}+\frac{1}{r}-\frac{\gamma-2r+k}{r^{2}}(q_{1}-1)\big]
(\alpha-\lambda)(q_{2}-1)}{(1-\alpha)(q_{2}-q_{1})\beta_{2}
+(\alpha-\lambda)(q_{2}-1)\beta_{1}}e^{-\big[(1-\alpha)(q_{2}-q_{1})\beta_{2}
+(\alpha-\lambda)(q_{2}-1)\beta_{1}\big]h}\notag\\
	    &+(1-\alpha)^{q_{2}-q_{1}}\frac{-\frac{q_{1}}{\gamma}+\frac{1}{r}-\frac{\gamma-2r+k}{r^{2}}(q_{1}-1)}{(1-q_{1})\beta_{2}}\frac{q_{2}-1}{q_{2}-q_{1}}
e^{-(1-\alpha)(1-q_{1})\beta_{2}h},\\
	    C_{5}(h)=&C_{7}(h)+\frac{1}{(q_{2}-q_{1})\beta_{2}}\Big[-\frac{q_{2}}{\gamma}+\frac{1}{r}-\frac{\gamma-2r+k}{r^{2}}(q_{2}-1)\Big]e^{-(1-\alpha)(1-q_{1})\beta_{2}h},\\
	    C_{3}(h)=&C_{5}(h)+\frac{\beta_{2}-\beta_{1}}{(q_{2}-q_{1})\beta_{1}\beta_{2}}\Big[-\frac{q_{2}}{\gamma}+\frac{1}{r}-\frac{\gamma-2r+k}{r^{2}}(q_{2}-1)\Big],\\
	    C_{1}(h)=&C_{3}(h)+\frac{1}{(q_{2}-q_{1})\beta_{1}}\Big[\frac{q_{2}}{\gamma}-\frac{1}{r}+\frac{\gamma-2r+k}{r^{2}}(q_{2}-1)\Big]e^{(\alpha-\lambda)(1-q_{1})\beta_{1}h}.
\end{align*}

It can verified that $q_{1}<\frac{\frac{\gamma-2r+k}{r^{2}}+\frac{1}{r}}{\frac{\gamma-2r+k}{r^{2}}+\frac{1}{\gamma}}<q_{2}$. In fact, an equivalent condition is $k(\frac{\frac{\gamma-2r+k}{r^{2}}+\frac{1}{r}}{\frac{\gamma-2r+k}{r^{2}}+\frac{1}{\gamma}})^2-(k+r-\gamma)\frac{\frac{\gamma-2r+k}{r^{2}}+\frac{1}{r}}{\frac{\gamma-2r+k}{r^{2}}+\frac{1}{\gamma}}-\gamma<0$, which is equivalent(by direct computation) to the trivial inequality $-\frac{k^{2}}{\gamma r^{2}(\frac{\gamma-2r+k}{r^{2}}+\frac{1}{\gamma})^{2}}<0$. Then we obtain from $q_{1}<\frac{\frac{\gamma-2r+k}{r^{2}}+\frac{1}{r}}{\frac{\gamma-2r+k}{r^{2}}+\frac{1}{\gamma}}<q_{2}$ that

\begin{align*}
    &-\frac{q_{1}}{\gamma}+\frac{1}{r}-\frac{\gamma-2r+k}{r^{2}}(q_{1}-1)>0,\\
    &-\frac{q_{2}}{\gamma}+\frac{1}{r}-\frac{\gamma-2r+k}{r^{2}}(q_{2}-1)<0.\\
\end{align*}
 
The difference of $C_{i}(h),1\le i\le 8$ between $r=\gamma$ and $r\neq\gamma$ is that certain positive coefficients such as $\frac{k}{\gamma^{2}}\frac{q_{2}-1}{q_{2}-q_{1}}$ in case $r=\gamma$ are replaced by other positive coefficients such as $\frac{1}{q_{2}-q_{1}}\Big[\frac{q_{2}}{\gamma}-\frac{1}{r}+\frac{\gamma-2r+k}{r^{2}}(q_{2}-1)\Big]$ in case $r\neq\gamma$. As a result, the estimates of $C_{i}(h),1\le i\le 8$ is completely the same as in Section \ref{dersol} and the main results are similar. The optimal strategy is as follows.

\begin{theorem}
	For $(x_{0},h_{0})\in\mathcal{C}$, where $\mathcal{C}$ is the effective region given by (\ref{effective}), let $c^{*}_{\rm primal}(\cdot,\cdot)$ and $ \pi^{*}_{\rm primal}(\cdot,\cdot)$ be the feedback functions in terms of primal variable given respectively by
	
		\begin{align*}
		    &c_{\rm primal}^{*}(x,h)=\left\{
	\begin{array}{l}
	\lambda h,\hskip 3.3cm \frac{\lambda h}{\gamma}\le x\le W_{\rm low}(h),\\
	-\frac{1}{\beta_{1}}\ln\big[f_{2}(x,h)\big]+\alpha h,\ W_{\rm low}(h)<x\le W_{\rm ref}(h),\\
	-\frac{1}{\beta_{2}}\ln\big[f_{3}(x,h)\big]+\alpha h,\ W_{\rm ref}(h)<x\le W_{\rm peak}(h),\\
	h,\hskip 3.3cm W_{\rm peak}(h)<x\le W_{\rm updt}(h),
	\end{array}
	\right.\\
	&\pi_{\rm primal}^{*}(x,h)=\frac{\mu-r}{\sigma^{2}}\!\!\left\{
	\begin{array}{l}
	q_{1}(q_{1}\!\!-\!\!1)C_{1}(h)\big[f_{1}(x,h)\big]^{q_{1}-1}\!\!+\!\!q_{2}(q_{2}\!\!-\!\!1)C_{2}(h)\big[f_{1}(x,h)\big]^{q_{2}-1}\\,\ \frac{\lambda h}{\gamma}\le x\le W_{\rm low}(h),\\
	q_{1}(q_{1}\!\!-\!\!1)C_{3}(h)\big[f_{2}(x,h)\big]^{q_{1}-1}\!\!+\!\!q_{2}(q_{2}\!\!-\!\!1)C_{4}(h)\big[f_{2}(x,h)\big]^{q_{2}-1}\!\!+\!\!\frac{1}{r\beta_{1}}\\,\ W_{\rm low}(h)<x\le W_{\rm ref}(h),\\
	q_{1}(q_{1}\!\!-\!\!1)C_{5}(h)\big[f_{3}(x,h)\big]^{q_{1}-1}\!\!+\!\!q_{2}(q_{2}\!\!-\!\!1)C_{6}(h)\big[f_{3}(x,h)\big]^{q_{2}-1}\!\!+\!\!\frac{1}{r\beta_{2}}\\,\ W_{\rm ref}(h)<x\le W_{\rm peak}(h),\\
	q_{1}(q_{1}\!\!-\!\!1)C_{7}(h)\big[f_{4}(x,h)\big]^{q_{1}-1}\!\!+\!\!q_{2}(q_{2}\!\!-\!\!1)C_{8}(h)\big[f_{4}(x,h)\big]^{q_{2}-1}\\,\ W_{\rm peak}(h)<x\le W_{\rm updt}(h),
	\end{array}
	\right.
	\end{align*}
	where $f_{i}(x,h),\ 1\le i\le 4$ are uniquely determined by 
	
	\begin{align*}
	    x=&-C_{1}(h)q_{1}\big[f_{1}(x,h)\big]^{q_{1}-1}-C_{2}(h)q_{2}\big[f_{1}(x,h)\big]^{q_{2}-1}+\frac{\lambda h}{r},\\
	    x=&-C_{3}(h)q_{1}\big[f_{2}(x,h)\big]^{q_{1}-1}-C_{4}(h)q_{2}\big[f_{2}(x,h)\big]^{q_{2}-1}-\frac{1}{r\beta_{1}}\ln[f_{2}(x,h)]-\frac{\gamma-r+k}{r^{2}\beta_{1}}+\frac{\alpha h}{r},\\
	    x=&-C_{5}(h)q_{1}\big[f_{3}(x,h)\big]^{q_{1}-1}-C_{6}(h)q_{2}\big[f_{3}(x,h)\big]^{q_{2}-1}-\frac{1}{r\beta_{2}}\ln[f_{3}(x,h)]-\frac{\gamma-r+k}{r^{2}\beta_{2}}+\frac{\alpha h}{r},\\
	    x=&-C_{7}(h)q_{1}\big[f_{4}(x,h)\big]^{q_{1}-1}-C_{8}(h)q_{2}\big[f_{4}(x,h)\big]^{q_{2}-1}+\frac{h}{r},
	    \end{align*}
	    and $W_{\rm low}(h),W_{\rm ref}(h),W_{\rm peak}(h)$ and $W_{\rm updt}(h)$ are given by
	    
	    \begin{align*}
	    W_{\rm low}(h)=&-C_{1}(h)q_{1}e^{-(\alpha-\lambda)(1-q_{1})\beta_{1}h}-C_{2}(h)q_{2}e^{(\alpha-\lambda)(q_{2}-1)\beta_{1}h}+\frac{\lambda h}{r},\\
	    W_{\rm ref}(h)=&-C_{3}(h)q_{1}-C_{4}(h)q_{2}-\frac{\gamma-r+k}{r^{2}\beta_{1}}+\frac{\alpha h}{r},
	    \end{align*}
	    \begin{align*}
	    W_{\rm peak}(h)=&-C_{5}(h)q_{1}e^{(1-\alpha)(1-q_{1})\beta_{2}h}-C_{6}(h)q_{2}e^{-(1-\alpha)(q_{2}-1)\beta_{2}h}-\frac{\gamma-r+k}{r^{2}\beta_{1}}+\frac{ h}{r},\\
	    W_{\rm updt}(h)=&-C_{7}(h)q_{1}(1-\alpha)^{q_{1}-1}e^{(1-\alpha)(1-q_{1})\beta_{2}h}-C_{8}(h)q_{2}(1-\alpha)^{q_{2}-1}e^{-(1-\alpha)(q_{2}-1)\beta_{2}h}+\frac{h}{\gamma}.
	\end{align*}
		Then SDE
		\begin{equation*}
		\left\{
		\begin{array}{l}
		dX_{t}=rX_{t}dt+\pi^{*}_{\rm primal}(X_{t},H^{*}_{t})(\mu-r)dt+\pi^{*}_{\rm primal}(X_{t},H^{*}_{t})\sigma dW_{t}-c^{*}_{\rm primal}(X_{t},H^{*}_{t})dt,\\
		X_{0}=x_{0}
		\end{array}
		\right.
		\end{equation*}
		with $H^{*}_{t}\triangleq h_{0}\vee \sup\limits_{s\le t}c^{*}_{\rm primal}(X_{s},H^{*}_{s})$ and $H^{*}_{0}=h_{0}$,  has a unique strong solution $\{X^{*}_{t},\ t\ge 0\}$. The optimal consumption and investment policy is
		
		\begin{equation*}
		\Big\{\big(c^{*}_{\rm primal}(X^{*}_{t},H^{*}_{t}),\pi^{*}_{\rm primal}(X^{*}_{t},H^{*}_{t})\big),\ t\ge 0\Big\}.
		\end{equation*}
\end{theorem}

Here we just give the proof of Lemma \ref{vyypo} with $r\neq\gamma$. The proofs of other results with $r\neq\gamma$ are very similar to those with $r=\gamma$.

\begin{proof}[Proof of Lemma \ref{vyypo} with $r\neq\gamma$]

If $e^{(\alpha-\lambda)\beta_{1}h}\le y$, then

	    \begin{equation*}
	        \tilde{V}_{yy}(y,h)=C_{1}(h)q_{1}(q_{1}-1)y^{q_{1}-2}+C_{2}(h)q_{2}(q_{2}-1)y^{q_{2}-2}.
	    \end{equation*}
As $C_{1}(h)>0$ and $C_{2}(h)=0$, we have $\tilde{V}_{yy}(y,h)>0$ for $e^{(\alpha-\lambda)\beta_{1}h}\le y$.
	
If $1\le y<e^{(\alpha-\lambda)\beta_{1}h}$, then

	    \begin{equation*}
	        y\tilde{V}_{yy}(y,h)=C_{3}(h)q_{1}(q_{1}-1)y^{q_{1}-1}+C_{4}(h)q_{2}(q_{2}-1)y^{q_{2}-1}+\frac{1}{r\beta_{1}}.
	    \end{equation*}
	    Let $\psi(y)=y\tilde{V}_{yy}(y,h)$, then $$\psi'(y)=C_{3}(h)q_{1}(q_{1}-1)^{2}y^{q_{1}-2}+C_{4}(h)q_{2}(q_{2}-1)^{2}y^{q_{2}-2}.$$ Noting that $C_{4}(h)<0$,  $\psi(y)$ is either increasing, decreasing or first increasing then decreasing, we only need to show $\psi(1)>0$ and $\psi(e^{(\alpha-\lambda)\beta_{1}h})>0$. Precisely, using $C_{7}(h)>0$ and the fact that
	    
	    \begin{equation*}
     \frac{1}{r}=\frac{q_{1}(q_{1}-1)}{q_{2}-q_{1}}\Big[\frac{q_{2}}{\gamma}-\frac{1}{r}+\frac{\gamma-2r+k}{r^{2}}(q_{2}-1)\Big]+\frac{q_{2}(q_{2}-1)}{q_{2}-q_{1}}\Big[-\frac{q_{1}}{\gamma}+\frac{1}{r}-\frac{\gamma-2r+k}{r^{2}}(q_{1}-1)\Big],
 \end{equation*}
	     we have
	     
\begin{align*}
	        \psi(1)=&C_{3}(h)q_{1}(q_{1}-1)+C_{4}(h)q_2(q_{2}-1)+\frac{1}{r\beta_{1}}\\
	        =&q_{1}(q_{1}-1)C_{7}(h)+\frac{q_{1}(q_{1}-1)}{(q_{2}-q_{1})\beta_{2}}\Big[\frac{q_{2}}{\gamma}-\frac{1}{r}+\frac{\gamma-2r+k}{r^{2}}(q_{2}-1)\Big]\left[1-e^{-(1-\alpha)(1-q_{1})\beta_{2}h}\right]\\
&+\frac{q_{2}(q_{2}-1)}{(q_{2}-q_{1})\beta_{1}}\Big[-\frac{q_{1}}{\gamma}+\frac{1}{r}-\frac{\gamma-2r+k}{r^{2}}(q_{1}-1)\Big]\left[1-e^{-(\alpha-\lambda)(q_{2}-1)\beta_{1}h}\right ]\\
	        >&0
\end{align*}
and

\begin{align*}	        \psi(e^{(\alpha-\lambda)\beta_{1}h})=&C_{3}(h)q_{1}(q_{1}-1)e^{(q_{1}-1)
(\alpha-\lambda)\beta_{1}h}+C_{4}(h)q_{2}(q_{2}-1)e^{(q_{2}-1)(\alpha-\lambda)
\beta_{1}h}+\frac{1}{r\beta_{1}}\\
=&q_{1}(q_{1}-1)C_{7}(h)e^{(q_{1}-1)(\alpha-\lambda)\beta_{1}h}\\&+\frac{q_{1}(q_{1}-1)}{(q_{2}-q_{1})\beta_{2}}\Big[\frac{q_{2}}{\gamma}-\frac{1}{r}+\frac{\gamma-2r+k}{r^{2}}(q_{2}-1)\Big]\left[1-e^{-(1-\alpha)(1-q_{1})\beta_{2}h}\right]e^{(q_{1}-1)
(\alpha-\lambda)\beta_{1}h}
\\
&+\frac{q_{1}(q_{1}-1)}{(q_{2}-q_{1})\beta_{2}}\Big[\frac{q_{2}}{\gamma}-\frac{1}{r}+\frac{\gamma-2r+k}{r^{2}}(q_{2}-1)\Big]
\left[1-e^{-(\alpha-\lambda)(1-q_{1})\beta_{1}h}\right]\\
>&0.
\end{align*}
Thus $\tilde{V}_{yy}(y,h)>0$ for $1\le y<e^{(\alpha-\lambda)\beta_{1}h}$.
\vskip 5pt	
If $e^{-(1-\alpha)\beta_{2}h}\le y<1$, then

\begin{align*}
	        y\tilde{V}_{yy}(y,h)=&C_{5}(h)q_{1}(q_{1}-1)y^{q_{1}-1}+C_{6}(h)q_{2}(q_{2}-1)y^{q_{2}-1}+\frac{1}{r\beta_{2}}\\
	       =&q_{1}(q_{1}-1)\left\{C_{5}(h)y^{q_{1}-1}+\frac{1}{(q_{2}-q_{1})\beta_{2}}\Big[\frac{q_{2}}{\gamma}-\frac{1}{r}+\frac{\gamma-2r+k}{r^{2}}(q_{2}-1)\Big]\right\}
\\&+q_{2}(q_{2}-1)\left\{C_{6}(h)y^{q_{2}-1}+\frac{1}{(q_{2}-q_{1})\beta_{2}}\Big[-\frac{q_{1}}{\gamma}+\frac{1}{r}-\frac{\gamma-2r+k}{r^{2}}(q_{1}-1)\Big]\right\}.
\end{align*}
For any fixed $h>0$, if $C_{5}(h)\ge 0$, then $C_{5}(h)y^{q_{1}-1}+\frac{1}{(q_{2}-q_{1})\beta_{2}}\Big[\frac{q_{2}}{\gamma}-\frac{1}{r}+\frac{\gamma-2r+k}{r^{2}}(q_{2}-1)\Big]>0$; If $C_{5}(h)<0$, then $C_{5}(h)y^{q_{1}-1}+\frac{1}{(q_{2}-q_{1})\beta_{2}}\Big[\frac{q_{2}}{\gamma}-\frac{1}{r}+\frac{\gamma-2r+k}{r^{2}}(q_{2}-1)\Big]$ is increasing in $y$ and

\begin{align*}
	        &C_{5}(h)y^{q_{1}-1}+\frac{1}{(q_{2}-q_{1})\beta_{2}}\Big[\frac{q_{2}}{\gamma}-\frac{1}{r}+\frac{\gamma-2r+k}{r^{2}}(q_{2}-1)\Big]\\\ge& C_{5}(h)e^{(1-\alpha)(1-q_{1})\beta_{2}h}+\frac{1}{(q_{2}-q_{1})\beta_{2}}\Big[\frac{q_{2}}{\gamma}-\frac{1}{r}+\frac{\gamma-2r+k}{r^{2}}(q_{2}-1)\Big]\\
	        =&C_{7}(h)e^{(1-\alpha)(1-q_{1})\beta_{2}h}\\
	        >&0.
	    \end{align*}
Similarly,  we have

	    \begin{equation*}
	        C_{6}(h)y^{q_{2}-1}+\frac{1}{(q_{2}-q_{1})\beta_{2}}\Big[-\frac{q_{1}}{\gamma}+\frac{1}{r}-\frac{\gamma-2r+k}{r^{2}}(q_{1}-1)\Big]>0.
	    \end{equation*}
It follows that $\tilde{V}_{yy}(y,h)>0$ for $e^{-(1-\alpha)\beta_{2}h}\le y<1$.
\vskip 5pt
Finally, If $(1-\alpha)e^{-(1-\alpha)\beta_{2}h}\le y<e^{-(1-\alpha)\beta_{2}h}$, then

	    \begin{equation*}
	        \tilde{V}_{yy}(y,h)=C_{7}(h)q_{1}(q_{1}-1)y^{q_{1}-2}+C_{8}(h)q_{2}(q_{2}-1)y^{q_{2}-2}.
	    \end{equation*}
As $C_{7}(h)>0$ and $C_{8}(h)>0$, we have  $\tilde{V}_{yy}(y,h)>0$ for $(1-\alpha)e^{-(1-\alpha)\beta_{2}h}\le y<e^{-(1-\alpha)\beta_{2}h}$. Thus, the proof is complete.
\end{proof}

\bibliographystyle{plainnat}
\bibliography{reference.bib}

 \end{document}